\documentclass{lmcs}
\pdfoutput=1

\usepackage{lastpage}
\lmcsdoi{20}{3}{19}
\lmcsheading{}{\pageref{LastPage}}{}{}%
{Mar.~22,~2023}{Aug.~29,~2024}{}

\usepackage[utf8]{inputenc}


\usepackage{tikz}

\usepackage{amsfonts,amssymb}
\usepackage{mathtools}
\usepackage{tikz}

\newcommand{\qedEx}{\hfill\ensuremath{\lhd}}

\newcommand{\IN}{\mathbb{N}}

\renewcommand{\epsilon}{\varepsilon}

\newcommand{\emptyPth}{\ensuremath{\hist}}
\newcommand{\hist}{\ensuremath{\gamma}}
\newcommand{\play}{\ensuremath{\pi}}
\newcommand{\clr}{\ensuremath{c}}
\newcommand{\colors}{\ensuremath{C}}
\newcommand{\goodColors}{\ensuremath{F}}
\newcommand{\s}{\ensuremath{v}}
\newcommand{\states}{\ensuremath{V}}
\newcommand{\edge}{\ensuremath{e}}
\newcommand{\edges}{\ensuremath{E}}
\newcommand{\edgeOut}{\ensuremath{\mathsf{last}}}
\newcommand{\arena}{\ensuremath{\mathcal{A}}}
\newcommand{\arenaFull}{\ensuremath{(\states, \states_1, \states_2, \edges)}}
\newcommand*{\inverse}[1]{#1^{-1}}
\newcommand{\gameFull}{\ensuremath{(\arena, \wc)}}
\newcommand{\strat}{\ensuremath{\sigma}}
\newcommand*{\player}[1]{\ensuremath{\mathcal{P}_{#1}}}
\newcommand{\Pone}{\ensuremath{\player{1}}}
\newcommand{\Ptwo}{\ensuremath{\player{2}}}
\newcommand{\Hists}{\ensuremath{\mathsf{Hists}}}

\newcommand{\alphabet}{\ensuremath{\colors}}
\newcommand{\dba}{\ensuremath{\mathcal{B}}}
\newcommand{\nba}{\ensuremath{\mathcal{B}}}
\newcommand{\struct}{\ensuremath{\mathcal{S}}}
\newcommand{\atmtnStates}{\ensuremath{Q}}
\newcommand{\atmtnState}{\ensuremath{q}}
\newcommand{\atmtnInit}{\ensuremath{\atmtnState_{\mathsf{init}}}}
\newcommand{\atmtnInits}{\ensuremath{\atmtnStates_{\mathsf{init}}}}
\newcommand{\atmtnTrans}{\ensuremath{\Delta}}
\newcommand{\atmtnTransSafe}{\ensuremath{\Delta}_{\accepSet{\text -}\mathsf{free}}}
\newcommand{\atmtnUpd}{\ensuremath{\delta}}
\newcommand{\atmtnUpdWord}{\ensuremath{\atmtnUpd^*}}
\newcommand{\atmtnLang}[1]{\ensuremath{\mathcal{L}(#1)}}
\newcommand{\accepSet}{\ensuremath{\alpha}}
\newcommand{\accepSetNBA}{\ensuremath{\widehat{\accepSet}}}
\newcommand{\dbaFull}{\ensuremath{(\atmtnStates, \alphabet, \atmtnInit, \atmtnUpd, \accepSet)}}
\newcommand{\nbaFull}{\ensuremath{(\atmtnStates, \alphabet, \atmtnInits, \atmtnTrans, \accepSetNBA)}}
\newcommand{\nbaStruct}{\ensuremath{(\atmtnStates, \alphabet, \atmtnInits, \atmtnTrans)}}
\newcommand{\dbaStruct}{\ensuremath{(\atmtnStates, \alphabet, \atmtnInit, \atmtnUpd)}}

\newcommand{\dbaSat}{\ensuremath{\dba_{\mathsf{sat}}}}
\newcommand{\dbaSatFull}{\ensuremath{(\atmtnStates, \alphabet, \atmtnInit, \atmtnUpd, \accepSetSat)}}
\newcommand{\accepSetSat}{\ensuremath{\accepSet_{\mathsf{sat}}}}
\newcommand{\run}{\ensuremath{\varrho}}
\newcommand{\word}{\ensuremath{w}}
\newcommand*{\dbaRun}[1]{\dba(#1)}

\newcommand{\minStateAtmtn}{\ensuremath{\struct_\prefEq}}
\newcommand{\prefClass}{\ensuremath{\minStateAtmtn}}
\newcommand{\minStateStates}{\ensuremath{\atmtnStates_\prefEq}}
\newcommand{\minStateInit}{\ensuremath{\widetilde{\atmtnState}_{\mathsf{init}}}}
\newcommand{\minStateUpd}{\ensuremath{\atmtnUpd_\prefEq}}
\newcommand{\minStateAtmtnFull}{\ensuremath{(\minStateStates, \alphabet, \minStateInit, \minStateUpd)}}

\newcommand{\atmtnStateMax}{\ensuremath{\atmtnState_\mathsf{max}}}

\newcommand{\alphProof}{\ensuremath{\colors}}
\newcommand{\qEquiv}{\ensuremath{\atmtnState_\prefEq}}
\newcommand{\alphProofGood}{\ensuremath{\goodColors_{\eqClass{\qEquiv}}}}

\newcommand{\mainOrd}{\ensuremath{\theta}}
\newcommand{\ord}{\ensuremath{\lambda}}
\newcommand{\cardinal}{\ensuremath{\kappa}}
\newcommand{\ordBis}{\ensuremath{\eta}}
\newcommand{\grSt}{\ensuremath{v}}
\newcommand{\grStates}{\ensuremath{V}}
\newcommand{\grEdges}{\ensuremath{\edges}}
\newcommand*{\grTr}[1]{\ensuremath{\xrightarrow{#1}}}
\newcommand{\univGrStates}{\ensuremath{U}}
\newcommand{\univGr}{\ensuremath{\mathcal{U}}}
\newcommand{\univGrB}{\ensuremath{\univGr_{\dba, \mainOrd}}}
\newcommand{\univGrStatesB}{\ensuremath{\univGrStates_{\dba, \mainOrd}}}
\newcommand{\gr}{\ensuremath{\mathcal{G}}}
\newcommand{\grFull}{\ensuremath{(\grStates, \grEdges)}}
\newcommand{\grPth}{\ensuremath{\pi}}
\newcommand{\finGrPth}{\ensuremath{\hist}}
\newcommand{\morph}{\ensuremath{\phi}}

\newcommand*{\Buchi}[1]{\ensuremath{\mathsf{B\ddot{u}chi}(#1)}}

\newcommand{\wc}{\ensuremath{W}} 
\newcommand*{\comp}[1]{\overline{#1}} 

\newcommand*{\card}[1]{\ensuremath{\lvert#1\rvert}}

\newcommand{\emptyWord}{\ensuremath{\varepsilon}}

\newcommand{\prefEq}{\ensuremath{\sim}}
\newcommand*{\eqClass}[1]{\ensuremath{[#1]}}

\newcommand{\prefOrd}{\ensuremath{\preceq}}
\newcommand{\invPrefOrd}{\ensuremath{\succeq}}
\newcommand{\strictPrefOrd}{\ensuremath{\prec}}
\newcommand{\strictInvPrefOrd}{\ensuremath{\succ}}
\newcommand*{\quotient}[2]{{\raisebox{.2em}{$#1$}\!\left/\raisebox{-.2em}{$#2$}\right.}}

\newcommand*{\safe}[1]{\ensuremath{\accepSet{\text -}\mathsf{Free}_{\dba}(#1)}}
\newcommand*{\safep}[1]{\ensuremath{\accepSet{\text -}\mathsf{Free}_{\dba'}(#1)}}
\newcommand*{\safeCycles}[1]{\ensuremath{\accepSet{\text -}\mathsf{FreeCycles}_{\dba}(#1)}}
\newcommand*{\safeCyclesp}[1]{\ensuremath{\accepSet{\text -}\mathsf{FreeCycles}_{\dba'}(#1)}}
\newcommand*{\safeCyclesSat}[1]{\ensuremath{\accepSet{\text -}\mathsf{FreeCycles}_{\dbaSat}(#1)}}

\newcommand{\bigO}{\ensuremath{\mathcal{O}}}

\tikzstyle{rond}=[draw,circle,minimum height=7mm]
\tikzstyle{oval}=[draw,ellipse,minimum height=7mm]
\tikzstyle{diamant}=[draw,diamond,minimum height=9mm,minimum width=9mm,aspect=1]
\tikzstyle{carre}=[draw,minimum width=6mm,minimum height=6mm]

\newcommand{\ReachAorAA}{
	\begin{tikzpicture}[every node/.style={font=\small,inner sep=1pt}]
		\draw (0,0) node[diamant] (qinit) {$\atmtnInit$};
		\draw ($(qinit.north)+(0,0.4)$) edge[-latex'] (qinit);
		\draw ($(qinit)+(1.5,0)$) node[diamant] (qa) {$\atmtnState_a$};
		\draw ($(qinit)+(3,0)$) node[diamant] (top) {$\atmtnState_{aa}$};
		\draw (qinit) edge[-latex',out=30,in=180-30,accepting] node[above=4pt] {$a$} (qa);
		\draw (qa) edge[-latex',out=180+30,in=-30,accepting] node[below=4pt] {$b$} (qinit);
		\draw (qa) edge[-latex',accepting] node[above=4pt] {$a$} (top);
		\draw (qinit) edge[-latex',out=150,in=210,distance=0.8cm] node[left=4pt] {$b$} (qinit);
		\draw (top) edge[-latex',out=-30,in=30,distance=0.8cm,accepting] node[right=4pt] {$a, b$} (top);
\end{tikzpicture}}

\usetikzlibrary{arrows,automata,calc,decorations.pathmorphing,decorations.pathreplacing,shapes.geometric}
\tikzset{decoration={snake,amplitude=.5mm,segment length=3mm,post length=0.8mm,pre length=0mm}}
\tikzstyle{accepting}=[edge node={node[scale=1.5] {$\bullet$}}]

\theoremstyle{definition}\newtheorem{condition}[thm]{Condition}

\begin{document}
	\title[Half-Positional Objectives Recognized by DBA]{Half-Positional Objectives\texorpdfstring{\\}{ }Recognized by Deterministic B\"uchi Automata\rsuper*}
	\titlecomment{{\lsuper*}Extended version of CONCUR 2022 article~\cite{BCRV22Conf} with the same name.}
	\thanks{This work has been partially supported by the ANR Project MAVeriQ (ANR-20-CE25-0012).
		Mickael Randour is an F.R.S.-FNRS Research Associate and a member of the TRAIL Institute.
		Pierre Vandenhove is an F.R.S.-FNRS Research Fellow.
	}

	\author{Patricia Bouyer\lmcsorcid{0000-0002-2823-0911}}[a]
	\author{Antonio Casares\lmcsorcid{0000-0002-6539-2020}}[b]
	\author{Mickael Randour\lmcsorcid{0000-0001-8777-2385}}[c]
	\author{Pierre Vandenhove\lmcsorcid{0000-0001-5834-1068}}[a,c]

	\address{Université Paris-Saclay, CNRS, ENS Paris-Saclay, Laboratoire Méthodes Formelles, 91190, Gif-sur-Yvette, France}

	\address{LaBRI, Université de Bordeaux, Bordeaux, France}

	\address{F.R.S.-FNRS \& UMONS -- Université de Mons, Mons, Belgium}

	\begin{abstract}
		In two-player games on graphs, the simplest possible strategies are those that can be implemented without any memory.
		These are called \emph{positional strategies}.
		In this paper, we characterize objectives recognizable by deterministic B\"uchi automata (a subclass of \emph{$\omega$-regular objectives}) that are \emph{half-positional}, that is, for which the protagonist can always play optimally using positional strategies (both over finite and infinite graphs).
		Our characterization consists of three natural conditions linked to the language-theoretic notion of \emph{right congruence}.
		Furthermore, this characterization yields a polynomial-time algorithm to decide half-positionality of an objective recognized by a given deterministic B\"uchi automaton.
	\end{abstract}

	\maketitle

	\section{Introduction} \label{sec:intro}

	\paragraph{Graph games and reactive synthesis.}
	We study \emph{zero-sum turn-based games on graphs} in which two players (a protagonist and its opponent) confront each other.
	They interact by moving a pebble in turns through the edges of a graph for an infinite amount of time.
	Each vertex belongs to a player, and the player controlling the current vertex decides on the next state of the game.
	Edges of the graph are labeled with \emph{colors}, and the interaction of the two players therefore produces an infinite sequence of them.
	The objective of the game is specified by a subset of infinite sequences of colors, and the protagonist wins if the produced sequence belongs to this set.
	We are interested in finding a \emph{winning strategy} for the protagonist, that is, a function indicating how the protagonist should move in any situation, guaranteeing the achievement of the objective.

	This game-theoretic model is particularly fitted to study the \emph{reactive synthesis problem}~\cite{BCJ18}: a system (the protagonist) wants to satisfy a specification (the objective) while interacting continuously with its environment (the opponent).
	The goal is to build a controller for the system satisfying the specification, whenever possible.
	This comes down to finding a winning strategy for the protagonist in the derived game.

	\paragraph{Half-positionality.}
	In order to obtain a controller for the system that is simple to implement, we are interested in finding the simplest possible winning strategy.
	Here, we focus on the amount of information that winning strategies have to remember.
	The simplest strategies are then arguably \emph{positional} (also called \emph{memoryless}) strategies, which do not remember anything about the past and base their decisions solely on the current state of the game.
	We intend to understand for which objectives positional strategies suffice for the protagonist to play optimally (i.e., to win whenever it is possible) --- we call these objectives \emph{half-positional}.
	We distinguish half-positionality from \emph{bipositionality} (sometimes called \emph{memoryless-determinacy} or just \emph{positionality}), which refers to objectives for which positional strategies suffice to play optimally for \emph{both} players.

	Many natural objectives have been shown to be bipositional over games on finite and sometimes infinite graphs: e.g., discounted sum~\cite{Sha53}, mean-payoff~\cite{EM79}, parity~\cite{EJ91}, total payoff~\cite{GZ04}, energy~\cite{BFLMS08}, or average-energy games~\cite{BMRLL18}.
	Bipositionality can be established using general criteria and characterizations, over games on both finite~\cite{GZ04,GZ05,AR17} and infinite~\cite{CN06} graphs.
	Yet, there exist many objectives and combinations thereof for which one player, but not both, has positional optimal strategies (Rabin conditions~\cite{KK91RabinMeasures,Kla94}, mean-payoff parity~\cite{CHJ05MPPar}, energy parity~\cite{CD12EP}, some window objectives~\cite{CDRR15,BHR16}, energy mean-payoff~\cite{BHRR19}\ldots), and to which these results do not apply.

	Various attempts have been made to understand common underlying properties of half-positional objectives and provide sufficient conditions~\cite{Kop06,Kop07,KopThesis,BFMM11}, but little more was known until the recent work of Ohlmann~\cite{Ohl23} (discussed below).
	These sufficient conditions are not general enough to prove half-positionality of some very simple objectives, even in the well-studied class of \emph{$\omega$-regular objectives}~\cite[Lemma~13]{BFMM11}.
	Furthermore, multiple questions concerning half-positionality remain open.
	For instance, in~\cite{KopThesis}, Kopczy\'nski conjectured that \emph{prefix-independent} half-positional objectives are closed under finite union (this conjecture was recently refuted for games on finite graphs~\cite{Koz22Refutation}, but is still unsolved for games on infinite graphs).
	Also, Kopczy\'nski showed that given a deterministic parity automaton recognizing a prefix-independent objective $\wc$, we can decide if $\wc$ is half-positional over finite arenas~\cite{Kop07}.
	However, the time complexity of his algorithm is $\bigO(n^{\bigO(n^2)})$, where $n$ is the number of states of the automaton.
	It is unknown whether this can be done in polynomial time, and no algorithm is known for deciding half-positionality over infinite arenas or in the non-prefix-independent case.

	\paragraph{$\omega$-regular objectives and deterministic B\"uchi automata.}
	A central class of objectives, whose half-positionality is not yet completely understood, is the class of \emph{$\omega$-regular objectives}.
	There are multiple equivalent definitions for them: they are the objectives defined, e.g., by $\omega$-regular expressions, by non-deterministic B\"uchi automata~\cite{McN66}, and by deterministic parity automata~\cite{Mos84}.
	These objectives coincide with the class of objectives defined by monadic second-order formulas~\cite{Buchi1962decision}, and they encompass linear-time temporal logic (LTL) specifications~\cite{Pnu77}.
	Part of their interest is due to the landmark result that finite-state machines are sufficient to implement optimal strategies in $\omega$-regular games~\cite{BL69,GH82}, implying the decidability of the monadic second-order theory of natural numbers with the successor relation~\cite{Buchi1962decision} and the decidability of the synthesis problem under LTL specifications~\cite{PR89Synthesis}.

	In this paper, we focus on the subclass of $\omega$-regular objectives recognized by \emph{deterministic B\"uchi automata} (DBA), that we call \emph{DBA-recognizable}.
	DBA-recognizable objectives correspond to the $\omega$-regular objectives that can be written as a countable intersection of open objectives (for the product topology, that is, that are $G_\delta$-sets of the Borel hierarchy); or equivalently, that are the limit of a regular language of finite words~\cite{Landweber69,PP04}.
	Deciding the winner of a game with a DBA-recognizable objective is doable in polynomial time in the size of the arena and the DBA (by solving a B\"uchi game on the product of the arena and the DBA~\cite{BCJ18}).

	We now discuss two technical tools at the core of our approach: \emph{universal graphs} and \emph{right congruences}.

	\paragraph{Universal graphs.}
	One recent breakthrough in the study of half-positionality is the introduction of \emph{well-monotonic universal graphs}, combinatorial structures that can be used to provide a witness of winning strategies in games with a half-positional objective.
	Recently, Ohlmann~\cite{Ohl23} has shown that the existence of a \emph{well-monotonic universal graph} for an objective $\wc$ exactly characterizes half-positionality (under minor technical assumptions on $\wc$).
	Moreover, under these assumptions, a wide class of algorithms, called \emph{value iteration algorithms}, can be applied to solve any game with a half-positional objective~\cite{CFGO22,Ohl23}.

	Although it brings insight into the structure of half-positional objectives, showing half-positionality through the use of universal graphs is not always straightforward, and has not yet been applied in a systematic way to $\omega$-regular objectives.

	\paragraph{Right congruence.}
	Given an objective $\wc$, the \emph{right congruence $\prefEq_\wc$ of $\wc$} is an equivalence relation on finite words: two finite words $\word_1$ and $\word_2$ are equivalent for $\prefEq_\wc$ if for all infinite continuations $\word$, $\word_1\word\in\wc$ if and only if $\word_2\word\in\wc$.
	There is a natural automaton classifying the equivalence classes of the right congruence, which we refer to as the \emph{prefix-classifier}~\cite{Sta83,MS97}.

	In the case of languages of \emph{finite} words, a straightforward adaptation of the right congruence recovers the known Myhill-Nerode congruence.
	This equivalence relation characterizes the regular languages (a language is regular if and only if its congruence has finitely many equivalence classes), and the prefix-classifier is exactly the smallest deterministic finite automaton recognizing a language --- this is the celebrated Myhill-Nerode theorem~\cite{Ner58}.

	Objectives are languages of \emph{infinite} words, for which the situation is not so clear-cut.
	In particular, some $\omega$-regular objectives cannot be recognized by their prefix-classifier along with a natural acceptance condition (B\"uchi, coB\"uchi, parity, Muller\ldots)~\cite{MS97,AF21}.

	\paragraph{Contributions.}
	Our main contribution is a \emph{characterization} of half-positionality for DBA-recognizable objectives through a conjunction of three easy-to-check conditions (Theorem~\ref{thm:mainChar}).
	\begin{enumerate}[(1)]
		\item The equivalence classes of the right congruence are \emph{totally} ordered w.r.t.\ inclusion of their winning continuations.\label{cond:1}
		\item Whenever the set of winning continuations of a finite word $w_1$ is a proper subset of the set of winning continuations of a concatenation $w_1w_2$, the word $w_1(w_2)^\omega$ produced by repeating infinitely often $w_2$ is winning.\label{cond:2}
		\item The objective has to be recognizable by a DBA using the structure of its prefix-classifier.\label{cond:3}%
	\end{enumerate}%

	A few examples of simple DBA-recognizable objectives that were not encompassed by previous half-positionality criteria~\cite{Kop06,BFMM11} are, e.g., reaching a color twice~\cite[Lemma~13]{BFMM11} and weak parity~\cite{Tho08}.
	We also refer to Example~\ref{ex:AAorBuchiA}, which is half-positional but not bipositional, and whose half-positionality is straightforward using our characterization.

	Various corollaries with practical and theoretical interest follow from our characterization.
	\begin{itemize}
		\item We obtain a painless path to show (by checking each of the three conditions) that given a deterministic B\"uchi automaton, the half-positionality (over both finite and infinite arenas) of the objective it recognizes is decidable in time $\bigO(k\cdot n^4)$, where $k$ is the number of colors and $n$ is the number of states of the DBA (Section~\ref{sec:complexity}).
		\item Prefix-independent DBA-recognizable half-positional objectives are exactly the very simple \emph{B\"uchi conditions}, which consist of all the infinite words containing infinitely many occurrences of colors from some fixed subset (Proposition~\ref{prop:charPI}).
		In particular, Kopczy\'nski's conjecture trivializes for DBA-recognizable objectives (the finite union of B\"uchi conditions is a B\"uchi condition).
		\item We obtain a \emph{finite-to-infinite} and \emph{one-to-two-player} lift result (Proposition~\ref{prop:1to2}): in order to check that a DBA-recognizable objective is half-positional over arbitrary --- possibly two-player and infinite --- graphs, it suffices to check the existence of positional optimal strategies over \emph{finite} graphs where all the vertices are controlled by the protagonist.
	\end{itemize}

	\paragraph{Technical overview.}
	Condition~\ref{cond:1} turns out to be equivalent to earlier properties used to study bipositionality and half-positionality~\cite{GZ05,BFMM11} (details in Appendix~\ref{app:relation}).
	Condition~\ref{cond:2}, to the best of our knowledge, is a novel condition.
	A property similar to Condition~\ref{cond:3} has been studied multiple times in the language-theoretic literature, both for itself and for minimization and learning algorithms~\cite{Sta83,Saec90,MS97,AF21}.
	As an example, all objectives defined by deterministic \emph{weak} automata (a restriction on DBA) satisfy Condition~\ref{cond:3}~\cite{Sta83,AF21}.
	This property is sometimes called having an \emph{informative right congruence} (for some given acceptance condition).
	However, its links with positionality had never been explored.

	Conditions~\ref{cond:1} and~\ref{cond:2} are necessary for respectively bipositionality and half-positionality of general objectives.
	Condition~\ref{cond:3} is necessary for half-positionality of DBA-recognizable objectives, but not for all (even $\omega$-regular) objectives in general (see Example~\ref{ex:coBuchiAorB}).
	The proof of its necessity is more involved than for the first two conditions, and will build on automata-theoretic ideas introduced for good-for-games coB\"uchi automata~\cite{AK22Minimizing}.
	Together, the three conditions are sufficient for half-positionality of DBA-recognizable objectives: the proof of sufficiency uses the theory of universal graphs, and consists of building a family of well-monotonic universal graphs~\cite{Ohl23} for objectives satisfying the three properties.

	\paragraph{Other related works.}
	We have discussed relevant literature on half-positionality~\cite{Kop06,Kop07,BFMM11,Ohl23} and bipositionality~\cite{GZ04,GZ05,CN06,AR17}.
	A more general quest is to understand \emph{memory requirements} when positional strategies are not powerful enough: e.g.,~\cite{LPR18,BLORV22,BORV21,BRV23}.

	Memory requirements have been precisely characterized for some classes of $\omega$-regular objectives (not encompassing the class of DBA-recognizable objectives), such as Muller conditions~\cite{DJW97,Zie98,Cas22,CCL22SizeGFG} and safety specifications, i.e., objectives that are closed for the product topology~\cite{CFH14}.
	The latter also uses the order of the equivalence classes of the right congruence as part of its characterization.

	Recently, a link between the prefix-classifier, the memory requirements, and the recognizability of $\omega$-regular objectives was established~\cite{BRV23}.
	However, this result does not provide optimal bounds on the strategy complexity, and is thereby insufficient to study half-positionality.

	Our article is an extended version (with complete proofs and additional examples and remarks) of a preceding conference version~\cite{BCRV22Conf}.

	\paragraph{Structure of the paper.}
	Notations and definitions are introduced in Section~\ref{sec:preliminaries}.
	Our main contributions are presented in Section~\ref{sec:characterization}: we introduce and discuss the three conditions used in our results, then we state our main characterization (Theorem~\ref{thm:mainChar}) and some corollaries, and we end with an explanation on how to use the characterization to decide half-positionality of DBA-recognizable objectives in polynomial time.
	Section~\ref{sec:necessary} and Section~\ref{sec:sufficient} contain the proof of Theorem~\ref{thm:mainChar}: the former shows the necessity of the three conditions for half-positionality of DBA-recognizable objectives, and the latter shows their sufficiency through the use of universal graphs.

	\paragraph{Conference version.}
	Our article extends a work already published as a conference version~\cite{BCRV22Conf} with many details.
	It contains in particular the full proofs of the statements and additional explanations, remarks, and examples.

	\section{Preliminaries} \label{sec:preliminaries}
	In the whole article, letter $\colors$ refers to a (finite or infinite) non-empty set of \emph{colors}.
	Given a set $A$, we write respectively $A^*$, $A^+$, and $A^\omega$ for the set of finite, non-empty finite, and infinite sequences of elements of $A$.
	We denote by $\emptyWord$ the empty word.

	\subsection{Games and positionality}
	\paragraph{Graphs.}
	An \emph{(edge-colored) graph} $\gr=(\grStates, \grEdges)$ is given by a non-empty set of \emph{vertices} $\grStates$ (of any cardinality) and a set of \emph{edges} $\grEdges \subseteq \grStates\times \colors \times \grStates$.
	We write $\grSt\grTr{\clr} \grSt'$ if $(\grSt,\clr,\grSt')\in \grEdges$.
	We assume graphs to be \emph{non-blocking}: for all $\grSt\in\grStates$, there exists $(\grSt, \clr, \grSt')\in\edges$.
	We allow graphs with infinite branching.
	For $\grSt\in\grStates$, an \emph{infinite path of $\gr$ from $\grSt$} is an infinite sequence of edges $\grPth = (\grSt, \clr_1, \grSt_1)(\grSt_1, \clr_2, \grSt_2)(\grSt_2, \clr_3, \grSt_3)\ldots\in \edges^\omega$.
	A \emph{finite path of $\gr$ from $\grSt$} is a finite prefix in $\edges^*$ of an infinite path of $\gr$ from $\grSt$.
	For convenience, we assume that there is a distinct \emph{empty path starting from $\s$} for every $\s\in\states$.
	If $\hist = (\grSt_0, \clr_1, \grSt_1)\ldots(\grSt_{n-1}, \clr_{n}, \grSt_n)$ is a non-empty finite path of $\gr$, we define $\edgeOut(\hist) = \grSt_n$.
	For an empty path $\emptyPth_\s$ starting from $\s$, we define $\edgeOut(\emptyPth_\s) = \s$.
	An infinite (resp.\ finite) path $(\grSt_0, \clr_1, \grSt_1)(\grSt_1, \clr_2, \grSt_2)\ldots$ (resp.\ $(\grSt_0, \clr_1, \grSt_1)\ldots(\grSt_{n-1}, \clr_{n}, \grSt_n)$) is sometimes denoted $\grSt_0 \grTr{\clr_1} \grSt_1 \grTr{\clr_2} \ldots$ (resp.\ $\grSt_0 \grTr{\clr_1} \ldots \grTr{\clr_n} \grSt_n$).
	A graph $\gr=(\grStates, \grEdges)$ is \emph{finite} if both $\grStates$ and $\grEdges$ are finite.
	A graph is \emph{strongly connected} if for every pair of vertices $(\grSt,\grSt')\in \grStates\times \grStates$ there is a path from $\grSt$ to $\grSt'$.
	A \emph{strongly connected component} of $\gr$ is a maximal strongly connected subgraph.

	\paragraph{Arenas and strategies.}
	We consider two players $\Pone$ and $\Ptwo$.
	An \emph{arena} is a tuple $\arena = \arenaFull$ such that $(\states, \edges)$ is a graph and $\states$ is the disjoint union of $\states_1$ and $\states_2$.
	Intuitively, vertices in $\states_1$ are controlled by $\Pone$ and vertices in $\states_2$ are controlled by $\Ptwo$.
	An arena $\arena = \arenaFull$ is a \emph{one-player arena of $\Pone$} (resp.\ \emph{of $\Ptwo$}) if $\states_2 = \emptyset$ (resp.\ $\states_1 = \emptyset$).
	Finite paths of $(\states, \edges)$ are called \emph{histories of $\arena$}.
	For $i\in\{1, 2\}$, we denote by $\Hists_i(\arena)$ the set of histories $\hist$ of $\arena$ such that~$\edgeOut(\hist)\in\states_i$.

	Let $i\in\{1, 2\}$.
	A \emph{strategy of $\player{i}$ on $\arena$} is a function $\strat_i\colon \Hists_i(\arena) \to \edges$ such that for all $\hist\in\Hists_i(\arena)$, the first component of $\strat_i(\hist)$ coincides with $\edgeOut(\hist)$.
	Given a strategy $\strat_i$ of $\player{i}$, we say that an infinite path $\play = \edge_1\edge_2\ldots$ is \emph{consistent with $\strat_i$} if for all finite prefixes $\hist = \edge_1\ldots\edge_j$ of $\play$ such that $\edgeOut(\hist) \in \states_i$, $\strat_i(\hist) = \edge_{j+1}$.
	A strategy $\strat_i$ is \emph{positional} (also called \emph{memoryless} in the literature) if its outputs only depend on the current vertex and not on the whole history, i.e., if there exists a function $f\colon \states_i \to \edges$ such that for $\hist\in\Hists_i(\arena)$, $\strat_i(\hist) = f(\edgeOut(\hist))$.

	\paragraph{Objectives.}
	An \emph{objective} is a set $\wc\subseteq \colors^\omega$ (subsets of $\colors^\omega$ are sometimes also called \emph{languages of infinite words}, \emph{$\omega$-languages}, or \emph{winning conditions} in the literature).
	When an objective $\wc$ is clear in the context, we say that an infinite word $\word\in\colors^\omega$ is \emph{winning} if $\word\in\wc$, and \emph{losing} if $\word\notin\wc$.
	We write $\comp{\wc}$ for the complement $\colors^\omega\setminus\wc$ of an objective $\wc$.
	An objective $\wc$ is \emph{prefix-independent} if for all $\word\in\colors^*$ and $\word'\in\colors^\omega$, $\word'\in\wc$ if and only if $\word\word'\in\wc$.
	An objective that we will often consider is the \emph{B\"uchi condition}: given a subset $\goodColors\subseteq\colors$, we denote by $\Buchi{\goodColors}$ the set of infinite words containing infinitely many occurrences of colors in $\goodColors$.
	Such an objective is prefix-independent.
	A \emph{game} is a tuple $\gameFull$ of an arena $\arena$ and an objective~$\wc$.

	\paragraph{Optimality and half-positionality.}
	Let $\arena = \arenaFull$ be an arena, $\gameFull$ be a game, and $\s\in\states$.
	We say that \emph{a strategy $\strat_1$ of $\Pone$ is winning from $\s$} if for all infinite paths $\s_0 \grTr{\clr_1} \s_1 \grTr{\clr_2} \ldots$ from $\s$ consistent with $\strat_1$, $\clr_1\clr_2\ldots \in \wc$.

	A strategy of $\Pone$ is \emph{optimal for $\Pone$ in $(\arena, \wc)$} if it is winning from all the vertices from which $\Pone$ has a winning strategy.
	We often write \emph{optimal for $\Pone$ in $\arena$} if the objective $\wc$ is clear from the context.
	We stress that this notion of optimality requires a \emph{single} strategy to be winning from \emph{all} the winning vertices (a property sometimes called \emph{uniformity}).
	However, we do not require an optimal strategy to win when starting from some losing vertex, even if winning becomes possible eventually due to opponent's mistakes.

	An objective $\wc$ is \emph{half-positional} if for all arenas $\arena$, there exists a positional strategy of $\Pone$ on $\arena$ that is optimal for $\Pone$ in $\arena$.
	We sometimes only consider half-positionality on a restricted set of arenas (typically, finite and/or one-player arenas).
	For a class of arenas $\mathcal{X}$, an objective $\wc$ is \emph{half-positional over $\mathcal{X}$} if for all arenas $\arena\in\mathcal{X}$, there exists a positional strategy of $\Pone$ on $\arena$ that is optimal for $\Pone$ in $\arena$.

	\begin{rem}[$\epsilon$-edges]\label{rem:eps-transitions-definition}
		Sometimes, arenas are considered to be colored over the alphabet $\colors\cup \{\epsilon\}$, adding the restriction that no cycle is entirely labeled with $\epsilon$~\cite{DJW97,Zie98,Kop07,Horn2009RandomFruits,Cas22}.
		In that case, an infinite word in $(\colors\cup \{\epsilon\})^\omega$ labeling a path is winning if the word obtained by removing the occurrences of $\epsilon$ belongs to $\wc$.
		In this paper, we consider arenas without $\epsilon$-edges, but all our results apply to this other setting (cf.\ Remark~\ref{rem:epsilon-positional}).
		In general, allowing for arenas with $\epsilon$-edges has an effect on strategy complexity~\cite{Cas22}.
		\qedEx
	\end{rem}

	\newpage
	\subsection{B\"uchi automata}
	\paragraph{Automaton structures and B\"uchi automata.}
	A \emph{non-deterministic automaton structure (on $\colors$)} is a tuple $\struct = \nbaStruct$ such that $\atmtnStates$ is a finite set of \emph{states}, $\atmtnInits\subseteq\atmtnStates$ is a non-empty set of \emph{initial states} and $\atmtnTrans\subseteq \atmtnStates \times \colors \times \atmtnStates$ is a set of \emph{transitions}.
	We assume that all states of automaton structures are reachable from an initial state in $\atmtnInits$ by taking transitions in $\atmtnTrans$.

	A \emph{(transition-based) non-deterministic B\"uchi automaton} (NBA) is an automaton structure $\struct$ together with a set of transitions $\accepSetNBA \subseteq \atmtnTrans$.
	The transitions in $\accepSetNBA$ are called \emph{B\"uchi transitions}.
	Given an NBA $\nba = \nbaFull$, a \emph{(finite or infinite) run of $\nba$ on a (finite or infinite) word $\word = \clr_1\clr_2\ldots\in\colors^*\cup\colors^\omega$} is a sequence $(\atmtnState_0,\clr_1,\atmtnState_1)(\atmtnState_1,\clr_2,\atmtnState_2)\ldots\in\atmtnTrans^*\cup\atmtnTrans^\omega$ such that $\atmtnState_0 \in \atmtnInits$.
	An infinite run $(\atmtnState_0,\clr_1,\atmtnState_1)(\atmtnState_1,\clr_2,\atmtnState_2)\ldots \in \atmtnTrans^\omega$ of $\nba$ is \emph{accepting} if for infinitely many $i\ge 0$, $(\atmtnState_i, \clr_{i+1}, \atmtnState_{i+1})\in\accepSetNBA$.
	A word $\word\in\colors^\omega$ is \emph{accepted} by $\nba$ if there exists an accepting run of $\nba$ on $\word$ --- if not, it is \emph{rejected}.
	We denote the set of infinite words accepted by $\nba$ by $\atmtnLang{\nba}$, and we then say that $\atmtnLang{\nba}$ is the objective \emph{recognized by $\nba$}.

	Here, we take the definition of an \emph{$\omega$-regular objective} as an objective that can be recognized by an NBA (the classical definition uses \emph{$\omega$-regular expressions}, but our definition is well-known to be equivalent~\cite{McN66}).
	Given an automaton structure $\struct = \nbaStruct$, we say that an NBA $\nba$ is \emph{built on top of $\struct$} if there exists $\accepSetNBA \subseteq \atmtnTrans$ such that $\nba = (\atmtnStates, \colors, \atmtnInits, \atmtnTrans, \accepSetNBA)$.

	\begin{rem}
		Notice that for generality, we allow the set of colors $\colors$ to be infinite (as opposed to the finite set of states $\atmtnStates$).
		This is uncommon for $\omega$-regular objectives and automata, but has no real impact: when an objective is specified by an NBA, there will be at most finitely many equivalence classes of colors for the equivalence relation ``inducing exactly the same transitions in the NBA''.
		In Section~\ref{sec:generalCase}, it will be helpful to consider infinite sets of colors.
		\qedEx
	\end{rem}

	\paragraph{Deterministic automata.}
	An automaton structure $\struct = \nbaStruct$ is \emph{deterministic} if $\card{\atmtnInits} = 1$ and, for each $q\in \atmtnStates$ and $\clr\in \colors$, there is exactly one $q'\in \atmtnStates$ such that $(q,\clr,q')\in \atmtnTrans$ (we remark that, without loss of generality, we define deterministic automaton structures so that for each state and each color there is one outgoing transition --- such automata are sometimes called \emph{complete} or \emph{total}).
	For deterministic automaton structures, we use the simplified notation $\struct = \dbaStruct$, which corresponds to structure $\nbaStruct$ where $\atmtnInit$ is the unique element of $\atmtnInits$ and $\atmtnUpd\colon\atmtnStates\times\colors \to \atmtnStates$ is the \emph{update function} that associates to $(\atmtnState,\clr)\in \atmtnStates\times\colors$ the unique $\atmtnState'\in \atmtnStates$ such that $(\atmtnState,\clr,\atmtnState')\in \atmtnTrans$.
	We denote by $\atmtnUpdWord$ the natural extension of $\atmtnUpd$ to finite words --- by induction, the function $\atmtnUpdWord\colon \atmtnStates \times \colors^*\to \atmtnStates$ is such that $\atmtnUpdWord(\atmtnState, \emptyWord) = \atmtnState$, and for $\word\in\colors^*$, $\clr\in\colors$, $\atmtnUpdWord(\atmtnState, \word\clr) = \atmtnUpd(\atmtnUpdWord(\atmtnState, \word), \clr)$.

	A \emph{deterministic B\"uchi automaton} (DBA) is an NBA whose underlying automaton structure is deterministic.
	For DBA, we use the simplified notation $\dba = \dbaFull$, which corresponds to NBA $\nbaFull$ where $\atmtnInit$ and $\atmtnUpd$ are defined as for deterministic structures and $\accepSet \subseteq \atmtnStates \times \colors$ is the projection of $\accepSetNBA$ to the first two components (which does not lose information thanks to determinism).

	For a DBA $\dba$, a state $\atmtnState\in \atmtnStates$, and a word $\word=\clr_1\clr_2\ldots\in \colors^*\cup \colors^\omega$, we denote by $\dbaRun{\atmtnState,\word} = (\atmtnState,\clr_1,q_1)(q_1,\clr_2,q_2)\ldots \in \atmtnTrans^*\cup \atmtnTrans^\omega$ the only run on $\word$ starting from $\atmtnState$.

	An objective $\wc$ is \emph{DBA-recognizable} if there exists a DBA $\dba$ such that $\wc = \atmtnLang{\dba}$.
	For $\goodColors\subseteq\colors$, notice that $\Buchi{\goodColors}$ is DBA-recognizable: it is recognized by the DBA $(\{\atmtnInit\}, \colors, \atmtnInit, \atmtnUpd, \accepSet)$ with a \emph{single} state such that $(\atmtnInit, \clr)\in\accepSet$ if and only if $\clr\in\goodColors$.

	\begin{rem} \label{rem:transitionBased}
		The fact that an automaton with a single state suffices to recognize $\Buchi{\goodColors}$ relies on the assumption that our DBA are \emph{transition-based} and not \emph{state-based} ($\accepSet$ is a set of transitions, not of states).
		Indeed, apart from the trivial cases $\goodColors = \emptyset$ and $\goodColors = \colors$, a state-based DBA recognizing $\Buchi{\goodColors}$ requires two states.

		For standard acceptance conditions, including the B\"uchi one, there is no difference of expressivity between state-based and transition-based automata: every state-based automaton can be converted to an equivalent transition-based automaton (with as many states in general) and every transition-based automaton can be converted to an equivalent state-based automaton (with more states in general).
		Considering transition-based acceptance is especially useful in our paper as it enables to state Condition~\ref{cond:3} of the upcoming characterization (Theorem~\ref{thm:mainChar}); an equivalent formulation for state-based automata would arguably be more technical and less intuitive.
		For instance, Condition~\ref{cond:3} does not apply as stated to the simple $\Buchi{\goodColors}$ if we only consider state-based DBA.
		Moreover, the use of transition-based automata allows the use of techniques such as the saturation of automata with B\"uchi transitions, which is used repeatedly via the application of Lemma~\ref{lem:saturatedAutomaton}.
		\qedEx
	\end{rem}

	In this paper, all automata will be deterministic, and the term ``automaton'' will stand for ``deterministic automaton'' by default.

	\begin{rem}
		Deterministic B\"uchi automata recognize a proper subset of the $\omega$-regular objectives.
		That is, not every non-deterministic B\"uchi automaton can be converted into a deterministic one recognizing the same objective~\cite{Wagner1979omega}.
		\qedEx
	\end{rem}

	\paragraph{Right congruence and prefix preorder.}
	Let $\wc\subseteq\colors^\omega$ be an objective.
	For a finite word $\word\in\colors^*$, we write $\inverse{\word}\wc = \{\word'\in\colors^\omega\mid \word\word'\in\wc\}$ for the set of \emph{winning continuations of $\word$}.
	We define the \emph{right congruence~${\prefEq_\wc} \subseteq \colors^*\times\colors^*$ of $\wc$} as~$\word_1\prefEq_\wc\word_2$ if $\inverse{\word_1}\wc = \inverse{\word_2}\wc$ (meaning that $\word_1$ and $\word_2$ have the same winning continuations).
	Relation~$\prefEq_\wc$ is an equivalence relation.
	When~$\wc$ is clear from the context, we write~$\prefEq$ for~$\prefEq_\wc$ and, for $\word\in\colors^*$, we denote by $\eqClass{\word} \subseteq \colors^*$ its equivalence class of $\prefEq$.

	When $\prefEq$ has finitely many equivalence classes, we can associate a natural deterministic automaton structure $\minStateAtmtn = \minStateAtmtnFull$ to $\prefEq$ such that $\minStateStates$ is the set of equivalence classes of $\prefEq$, $\minStateInit = \eqClass{\emptyWord}$, and $\atmtnUpd_\prefEq(\eqClass{\word}, \clr) = \eqClass{\word\clr}$~\cite{Sta83,MS97}.
	The transition function $\atmtnUpd_\prefEq$ is well-defined since it follows from the definition of $\prefEq$ that if $\word_1 \prefEq \word_2$, then for all $\clr\in\colors$, $\word_1\clr \prefEq \word_2\clr$.
	Hence, the choice of representatives for the equivalence classes does not have an impact on this definition.
	We call the automaton structure $\minStateAtmtn$ the \emph{prefix-classifier of $\wc$}.

	\begin{rem} \label{rem:PIClass}
		Equivalence relation $\prefEq_\wc$ has only one equivalence class if and only if $\wc$ is prefix-independent.
		In particular, an objective has a prefix-classifier with a single state if and only if it is prefix-independent.
		\qedEx
	\end{rem}

	An important property of $\prefEq$ is the following.
	\begin{lem} \label{lem:alwaysEquiv}
		Let $\word_1, \word_2\in\colors^*$.
		If $\word_1 \prefEq \word_2$, then for all $\word\in\colors^*$, $\word_1\word \prefEq \word_2\word$.
	\end{lem}
	\begin{proof}
		If $\word_1$ and $\word_2$ have the same winning continuations, they have in particular the same winning continuations starting with $\word$.
	\end{proof}

	We define the \emph{prefix preorder $\prefOrd_\wc$ of $\wc$}: for $\word_1, \word_2\in\colors^*$, we write $\word_1\prefOrd_\wc\word_2$ if $\inverse{\word_1}\wc \subseteq \inverse{\word_2}\wc$ (meaning that any continuation that is winning after $\word_1$ is also winning after $\word_2$).
	Intuitively, $\word_1\prefOrd_\wc\word_2$ means that a game starting with $\word_2$ is always preferable to a game starting with $\word_1$ for $\Pone$, as there are more ways to win after $\word_2$.
	When $\wc$ is clear from the context, we write $\prefOrd$ for $\prefOrd_\wc$.
	Relation ${\prefOrd} \subseteq {\colors^*\times \colors^*}$ is a preorder.
	Notice that $\prefEq$ is equal to ${\prefOrd} \cap {\invPrefOrd}$.
	We also define the strict preorder ${\strictPrefOrd} = {\prefOrd} \setminus {\prefEq}$.

	Given a DBA $\dba = \dbaFull$ recognizing the objective $\wc$, observe that for $\word, \word'\in\colors^*$ such that $\atmtnUpdWord(\atmtnInit, \word) = \atmtnUpdWord(\atmtnInit, \word')$, we have $\word \prefEq \word'$.
	In this case, equivalence relation $\prefEq$ has at most $\card{\atmtnStates}$ equivalence classes.
	For $\atmtnState\in\atmtnStates$, we write $\dba^\atmtnState$ for the automaton $(\atmtnStates, \alphabet, \atmtnState, \atmtnUpd, \accepSet)$, with $\atmtnState$ as the initial state and everything else as in $\dba$.
	Observe that objective $\atmtnLang{\dba^\atmtnState}$ equals $\inverse{\word}\wc$ for any word $\word\in\colors^*$ such that $\atmtnUpdWord(\atmtnInit, \word) = \atmtnState$.
	This means that as long as $\atmtnState$ is reachable from $\atmtnInit$, objective $\atmtnLang{\dba^\atmtnState}$ is a set of winning continuations of $\wc$.
	We extend the equivalence relation~$\prefEq$ and preorder~$\prefOrd$ to elements of $\atmtnStates$ (we sometimes write $\prefEq_\dba$ and $\prefOrd_\dba$ to avoid any ambiguity): $\atmtnState \prefEq \atmtnState'$ if $\atmtnLang{\dba^\atmtnState} = \atmtnLang{\dba^{\atmtnState'}}$, and $\atmtnState \prefOrd \atmtnState'$ if $\atmtnLang{\dba^\atmtnState} \subseteq \atmtnLang{\dba^{\atmtnState'}}$.
	Transitions of $\dba$ following a given word are non-decreasing w.r.t.\ states.

	\begin{lem} \label{lem:monotonTrans}
		Let $\dba = \dbaFull$ be a DBA.
		For every $\word\in\colors^*$, function $\atmtnUpdWord(\cdot, \word)$ is non-decreasing for $\prefOrd_\dba$ (i.e., for $\atmtnState_1, \atmtnState_2\in\atmtnStates$, if $\atmtnState_1 \prefOrd_\dba \atmtnState_2$, then $\atmtnUpdWord(\atmtnState_1, \word) \prefOrd_\dba \atmtnUpdWord(\atmtnState_2, \word)$).
	\end{lem}
	\begin{proof}
		Let $\wc$ be the objective that $\dba$ recognizes.
		Let $\word\in\colors^*$ and $\atmtnState_1, \atmtnState_2\in\atmtnStates$ be such that $\atmtnState_1 \prefOrd_\dba \atmtnState_2$.
		Let $\atmtnState_1' = \atmtnUpdWord(\atmtnState_1, \word)$ and $\atmtnState_2' = \atmtnUpdWord(\atmtnState_2, \word)$.
		We show that $\atmtnState_1' \prefOrd_\dba \atmtnState_2'$, i.e., that $\atmtnLang{\dba^{\atmtnState_1'}} \subseteq \atmtnLang{\dba^{\atmtnState_2'}}$.
		Let $\word' \in \atmtnLang{\dba^{\atmtnState_1'}}$.
		This implies that $\word\word'\in\atmtnLang{\dba^{\atmtnState_1}}$.
		As $\atmtnState_1 \prefOrd_\dba \atmtnState_2$, we also have that $\word\word'\in\atmtnLang{\dba^{\atmtnState_2}}$.
		This implies that $\word' \in \atmtnLang{\dba^{\atmtnState_2'}}$.
	\end{proof}

	\paragraph{$\accepSet$-free words.}
	Let $\dba = \dbaFull$ be a DBA.
	We say that a finite run $\run$ of $\dba$ is \emph{$\accepSet$-free} if it does not contain any transition from $\accepSet$.
	For $\atmtnState \in \atmtnStates$, we define
	\begin{align*}
		\safe{\atmtnState} &= \{\word \in \colors^* \mid \text{$\dbaRun{\atmtnState, \word}$ is $\accepSet$-free}\},\\
		\safeCycles{\atmtnState} &= \{\word \in \colors^* \mid \word\in\safe{\atmtnState}\ \text{and}\ \atmtnUpdWord(\atmtnState, \word) = \atmtnState\}.
	\end{align*}
	We call the words in the first set the \emph{$\accepSet$-free words from $\atmtnState$}, and the words in the second set the \emph{$\accepSet$-free cycles from $\atmtnState$}.
	We state an important property of $\accepSet$-free words.
	\begin{lem} \label{lem:safeCongruence}
		Let $\atmtnState_1, \atmtnState_2\in\atmtnStates$ be such that $\safe{\atmtnState_1} = \safe{\atmtnState_2}$.
		Then for all $\word\in\safe{\atmtnState_1}$, $\safe{\atmtnUpdWord(\atmtnState_1, \word)} = \safe{\atmtnUpdWord(\atmtnState_2, \word)}$.
	\end{lem}
	\begin{proof}
		Let $\word\in\safe{\atmtnState_1}$.
		Let $q_1'=\atmtnUpdWord(\atmtnState_1, \word)$, $q_2'=\atmtnUpdWord(\atmtnState_2, \word)$, and $\word'\in \safe{q_1'}$.
		Since both runs $\dbaRun{q_1, \word}$ and $\dbaRun{q_1',\word'}$ are $\accepSet$-free, we have $\word\word'\in \safe{\atmtnState_1} = \safe{\atmtnState_2}$.
		Therefore, the finite run $\dbaRun{q_2, \word\word'}$ is $\accepSet$-free, so the run $\dbaRun{q_2', \word'}$ is $\accepSet$-free as well and $\word'\in\safe{\atmtnState'_2}$.
	\end{proof}

	\paragraph{Saturation of DBA with B\"uchi transitions.}
	In what follows, we will make extensive use of a ``normal form'' of B\"uchi automata satisfying that any $\accepSet$-free path can be extended to an $\accepSet$-free cycle.
	Such a normal form can be produced by saturating a given DBA~$\dba$ with B\"uchi transitions~\cite{KS15,AK22Minimizing}.
	To do so, we add to $\accepSet$ all transitions that do not appear in an $\accepSet$-free cycle of $\dba$.
	These transitions can be easily identified: we remove all B\"uchi transitions from~$\dba$, and decompose the resulting automaton in strongly connected components; transitions that we can add to $\accepSet$ are exactly those that are not included in any of those SCCs.


	We say that $\dba=\dbaFull$ is \emph{saturated} if for every $\accepSet' \varsupsetneq \accepSet$, the automaton obtained by replacing $\accepSet$ with $\accepSet'$ does not recognize $\atmtnLang{\dba}$.

	An \emph{$\accepSet$-free component} of $\dba$ is a strongly connected component of the graph obtained by removing the B\"uchi transitions from $\dba$.
	That is, it is a strongly connected component of the graph $(\atmtnStates, \atmtnTransSafe)$, with $\atmtnTransSafe = \{(\atmtnState, \clr, \atmtnUpd(\atmtnState, \clr)) \in \atmtnStates \times \colors \times \atmtnStates \mid (\atmtnState, \clr) \notin \accepSet\}$.

	\begin{lem}\label{lem:saturatedAutomaton}
		Let $\dba=\dbaFull$ be a DBA.
		There is a unique set $\accepSetSat\subseteq \atmtnStates \times \colors$ such that the automaton $\dbaSat = \dbaSatFull$ satisfies that:
		\begin{enumerate}
			\item $\atmtnLang{\dbaSat} = \atmtnLang{\dba}$.
			\item $\dbaSat$ is saturated.
		\end{enumerate}
		Moreover, $\accepSetSat$ is the set of transitions not appearing in any $\accepSet$-free component of $\dba$, and it can be computed in $\bigO(\card{\colors}\cdot\card{\atmtnStates})$ time when $\colors$ is finite.
	\end{lem}
	\begin{proof}
		We first prove the existence of such an $\accepSetSat$.
		Let $(\atmtnStates_1, \atmtnTrans_1), \ldots, (\atmtnStates_k,\atmtnTrans_k)$ be the $\accepSet$-free components of $\dba$.
		We consider the automaton $\dbaSat$ whose B\"uchi transitions are those that do not appear in any $\accepSet$-free component $(\atmtnStates_i,\atmtnTrans_i)$, that is, we let \[\accepSetSat = (\atmtnStates \times \colors) \setminus \{(\atmtnState, \clr) \mid (\atmtnState, \clr, \atmtnUpd(\atmtnState, \clr)) \in \bigcup\limits_{i=1}^{k} \Delta_i\} \; \text{ and } \; \dbaSat=\dbaSatFull.\]
		We show that $\atmtnLang{\dbaSat} = \atmtnLang{\dba}$.
		Since $\accepSet \subseteq \accepSetSat$, it holds that $\atmtnLang{\dba}\subseteq \atmtnLang{\dbaSat}$.
		For the other inclusion, let $\word\notin \atmtnLang{\dba}$.
		There are $w_0\in \colors^*, w'\in \colors^\omega$ such that $\word=w_0w'$ and the finite run $\dbaRun{q_0, w'}$ produced by reading~$w'$ from $q_0 = \atmtnUpdWord(\atmtnInit, w_0)$ does not visit any B\"uchi transition.
		In particular, $\dbaRun{q_0, w'}$ is an infinite path in the finite graph $(\atmtnStates, \atmtnTransSafe)$.
		This implies that eventually, run $\dbaRun{q_0, w'}$ reaches and stays in the same strongly connected component of graph $(\atmtnStates, \atmtnTransSafe)$.
		Formally, there are $\word_1\in\colors^*, \word_2\in\colors^\omega$ such that $\word' = \word_1\word_2$ and the infinite run $\dbaRun{q_1,\word_2}$ produced by reading~$w_2$ from $q_1=\atmtnUpdWord(\atmtnState_0,\word_1)$ lies entirely in some $\accepSet$-free component $(\atmtnStates_i,\atmtnTrans_i)$.
		Thus, all transitions in $\dbaRun{q_1, w_2}$ lie in $\atmtnTrans_i$.
		Therefore, $\dbaSat(\atmtnInit,w)$ is also a rejecting run of~$\dbaSat$ and $\word \notin \atmtnLang{\dbaSat}$.

		We prove that $\dbaSat$ is saturated and the uniqueness of $\accepSetSat$ at the same time.
		Let $\accepSet'$ be another acceptance set such that $\accepSet' \nsubseteq \accepSetSat$ and let $\dba'$ be the automaton obtained by replacing $\accepSetSat$ with $\accepSet'$ in $\dbaSat$.
		Let $(q,\clr)\in \accepSet'\setminus \accepSetSat$.
		Since $(q,\clr)\notin \accepSetSat$, there is an $\accepSet$-free component $(\atmtnStates_i, \atmtnTrans_i)$ such that $(q, \clr, \atmtnUpd(q, \clr))\in \atmtnTrans_i$.
		We can therefore consider a word $w\in \safeCyclesSat{q}$ labeling an $\accepSet$-free cycle including the transition $(q,\clr,\atmtnUpd(q, \clr))$.
		Let $w_0\in \colors^*$ such that $\atmtnUpdWord(\atmtnInit,w_0) = q$.
		Then, $w_0w^\omega\notin \atmtnLang{\dbaSat}$, whereas $w_0w^\omega\in \atmtnLang{\dba'}$, so $\dba'$ does not recognize the same objective as $\dbaSat$.

		When $\colors$ is finite, the set of transitions $\dbaSat$ can be computed in time $\bigO(\card{\colors}\cdot\card{\atmtnStates})$, as it consists of decomposing a graph with at most $\card{\colors}\cdot\card{\atmtnStates}$ transitions into strongly connected components~\cite{Tar72}.
	\end{proof}
	\begin{figure}[tbh]
		\centering
		\begin{minipage}{0.5\columnwidth}
			\centering
			\begin{tikzpicture}[every node/.style={font=\small,inner sep=1pt}]
				\draw (0,0) node[diamant] (qinit) {$\atmtnInit$};
				\draw ($(qinit.north)+(0,0.4)$) edge[-latex'] (qinit);
				\draw ($(qinit)+(1.5,0)$) node[diamant] (qa) {$\atmtnState_a$};
				\draw ($(qinit)+(3,0)$) node[diamant] (top) {$\atmtnState_{aa}$};
				\draw (qinit) edge[-latex',out=30,in=180-30,accepting] node[above=4pt] {$a$} (qa);
				\draw (qa) edge[-latex',out=180+30,in=-30] node[below=4pt] {$b$} (qinit);
				\draw (qa) edge[-latex'] node[above=4pt] {$a$} (top);
				\draw (qinit) edge[-latex',out=150,in=210,distance=0.8cm] node[left=4pt] {$b$} (qinit);
				\draw (top) edge[-latex',out=-30,in=30,distance=0.8cm,accepting] node[right=4pt] {$a, b$} (top);
			\end{tikzpicture}
		\end{minipage}%
		\begin{minipage}{0.5\columnwidth}
			\centering
			\ReachAorAA
		\end{minipage}%
		\caption{A DBA (left) and its unique saturation (right).
			Transitions labeled with a $\bullet$ symbol are the B\"uchi transitions.
			In figures, automaton states are depicted with diamonds.}
		\label{fig:saturation}
	\end{figure}

	In Figure~\ref{fig:saturation}, we show an example of the saturation process presented in the proof of Lemma~\ref{lem:saturatedAutomaton}.

	The following simple lemma follows, which holds true in saturated DBA: every word that is $\accepSet$-free from a state can be completed into an $\accepSet$-free cycle from the same state.
	This will be a key technical lemma in the upcoming proofs.
	\begin{lem} \label{lem:safeComp}
		Let $\dba = \dbaFull$ be a saturated DBA.
		Let $\atmtnState\in\atmtnStates$ and $\word \in \safe{\atmtnState}$.
		There exists $\word' \in \colors^*$ such that $\word\word'\in\safeCycles{\atmtnState}$.
	\end{lem}
	\begin{proof}
		Let $q'=\atmtnUpdWord(q,\word)$.
		Thanks to the saturation property and by Lemma~\ref{lem:saturatedAutomaton}, $\accepSet$ contains all transitions that do not belong to an $\accepSet$-free component.
		This implies that any two states connected by an $\accepSet$-free run are in the same $\accepSet$-free component.
		In particular, as $q'$ is reachable from $q$ through an $\accepSet$-free run, $q'$ must belong to the same $\accepSet$-free component as $q$.
		Therefore, there exists an $\accepSet$-free run from $q'$ to $q$.
		Taking the word $\word'$ labeling this run, we obtain the desired result.
	\end{proof}

	\section{Half-positionality characterization for DBA-recognizable objectives} \label{sec:characterization}
	In this section, we present our main contribution in Theorem~\ref{thm:mainChar}, by giving three conditions that exactly characterize half-positional DBA-recognizable objectives.
	These conditions are presented in Section~\ref{sec:three-conditions}.
	Theorem~\ref{thm:mainChar} and several consequences of it are stated in Section~\ref{sec:characterization-corollaries} (the proof of Theorem~\ref{thm:mainChar} is postponed to Sections~\ref{sec:necessary} and~\ref{sec:sufficient}).
	In Section~\ref{sec:complexity}, we use this characterization to show that we can decide the half-positionality of a DBA in polynomial time.

	\subsection{Three conditions for half-positionality} \label{sec:three-conditions}
	We define the three conditions of objectives at the core of our characterization.

	\begin{condition}[Total prefix preorder] \label{cond1}
		We say that an objective $\wc \subseteq \colors^\omega$ has a \emph{total prefix preorder} if for all $\word_1, \word_2 \in \colors^*$, $\word_1 \prefOrd_\wc \word_2$ or $\word_2 \prefOrd_\wc \word_1$.
	\end{condition}

	\begin{exa}[Not total prefix preorder] \label{ex:abbC}
		Let $\colors = \{a, b\}$.
		We consider the objective $\wc$ recognized by the DBA $\dba$ depicted in Figure~\ref{fig:abbC} (left).
		It consists of the infinite words starting with $aa$ or $bb$.
		This objective does not have a total prefix preorder: words $a$ and $b$ are incomparable for $\prefOrd_\wc$.
		Indeed, $a^\omega$ is winning after $a$ but not after $b$, and $b^\omega$ is winning after $b$ but not after $a$.
		In terms of automaton states, we have that $\atmtnState_a$ and $\atmtnState_b$ are incomparable for~$\prefOrd_\dba$.
		This objective is not half-positional, as witnessed by the arena on the right of Figure~\ref{fig:abbC}.
		In this arena, $\Pone$ is able to win when the game starts in $\s_1$ by playing $a$ in $\s_3$, and when the game starts in $\s_2$ by playing $b$.
		However, no positional strategy wins from both $\s_1$ and~$\s_2$.
		\qedEx
	\end{exa}
	\begin{figure}[tbh]
		\centering
		\begin{minipage}{0.55\columnwidth}
			\centering
			\begin{tikzpicture}[every node/.style={font=\small,inner sep=1pt}]
				\draw (0,0) node[diamant] (qinit) {$\atmtnInit$};
				\draw ($(qinit.north)+(0,0.4)$) edge[-latex'] (qinit);
				\draw ($(qinit)+(2,0.7)$) node[diamant] (qa) {$\atmtnState_a$};
				\draw ($(qinit)+(2,-0.7)$) node[diamant] (qb) {$\atmtnState_b$};
				\draw ($(qa)+(2,0)$) node[diamant] (top) {$\atmtnState_{\mathsf{win}}$};
				\draw ($(qb)+(2,0)$) node[diamant] (bot) {$\atmtnState_{\mathsf{lose}}$};
				\draw (qinit) edge[-latex'] node[above=4pt] {$a$} (qa);
				\draw (qinit) edge[-latex'] node[below=4pt] {$b$} (qb);
				\draw (qa) edge[-latex'] node[above=4pt] {$a$} (top);
				\draw (qb) edge[-latex'] node[above=4pt] {$b$} (top);
				\draw (qb) edge[-latex'] node[below=4pt] {$a$} (bot);
				\draw (qa) edge[-latex'] node[below=4pt] {$b$} (bot);
				\draw (top) edge[-latex',out=-30,in=30,distance=0.8cm,accepting] node[right=4pt] {$a, b$} (top);
				\draw (bot) edge[-latex',out=-30,in=30,distance=0.8cm] node[right=4pt] {$a, b$} (bot);
			\end{tikzpicture}
		\end{minipage}%
		\begin{minipage}{0.45\columnwidth}
			\centering
			\begin{tikzpicture}[every node/.style={font=\small,inner sep=1pt}]
				\draw (0,0) node[rond] (s1) {$\s_1$};
				\draw ($(s1)+(0,-1.4)$) node[rond] (s2) {$\s_2$};
				\draw ($(s1)!0.5!(s2)+(2,0)$) node[rond] (s3) {$\s_3$};
				\draw (s1) edge[-latex'] node[above=4pt] {$a$} (s3);
				\draw (s2) edge[-latex'] node[below=4pt] {$b$} (s3);
				\draw (s3) edge[-latex',out=60,in=120,distance=0.8cm] node[above=4pt] {$a$} (s3);
				\draw (s3) edge[-latex',out=-120,in=-60,distance=0.8cm] node[below=4pt] {$b$} (s3);
			\end{tikzpicture}
		\end{minipage}%
		\caption{DBA $\dba$ recognizing objective $\wc = (aa+bb)\colors^\omega$ (left), and an arena in which positional strategies do not suffice for $\Pone$ to play optimally for this objective (right).
			In figures, circles represent arena vertices controlled by $\Pone$.}
		\label{fig:abbC}
	\end{figure}

	\begin{rem} \label{rem:totalSym}
		The prefix preorder of an objective $\wc$ is total if and only if the prefix preorder of its complement $\comp{\wc}$ is total.
		\qedEx
	\end{rem}

	\begin{rem}
		Having a total prefix preorder is equivalent to the \emph{strong monotony} notion~\cite{BFMM11} in general, and equivalent to \emph{monotony}~\cite{GZ05} for $\omega$-regular objectives.
		We discuss in more depth the relation between the conditions appearing in the characterization and other properties from the literature studying half-positionality in Appendix~\ref{app:relation}.
		\qedEx
	\end{rem}

	A straightforward result for an objective $\wc$ recognized by a DBA $\dba$ is that it has a total prefix preorder if and only if the (reachable) states of $\dba$ are totally ordered for $\prefOrd_\dba$.

	\begin{condition}[Progress-consistency] \label{cond2}
		We say that an objective $\wc$ is \emph{progress-consistent} if for all $\word_1\in\colors^*$ and $\word_2\in\colors^+$ such that $\word_1 \strictPrefOrd \word_1\word_2$, we have $\word_1(\word_2)^\omega\in\wc$.
	\end{condition}

	Intuitively, this means that whenever a word $\word_2$ can be used to make progress after starting with a word $\word_1$ (in the sense of getting to a position in which more continuations are winning), then repeating this word has to be winning.

	\begin{exa}[Not progress-consistent] \label{ex:AA}
		Let $\colors = \{a, b\}$.
		We consider the objective $\wc=\colors^*aa\colors^\omega$ recognized by the DBA with three states in Figure~\ref{fig:AA} (left).
		This objective consists of the words containing, at some point, twice the color $a$ in a row.
		Notice that the prefix preorder of this objective is total ($\atmtnInit \strictPrefOrd \atmtnState_a \strictPrefOrd \atmtnState_{aa}$).
		This objective is not progress-consistent: we have $\emptyWord \strictPrefOrd ba$, but $(ba)^\omega\notin\wc$.
		This objective is not half-positional: if $\Pone$ plays in an arena with a choice among two cycles $ba$ and $ab$ depicted in Figure~\ref{fig:AA} (right), it is possible to win by playing $ba$ and then $ab$, but a positional strategy can only achieve words $(ba)^\omega$ or $(ab)^\omega$, which are both losing.
		\qedEx
	\end{exa}
	\begin{figure}[tbh]
		\centering
		\begin{minipage}{0.5\columnwidth}
			\centering
			\begin{tikzpicture}[every node/.style={font=\small,inner sep=1pt}]
				\draw (0,0) node[diamant] (qinit) {$\atmtnInit$};
				\draw ($(qinit.north)+(0,0.4)$) edge[-latex'] (qinit);
				\draw ($(qinit)+(1.5,0)$) node[diamant] (qa) {$\atmtnState_a$};
				\draw ($(qinit)+(3,0)$) node[diamant] (top) {$\atmtnState_{aa}$};
				\draw (qinit) edge[-latex',out=30,in=180-30] node[above=4pt] {$a$} (qa);
				\draw (qa) edge[-latex',out=180+30,in=-30] node[below=4pt] {$b$} (qinit);
				\draw (qa) edge[-latex',accepting] node[above=4pt] {$a$} (top);
				\draw (qinit) edge[-latex',out=150,in=210,distance=0.8cm] node[left=4pt] {$b$} (qinit);
				\draw (top) edge[-latex',out=-30,in=30,distance=0.8cm,accepting] node[right=4pt] {$a, b$} (top);
			\end{tikzpicture}
		\end{minipage}%
		\begin{minipage}{0.5\columnwidth}
			\centering
			\begin{tikzpicture}[every node/.style={font=\small,inner sep=1pt}]
				\draw (0,0) node[rond] (s) {$\s$};
				\draw (s) edge[-latex',out=150,in=210,distance=0.8cm,decorate] node[left=4pt] {$ab$} (s);
				\draw (s) edge[-latex',out=-30,in=30,distance=0.8cm,decorate] node[right=4pt] {$ba$} (s);
			\end{tikzpicture}
		\end{minipage}%
		\caption{A DBA recognizing the set of words containing $aa$ at some point (left), an arena in which positional strategies do not suffice for $\Pone$ to play optimally for this objective (right).
			Squiggly arrows indicate a sequence of edges or transitions (here, a sequence of two edges).}
		\label{fig:AA}
	\end{figure}

	\begin{exa}[Progress-consistent objective] \label{ex:AAorBuchiA}
		We consider a slight modification of the previous example by adding two B\"uchi transitions: see the DBA in Figure~\ref{fig:AAorReachAA}.
		The objective recognized by this DBA is $\wc=\Buchi{\{a\}}\cup \colors^*aa\colors^\omega$: $\wc$ is the set of words containing infinitely many occurrences of $a$, or containing $a$ twice in a row at some point.
		The sets of winning continuations of $\wc$ are $\atmtnLang{\dba^{\atmtnInit}}=\wc$, $\atmtnLang{\dba^{\atmtnState_{a}}}=a\colors^\omega\cup \wc$ and $\atmtnLang{\dba^{\atmtnState_{aa}}}=\colors^\omega$.
		This objective is progress-consistent: any word reaching $\atmtnState_{aa}$ is straightforwardly accepted when repeated infinitely often, and any word $\word$ such that $\atmtnUpdWord(\atmtnInit, \word) = \atmtnState_a$ necessarily contains at least one $a$, and thus is accepted when repeated infinitely often.
		Objective $\wc$ is half-positional, which will be readily shown with our upcoming characterization (Theorem~\ref{thm:mainChar}).

		Here, notice that the complement $\comp{\wc}$ of $\wc$ is not progress-consistent.
		Indeed, $a \strictPrefOrd_{\comp{\wc}} a(bab)$, but $a(bab)^\omega\notin \comp{\wc}$.
		Unlike having a total prefix preorder, progress-consistency can hold for an objective but not its complement.

		Note that half-positionality of $\wc$ cannot be shown using existing half-positionality criteria~\cite{Kop06,BFMM11} (it is neither prefix-independent nor \emph{concave}) or bipositionality criteria, as it is simply not bipositional.
		\qedEx
	\end{exa}
	\begin{figure}[tbh]
		\centering
		\ReachAorAA
		\caption{A DBA recognizing the set of words containing infinitely many occurrences of $a$, or containing $aa$ at some point.}
		\label{fig:AAorReachAA}
	\end{figure}

	\begin{condition}[Recognizability by the prefix-classifier] \label{cond3}
		Being recognized by a B\"uchi automaton built on top of the prefix-classifier is our third condition.
		In other words, for a DBA-recognizable objective $\wc\subseteq\colors^\omega$ and its prefix-classifier $\prefClass = \minStateAtmtnFull$, this condition requires that there exists $\accepSet_\prefEq \subseteq \minStateStates \times \colors$ such that $\wc$ is recognized by DBA~$(\minStateStates, \colors, \minStateInit, \minStateUpd, \accepSet_\prefEq)$.
	\end{condition}
	We show an example of a DBA-recognizable objective that satisfies the first two conditions (having a total prefix preorder and progress-consistency), but not this third condition, and which is not half-positional.

	\begin{exa}[Not recognizable by the prefix-classifier] \label{ex:BuchiABuchiB}
		Let $\colors = \{a, b\}$.
		We consider the objective $\wc = \Buchi{\{a\}} \cap \Buchi{\{b\}}$ recognized by the DBA in Figure~\ref{fig:BuchiABuchiB}.
		This objective is prefix-independent: as such (Remark~\ref{rem:PIClass}), there is only one equivalence class of $\prefEq$.
		This implies that the prefix preorder is total, and that $\wc$ is progress-consistent (the premise of the progress-consistency property can never be true).
		This objective is not half-positional, as witnessed by the arena in Figure~\ref{fig:BuchiABuchiB} (right): $\Pone$ has a winning strategy from $\s$, but it needs to take infinitely often both $a$ and $b$.

		Any DBA recognizing this objective has at least two states, but all their (reachable) states are equivalent for $\prefEq$ --- no matter the state we choose as an initial state, the recognized objective is the same (by prefix-independence).
		As it is prefix-independent, its prefix-classifier $\minStateAtmtn$ has only one state.
		\qedEx
	\end{exa}
	\begin{figure}[tbh]
		\centering
		\begin{minipage}{0.5\columnwidth}
			\centering
			\begin{tikzpicture}[every node/.style={font=\small,inner sep=1pt}]
				\draw (0,0) node[diamant] (qa) {$\atmtnState_1$};
				\draw ($(qa.north)+(0,0.4)$) edge[-latex'] (qa);
				\draw ($(qinit)+(2,0)$) node[diamant] (qb) {$\atmtnState_2$};
				\draw (qa) edge[-latex',out=30,in=180-30,accepting] node[above=4pt] {$b$} (qb);
				\draw (qb) edge[-latex',out=180+30,in=-30,accepting] node[below=4pt] {$a$} (qa);
				\draw (qa) edge[-latex',out=150,in=210,distance=0.8cm] node[left=4pt] {$a$} (qa);
				\draw (qb) edge[-latex',out=-30,in=30,distance=0.8cm] node[right=4pt] {$b$} (qb);
			\end{tikzpicture}
		\end{minipage}%
		\begin{minipage}{0.5\columnwidth}
			\centering
			\begin{tikzpicture}[every node/.style={font=\small,inner sep=1pt}]
				\draw (0,0) node[rond] (s) {$\s$};
				\draw (s) edge[-latex',out=150,in=210,distance=0.8cm] node[left=4pt] {$a$} (s);
				\draw (s) edge[-latex',out=-30,in=30,distance=0.8cm] node[right=4pt] {$b$} (s);
			\end{tikzpicture}
		\end{minipage}%
		\caption{DBA recognizing the objective $\Buchi{\{a\}} \cap \Buchi{\{b\}}$ (left), and an arena in which positional strategies do not suffice for $\Pone$ to play optimally for this objective (right).}
		\label{fig:BuchiABuchiB}
	\end{figure}

	As will be shown formally, being recognized by a DBA built on top of the prefix-classifier is necessary for half-positionality of \emph{DBA-recognizable} objectives over finite one-player arenas.
	Unlike the two other conditions, it is in general not necessary for half-positionality of general objectives, including objectives recognized by other standard classes of $\omega$-automata.%
	\begin{exa} \label{ex:coBuchiAorB}
		We consider the complement $\comp{\wc}$ of the objective $\wc = \Buchi{\{a\}} \cap \Buchi{\{b\}}$ of Example~\ref{ex:BuchiABuchiB}, which consists of the words ending with $a^\omega$ or $b^\omega$.
		Objective $\comp{\wc}$ is not DBA-recognizable (a close proof can be found in~\cite[Theorem~4.50]{BK08}).
		Still, it is recognizable by a \emph{deterministic coB\"uchi automaton} similar to the automaton in Figure~\ref{fig:BuchiABuchiB}, but which accepts infinite words that visit transitions labeled with $\bullet$ only finitely often.
		This objective is half-positional, which can be shown using~\cite[Theorem~6]{DJW97}.
		However, its prefix-classifier has just one state, and there is no way to recognize $\comp{\wc}$ by building a coB\"uchi (or even parity) automaton on top of it.
		\qedEx
	\end{exa}

	\subsection{Characterization and corollaries}\label{sec:characterization-corollaries}
	We have now defined the three conditions required for our characterization.
	\begin{thm} \label{thm:mainChar}
		Let $\wc\subseteq\colors^\omega$ be a DBA-recognizable objective.
		Objective $\wc$ is half-positional (over all arenas) if and only if
		\begin{itemize}
			\item its prefix preorder $\prefOrd$ is total,
			\item it is progress-consistent, and
			\item it can be recognized by a B\"uchi automaton built on top of its prefix-classifier $\prefClass$.
		\end{itemize}
	\end{thm}
	\begin{proof}
		The proof of the necessity of the three conditions can be found in Section~\ref{sec:necessary}, respectively in Propositions~\ref{prop:orderedNecessary},~\ref{prop:progCons}, and~\ref{prop:necessaryDBA}.
		The proof of the sufficiency of the conjunction of the three conditions can be found in Section~\ref{sec:sufficient}, Proposition~\ref{prop:sufficient}.
	\end{proof}

	This characterization is valuable to prove (and disprove) half-positionality of DBA-recognizable objectives.
	Examples~\ref{ex:abbC},~\ref{ex:AA}, and~\ref{ex:BuchiABuchiB} are all not half-positional, and each of them falsifies exactly one of the three conditions from the statement.
	On the other hand, Example~\ref{ex:AAorBuchiA} is half-positional.
	We have already discussed its progress-consistency, but it is also straightforward to verify that its prefix preorder is total and that it is recognizable by its prefix-classifier: the right congruence has three totally ordered equivalence classes corresponding to the states of the DBA in Figure~\ref{fig:AAorReachAA}.

	We state two notable consequences of Theorem~\ref{thm:mainChar} and of its proof technique.
	The first one is the specialization of Theorem~\ref{thm:mainChar} to prefix-independent objectives.
	It states that all prefix-independent, DBA-recognizable objectives that are half-positional are of the kind $\Buchi{\goodColors}$ for some $\goodColors\subseteq \colors$.
	Prefix-independence of objectives is a frequent assumption in the literature~\cite{Kop06,CN06,GK14,CFGO22} --- we show that under this assumption, half-positionality of DBA-recognizable objectives is very easy to understand and characterize.
	\begin{prop} \label{prop:charPI}
		Let $\wc\subseteq\colors^\omega$ be a prefix-independent, DBA-recognizable objective.
		Objective $\wc$ is half-positional if and only if there exists $\goodColors \subseteq \colors$ such that $\wc = \Buchi{\goodColors}$.
	\end{prop}
	\begin{proof}
		The right-to-left implication follows from the known half-positionality of objectives of the kind $\Buchi{\goodColors}$ (this is a special case of a parity game~\cite{EJ91}).
		For the left-to-right implication, we assume that $\wc$ is a prefix-independent, DBA-recognizable, half-positional objective.
		By Theorem~\ref{thm:mainChar}, it is recognized by a DBA $\dba$ built on top of $\prefClass$.
		As $\wc$ is prefix-independent (Remark~\ref{rem:PIClass}), its prefix-classifier has just one state, and there is a single transition from and to this single state for each color.
		Hence, $\wc = \Buchi{\goodColors}$, where $\goodColors$ is the set of colors whose only transition is a B\"uchi transition of $\dba$.
	\end{proof}
	\begin{rem}
		A corollary of this result is that when $\wc$ is prefix-independent, DBA-recognizable, and half-positional, we also have that $\comp{\wc}$ is half-positional.
		Indeed, the complement of objective $\wc = \Buchi{\goodColors}$ is a so-called \emph{coB\"uchi objective}, which is also known to be half-positional~\cite{EJ91}.
		This statement does not hold in general when $\wc$ is not prefix-independent, as was shown in Example~\ref{ex:AAorBuchiA}.
		Moreover, the reciprocal of the statement also does not hold, as was shown in Example~\ref{ex:coBuchiAorB}.%
		\qedEx
	\end{rem}

	\begin{rem}
		A second corollary is that prefix-independent DBA-recognizable half-positional objectives are closed under finite union (since a finite union of B\"uchi conditions is a B\"uchi condition).
		This settles Kopczy\'nski's conjecture for DBA-recognizable objectives.
		\qedEx
	\end{rem}

	A second consequence of Theorem~\ref{thm:mainChar} and its proof technique shows that half-positionality of DBA-recognizable objectives can be reduced to half-positionality over the restricted class of \emph{finite, one-player arenas}.
	Results reducing strategy complexity in two-player arenas to the easier question of strategy complexity in one-player arenas are sometimes called \emph{one-to-two-player lifts} and appear in multiple places in the literature~\cite{GZ05,BLORV22,Koz22,BRV23}.
	\begin{prop}[One-to-two-player and finite-to-infinite lift] \label{prop:1to2}
		Let $\wc\subseteq\colors^\omega$ be a DBA-recognizable objective.
		If objective $\wc$ is half-positional over finite one-player arenas, then it is half-positional over all arenas (of any cardinality).
	\end{prop}
	\begin{proof}
		When showing the necessity of the three conditions for half-positionality of DBA-recognizable objectives in Section~\ref{sec:necessary} (Propositions~\ref{prop:orderedNecessary},~\ref{prop:progCons}, and~\ref{prop:necessaryDBA}), we actually show their necessity for half-positionality over \emph{finite one-player} arenas.
		Hence, assuming half-positionality over finite one-player arenas, we have the three conditions from the characterization of Theorem~\ref{thm:mainChar}, so we have half-positionality over all arenas.
	\end{proof}

	One-to-two-player lifts from the literature all require an assumption on the strategy complexity of \emph{both} players, and are either stated solely over finite arenas, or solely over infinite arenas.
	Proposition~\ref{prop:1to2}, albeit set in the more restricted context of DBA-recognizable objectives, displays stronger properties than the known one-to-two-player lifts.
	\begin{itemize}
		\item It is \emph{asymmetric} in the sense that we simply need a hypothesis on \emph{one} player: \emph{half-}positionality of DBA-recognizable objectives over one-player arenas implies their \emph{half-}positionality over two-player arenas.
		\item It shows that half-positionality of DBA-recognizable objectives over \emph{finite} arenas implies half-positionality over \emph{infinite} arenas.
	\end{itemize}
	Both these properties do not hold for general objectives.
	\begin{itemize}
		\item Some objectives are half-positional over one-player but not over two-player arenas~\cite[Section~7]{GK14Old} --- we have that this is not possible for DBA-recognizable objectives.
		\item Some objectives are half-positional over finite but not over infinite arenas (see, e.g., the mean-payoff objective~\cite{EM79,Put94}) --- we have that this is not possible for DBA-recognizable objectives.
	\end{itemize}

	\subsection{Deciding half-positionality in polynomial time} \label{sec:complexity}
	In this section, we assume that $\colors$ is finite.
	We show that the problem of deciding, given a DBA $\dba = \dbaFull$ as an input, whether $\atmtnLang{\dba}$ is half-positional can be solved in polynomial time, and more precisely in time $\bigO(\card{\colors}\cdot\card{\atmtnStates}^4)$.
	We will use the following classical result about the complexity of checking the inclusion of languages recognized by DBA.

	\begin{lem} \label{lem:containmentComplexity}
		Let $\dba = \dbaFull$ and $\dba' = (\atmtnStates', \colors, \atmtnInit', \atmtnUpd', \accepSet')$ be two DBA on the same finite set of colors $\colors$.
		Deciding whether $\atmtnLang{\dba} \subseteq \atmtnLang{\dba'}$ can be done in time $\bigO(\card{\colors} \cdot \card{\atmtnStates} \cdot \card{\atmtnStates'})$.
	\end{lem}
	\begin{proof}
		Observe first that $\atmtnLang{\dba} \subseteq \atmtnLang{\dba'}$ if and only if $\atmtnLang{\dba} \cap \comp{\atmtnLang{\dba'}} = \emptyset$.
		Therefore, we first want to complement $\dba'$ (as an NBA $\overline{\nba'}$), and then check the emptiness of the intersection of the languages of $\dba$ and $\overline{\nba'}$.

		We use the fact that DBA can be easily complemented (as NBA): there is an NBA $\overline{\nba'}$ with at most $2\cdot\card{\atmtnStates'}$ states and at most two non-deterministic transitions for each element of $\colors\times\atmtnStates'$ such that $\atmtnLang{\overline{\nba'}} = \comp{\atmtnLang{\dba'}}$~\cite{Kur87}.\footnote{The construction to complement a DBA as an NBA is simple: keep the original DBA, add a copy of it in which the transitions in $\accepSet'$ are removed, add non-deterministic transitions from the original to the copy mimicking the original transitions, and define as B\"uchi transitions only the transitions left in the copy.
			In this way, accepted words in the new NBA are the ones that, from some point on (guessed non-deterministically), move to the copy and only use transitions that were not B\"uchi transitions of the original DBA.}

		Then, we use the fact that we can construct an NBA $\nba_\cap$ with at most $2\cdot \card{\atmtnStates}\cdot(2\cdot\card{\atmtnStates'}) = \bigO(\card{\atmtnStates}\cdot\card{\atmtnStates'})$ states and $\bigO(\card{\colors}\cdot \card{\atmtnStates} \cdot \card{\atmtnStates'})$ transitions such that $\atmtnLang{\nba_\cap} = \atmtnLang{\dba} \cap \atmtnLang{\overline{\nba'}}$~\cite[Section~3.1]{Bok18}.

		The emptiness of the language of $\nba_\cap$ can finally be checked in time $\bigO(\card{\atmtnStates}\cdot\card{\atmtnStates'} + \card{\colors}\cdot \card{\atmtnStates} \cdot \card{\atmtnStates'}) = \bigO(\card{\colors}\cdot \card{\atmtnStates} \cdot \card{\atmtnStates'})$ by decomposing the graph of the NBA into strongly connected components~\cite{Tar72}.
	\end{proof}

	We investigate how to verify each property used in the characterization of Theorem~\ref{thm:mainChar}.
	Let $\dba = \dbaFull$ be a DBA (we assume w.l.o.g.\ that all states in $\atmtnStates$ are reachable from $\atmtnInit$) and $\wc = \atmtnLang{\dba}$ be the objective it recognizes.
	Our algorithm first verifies that the prefix preorder is total and recognizability by $\prefClass$, and then, under these first two assumptions, progress-consistency.
	For each condition, we sketch an algorithm to decide it, and we discuss the time complexity of this algorithm.%

	\paragraph{Total prefix preorder.}
	To check that $\wc$ has a total prefix preorder, it suffices to check that the states of $\dba$ are totally preordered by $\prefOrd_\dba$.
	We start by computing, for each pair of states $\atmtnState, \atmtnState'\in\atmtnStates$, whether $\atmtnState \prefOrd_\dba \atmtnState'$, $\atmtnState' \prefOrd_\dba \atmtnState$, or none of these.
	This can be rephrased as an \emph{inclusion problem} for two DBA-recognizable objectives: if $\dba^\atmtnState = (\atmtnStates, \colors, \atmtnState, \atmtnUpd, \accepSet)$ and $\dba^{\atmtnState'} = (\atmtnStates, \colors, \atmtnState', \atmtnUpd, \accepSet)$,
	we have that $\atmtnState \prefOrd_\dba \atmtnState'$ if and only if $\atmtnLang{\dba^\atmtnState} \subseteq \atmtnLang{\dba^{\atmtnState'}}$.
	By Lemma~\ref{lem:containmentComplexity}, this can therefore be decided in time $\bigO(\card{\colors} \cdot \card{\atmtnStates}^2)$.
	We can therefore know for all $\card{\atmtnStates}^2$ pairs $\atmtnState, \atmtnState'\in\atmtnStates$ whether $\atmtnState \prefOrd_\dba \atmtnState'$, $\atmtnState' \prefOrd_\dba \atmtnState$, $\atmtnState' \prefEq_\dba \atmtnState$ (as ${\prefEq_\dba} = {\prefOrd_\dba} \cap {\invPrefOrd_\dba}$), or none of these in time $\bigO(\card{\atmtnStates}^2\cdot(\card{\colors}\cdot\card{\atmtnStates}^2)) = \bigO(\card{\colors}\cdot\card{\atmtnStates}^4)$.
	In particular, the prefix preorder is total if and only if for all $\atmtnState, \atmtnState' \in \atmtnStates$, we have $\atmtnState \prefOrd_\dba \atmtnState'$ or $\atmtnState' \prefOrd_\dba \atmtnState$.

	\paragraph{Recognizability by the prefix-classifier.}
	After all the relations $\prefOrd_\dba$ and $\prefEq_\dba$ between pairs of states are computed in the previous step, we can compute the states and transitions of the prefix-classifier $\prefClass = \minStateAtmtnFull$ by merging all the equivalence classes of $\prefEq_\dba$.
	We assume for simplicity that $\minStateStates = \quotient{\atmtnStates}{\prefEq_\dba}$.

	We now wonder whether it is possible to recognize $\wc$ with a B\"uchi automaton built on top of $\prefClass$, i.e., if it possible to find a set $\accepSet_\prefEq \subseteq \minStateStates \times \colors$ of transitions of $\prefClass$ such that the language of DBA $(\minStateStates, \colors, \minStateInit, \minStateUpd, \accepSet_\prefEq)$ is $\wc$.
	We simplify the search for such a set with the following result, which shows that when $\dba$ is saturated, it suffices to try with one specific set $\accepSet_\prefEq$.
	We can then simply check whether $\wc = \atmtnLang{(\minStateStates, \colors, \minStateInit, \minStateUpd, \accepSet_\prefEq)}$, an equivalence query which can be performed in time $\bigO(\card{\colors}\cdot\card{\atmtnStates}^2)$ by checking two inclusion queries as in Lemma~\ref{lem:containmentComplexity}.

	\begin{lem} \label{lem:recoComp}
		We assume that $\dba$ is saturated and that $\wc$ is recognized by a DBA built on top of the prefix-classifier $\prefClass=\minStateAtmtnFull$.
		We define \[
		\accepSet_\prefEq = \{(\eqClass{\atmtnState}, \clr) \in \minStateStates \times \colors \mid \forall \atmtnState'\in\eqClass{\atmtnState}, (\atmtnState', \clr) \in \accepSet\}.
		\]
		Then, $\wc$ is recognized by $(\minStateStates, \colors, \minStateInit, \minStateUpd, \accepSet_\prefEq)$.
	\end{lem}
	\begin{proof}
		We assume that $\wc$ is recognized by a DBA built on top of $\prefClass$.
		We start by saturating this DBA, which yields a set of B\"uchi transitions $\accepSet'$ such that $\wc$ is also recognized by the saturated DBA $\dba' = (\minStateStates, \colors, \minStateInit, \minStateUpd, \accepSet')$ (Lemma~\ref{lem:saturatedAutomaton}).
		To prove the claim, we show that $\accepSet' = \accepSet_\prefEq$.

		We first show that $\accepSet' \subseteq \accepSet_\prefEq$.
		Let $(\eqClass{\atmtnState}, \clr) \notin \accepSet_\prefEq$ --- we show that $(\eqClass{\atmtnState}, \clr) \notin \accepSet'$.
		As $(\eqClass{\atmtnState}, \clr) \notin \accepSet_\prefEq$, by definition of $\accepSet_\prefEq$, there is $\atmtnState'\in\atmtnStates$ such that $(\atmtnState', \clr) \notin \accepSet$.
		As $\dba$ is saturated, by Lemma~\ref{lem:safeComp}, there exists $\word'\in\colors^*$ such that $\clr\word'\in\safeCycles{\atmtnState'}$.
		By construction of the prefix-classifier, $\atmtnUpd_\prefEq^*(\eqClass{\atmtnState}, \clr\word') = \eqClass{\atmtnState}$.
		Also, as $\wc = \atmtnLang{\dba'}$, word $(\clr\word')^\omega$ must be rejected from $\eqClass{\atmtnState}$ in $\dba'$.
		Therefore, $(\eqClass{\atmtnState}, \clr)$ cannot be a B\"uchi transition of $\dba'$ and is not in $\accepSet'$.

		We now show that $\accepSet_\prefEq \subseteq \accepSet'$.
		Let $(\eqClass{\atmtnState}, \clr) \notin \accepSet'$ --- we show that $(\eqClass{\atmtnState}, \clr) \notin \accepSet_\prefEq$.
		As $\dba'$ is saturated, by Lemma~\ref{lem:safeComp}, there exists $\word'\in\colors^*$ such that $\clr\word'\in\safeCyclesp{\eqClass{\atmtnState}}$.
		As $\wc = \atmtnLang{\dba'}$, word $(\clr\word')^\omega$ is rejected from any state in $\eqClass{\atmtnState}$ in $\dba$.
		If for all $\atmtnState'\in\eqClass{\atmtnState}$, $(\atmtnState', \clr)$ was in $\accepSet$, $(\clr\word')^\omega$ would be accepted from all states in $\eqClass{\atmtnState}$ in $\dba$.
		Hence, there exists $\atmtnState'\in\eqClass{\atmtnState}$ such that $(\atmtnState', \clr)\notin\accepSet$.
		We conclude that $(\eqClass{\atmtnState}, \clr) \notin \accepSet_\prefEq$.
	\end{proof}

	\paragraph{Progress-consistency.}
	We assume that we have already checked that $\wc$ is recognizable by a B\"uchi automaton built on top of $\prefClass$, and that we know the (total) ordering of the states.
	We show that checking progress-consistency, under these two hypotheses, can be done in polynomial time.
	We prove a lemma reducing the search for words witnessing that $\wc$ is not progress-consistent to a problem computationally easier to investigate.

	\begin{lem} \label{lem:progConsComplexity}
		We assume that $\dba$ is built on top of the prefix-classifier $\prefClass$ and that the prefix preorder of $\wc$ is total.
		Then, $\wc$ is progress-consistent if and only if for all $\atmtnState, \atmtnState'\in\atmtnStates$ with $\atmtnState \strictPrefOrd_\dba \atmtnState'$,
		\[\{\word\in\colors^+ \mid \atmtnUpdWord(\atmtnState, \word) = \atmtnState'\} \cap \safeCycles{\atmtnState'} = \emptyset.\]
	\end{lem}
	\begin{proof}
		We prove the left-to-right implication by showing the contrapositive.
		We assume that there exist $\atmtnState, \atmtnState'\in\atmtnStates$ with $\atmtnState \strictPrefOrd_\dba \atmtnState'$ and $\word\in\colors^+$ such that $\atmtnUpdWord(\atmtnState, \word) = \atmtnState'$ and $\word\in\safeCycles{\atmtnState'}$.
		Let $\word_\atmtnState\in\colors^*$ be a word such that $\atmtnUpdWord(\atmtnInit, \word_\atmtnState) = \atmtnState$.
		We have that $\word_\atmtnState \strictPrefOrd \word_\atmtnState \word$, but $\word_\atmtnState \word^\omega$ is not accepted by $\dba$ as $\word$ is a cycle on $\atmtnState'$ that does not contain any B\"uchi transition.
		Hence, $\wc$ is not progress-consistent.

		For the right-to-left implication, we again prove the contrapositive.
		We assume that $\wc$ is not progress-consistent.
		Thus, there exist $\word'\in\colors^*$ and $\word\in\colors^+$ such that $\word'\strictPrefOrd \word'\word$ and $\word'\word^\omega \notin \wc$.
		Let $\atmtnState_1 = \atmtnUpdWord(\atmtnInit, \word')$ and $\atmtnState_2 = \atmtnUpdWord(\atmtnState_1, \word)$ --- we have $\atmtnState_1 \strictPrefOrd \atmtnState_2$.
		As $\atmtnState_1 \strictPrefOrd \atmtnState_2$, by Lemma~\ref{lem:monotonTrans}, we have $\atmtnUpdWord(\atmtnState_1, \word) = \atmtnState_2 \prefOrd \atmtnUpdWord(\atmtnState_2, \word)$.
		We distinguish two cases, using the fact that there is exactly one state per equivalence class of $\prefEq_\dba$.
		We represent what happens in Figure~\ref{fig:progConsComplexity}.
		\begin{itemize}
			\item If $\atmtnState_2 = \atmtnUpdWord(\atmtnState_2, \word)$, we then have that $\word\in\safeCycles{\atmtnState_2}$, and we have what we want with $\atmtnState = \atmtnState_1$ and $\atmtnState' = \atmtnState_2$.
			\item If not, we have that $\atmtnState_2 \strictPrefOrd \atmtnUpdWord(\atmtnState_2, \word)$.
			Let $\atmtnState_3 = \atmtnUpdWord(\atmtnState_2, \word)$.
			We can repeat the argument on $\atmtnState_2$ and $\atmtnState_3$: either $\word\in\safeCycles{\atmtnState_3}$ and we are done, or $\atmtnState_3 \strictPrefOrd \atmtnUpdWord(\atmtnState_3, \word)$.
			As there are finitely many states, this process necessarily ends with two states $\atmtnState = \atmtnState_n$ and $\atmtnState' = \atmtnState_{n+1}$ such that $\atmtnUpdWord(\atmtnState, \word) = \atmtnState'$ and $\word\in\safeCycles{\atmtnState'}$.
			\qedhere
		\end{itemize}
	\end{proof}
	\begin{figure}[tbh]
		\centering
		\begin{tikzpicture}[every node/.style={font=\small,inner sep=1pt}]
			\draw (0,0) node[diamant] (q1) {$\atmtnState_1$};
			\draw ($(q1)+(2,0)$) node[diamant] (q2) {$\atmtnState_2$};
			\draw ($(q2)+(2,0)$) node[diamant] (q3) {$\atmtnState_3$};
			\draw ($(q3)+(2,0)$) node[diamant,draw=none] (q4) {$\cdots$};
			\draw ($(q4)+(2,0)$) node[diamant] (qn) {$\atmtnState_n$};
			\draw ($(qn)+(2,0)$) node[diamant] (qn1) {$\atmtnState_{n+1}$};
			\draw ($(q1)!0.5!(q2)$) node[] () {$\strictPrefOrd$};
			\draw ($(q2)!0.5!(q3)$) node[] () {$\strictPrefOrd$};
			\draw ($(q3)!0.5!(q4)$) node[] () {$\strictPrefOrd$};
			\draw ($(q4)!0.5!(qn)$) node[] () {$\strictPrefOrd$};
			\draw ($(qn)!0.5!(qn1)$) node[] () {$\strictPrefOrd$};
			\draw (q1) edge[-latex',decorate,out=30,in=150] node[above=4pt] {$\word$} (q2);
			\draw (q2) edge[-latex',decorate,out=30,in=150] node[above=4pt] {$\word$} (q3);
			\draw (q3) edge[-latex',decorate,out=30,in=150] node[above=4pt] {$\word$} (q4);
			\draw (q4) edge[-latex',decorate,out=30,in=150] node[above=4pt] {$\word$} (qn);
			\draw (qn) edge[-latex',decorate,out=30,in=150] node[above=4pt] {$\word$} (qn1);
			\draw (qn1) edge[-latex',decorate,out=-30,in=30,distance=0.8cm] node[right=4pt] {$\word$} (qn1);
		\end{tikzpicture}
		\caption{Situation in the proof of Lemma~\ref{lem:progConsComplexity}.}
		\label{fig:progConsComplexity}
	\end{figure}

	Notice that for each pair of states $\atmtnState, \atmtnState'\in\atmtnStates$, the sets $\{\word\in\colors^+ \mid \atmtnUpd(\atmtnState, \word) = \atmtnState'\}$ and $\safeCycles{\atmtnState'}$ are both regular languages recognized by deterministic finite automata with at most $\card{\atmtnStates}$ states.
	The emptiness of their intersection can be decided in time $\bigO(\card{\colors}\cdot\card{\atmtnStates}^2)$ (by solving a reachability problem in the product of the two automata, using an argument similar to but easier than Lemma~\ref{lem:containmentComplexity})~\cite{RS59}.
	Thanks to Lemma~\ref{lem:progConsComplexity}, we can therefore decide whether $\dba$ is progress-consistent in time $\bigO(\card{\atmtnStates}^2\cdot(\card{\colors}\cdot\card{\atmtnStates}^2)) = \bigO(\card{\colors}\cdot\card{\atmtnStates}^4)$:
	for all $\card{\atmtnStates}^2$ pairs of states $\atmtnState, \atmtnState'\in\atmtnStates$, if $\atmtnState \strictPrefOrd \atmtnState'$, we test the emptiness of the intersection of these two regular languages.

	\paragraph{Complexity wrap-up.}
	By checking the three conditions as explained and in this order, the time complexities are respectively $\bigO(\card{\colors}\cdot\card{\atmtnStates}^4)$, $\bigO(\card{\colors}\cdot\card{\atmtnStates}^2)$, and $\bigO(\card{\colors}\cdot\card{\atmtnStates}^4)$.
	This yields a time complexity of $\bigO(\card{\colors}\cdot\card{\atmtnStates}^4)$ for the whole algorithm.

	\section{Necessity of the conditions}
	\label{sec:necessary}
	We prove that each of the three conditions introduced in Section~\ref{sec:three-conditions} is necessary for half-positionality over finite one-player arenas of DBA-recognizable objectives.
	Each condition gets its own devoted subsection.
	For the first two conditions (having a total prefix preorder and progress-consistency), we also show for completeness that they are necessary for half-positionality of general objectives over countably infinite one-player arenas, and necessary for half-positionality of $\omega$-regular objectives over finite one-player arenas.
	This distinction is worthwhile, as there exist half-positional objectives that are not $\omega$-regular (see, e.g., \emph{finitary B\"uchi} objectives~\cite{CF13}).

	We start with a well-known lemma about $\omega$-regular objectives: if two $\omega$-regular objectives are not equal, then they are distinguished by an ultimately periodic word.
	Ultimately periodic words can easily be finitely represented, and this lemma will be used throughout this section to force some behaviors to appear in \emph{finite} arenas.
	\begin{lem} \label{lem:ultPer}
		Let $\wc_1$, $\wc_2$ be two $\omega$-regular objectives.
		If $\wc_1 \neq \wc_2$, then there exist $\word_1\in\colors^*$ and $\word_2\in\colors^+$ such that either $\word_1(\word_2)^\omega \in \wc_1 \setminus \wc_2$, or $\word_1(\word_2)^\omega \in \wc_2 \setminus \wc_1$.
		In particular, if $\wc_1 \not\subseteq \wc_2$, then there exist $\word_1\in\colors^*$ and $\word_2\in\colors^+$ such that $\word_1(\word_2)^\omega \in \wc_1\setminus \wc_2$.
	\end{lem}
	\begin{proof}
		The first statement is standard and follows from McNaughton's theorem~\cite{McN66}.
		For the second statement, if $\wc_1 \not\subseteq \wc_2$, then $\wc_1 \setminus \wc_2 \neq \emptyset$.
		Objective $\wc_1 \setminus \wc_2\neq\emptyset$ is $\omega$-regular (as $\omega$-regular objectives are closed by complement and intersection) and so is $\emptyset$.
		By the first statement, we take $\word_1\in\colors^*$ and $\word_2\in\colors^+$ such that either $\word_1(\word_2)^\omega \in (\wc_1 \setminus \wc_2) \setminus \emptyset$, or $\word_1(\word_2)^\omega \in \emptyset \setminus (\wc_1 \setminus \wc_2)$.
		As $\emptyset \setminus (\wc_1 \setminus \wc_2) = \emptyset$, we have $\word_1(\word_2)^\omega \in (\wc_1 \setminus \wc_2) \setminus \emptyset = \wc_1 \setminus \wc_2$.
	\end{proof}

	\subsection{Total prefix preorder} \label{sec:totalNecessary}
	Having a total prefix preorder is necessary in general for half-positionality over countably infinite arenas, and even over finite arenas for $\omega$-regular objectives.

	\begin{prop} \label{prop:orderedNecessary}
		Let $\wc\subseteq \colors^\omega$ be an objective.
		If $\wc$ is half-positional over countably infinite one-player arenas, then its prefix preorder is total.
		If $\wc$ is $\omega$-regular and half-positional over finite one-player arenas, then its prefix preorder is total.
	\end{prop}
	\begin{proof}
		By contrapositive, we assume that the prefix preorder of $\wc$ is not total.
		Then, there exist two finite words $\word_1, \word_2\in\colors^*$ such that $\word_1 \not\prefOrd \word_2$ and $\word_2 \not\prefOrd \word_1$.
		We can find two infinite continuations $\word_1', \word_2'\in\colors^\omega$ such that $\word_1\word_1' \in \wc$, $\word_2\word_1' \notin \wc$, $\word_2\word_2' \in \wc$, and $\word_1\word_2' \notin \wc$.
		Using these four words, we build a countably infinite one-player arena depicted in Figure~\ref{fig:ordered} (left) for which $\Pone$ has no positional optimal strategy.
		Indeed, if the game started with $\word_1$, $\Pone$ needs to reply with $\word_1'$ in $\s_3$ to win, but if the game started with $\word_2$, $\Pone$ needs to reply with $\word_2'$ in $\s_3$ to win.

		Moreover, if $\wc$ is $\omega$-regular, so are $\inverse{\word_1}\wc$ and $\inverse{\word_2}\wc$.
		By Lemma~\ref{lem:ultPer}, we can therefore assume w.l.o.g.\ that $\word_1' = x_1(y_1)^\omega$ and $\word_2' = x_2(y_2)^\omega$ are ultimately periodic, and we can carry out similar arguments with the finite arena depicted in Figure~\ref{fig:ordered} (right).
	\end{proof}
	\begin{figure}[tbh]
		\centering
		\begin{minipage}{0.5\columnwidth}
			\centering
			\begin{tikzpicture}[every node/.style={font=\small,inner sep=1pt}]
				\draw (0,0) node[rond] (s3) {$\s_3$};
				\draw ($(s3)+(-2,0.8)$) node[rond] (s1) {$\s_1$};
				\draw ($(s3)+(-2,-0.8)$) node[rond] (s2) {$\s_2$};
				\draw ($(s3)+(1.8,0.7)$) node[rond,draw=none] (s4) {$\cdots$};
				\draw ($(s3)+(1.8,-0.7)$) node[rond,draw=none] (s5) {$\cdots$};
				\draw (s1) edge[-latex',decorate] node[above=4pt,xshift=1pt] {$\word_1$} (s3);
				\draw (s2) edge[-latex',decorate] node[below=4pt,xshift=1pt] {$\word_2$} (s3);
				\draw (s3) edge[-latex',decorate,bend left=20pt] node[above=4pt,xshift=-1pt] {$\word_1'$} (s4);
				\draw (s3) edge[-latex',decorate,bend right=20pt] node[below=4pt,xshift=-1pt] {$\word_2'$} (s5);
			\end{tikzpicture}
		\end{minipage}%
		\begin{minipage}{0.5\columnwidth}
			\centering
			\begin{tikzpicture}[every node/.style={font=\small,inner sep=1pt}]
				\draw (0,0) node[rond] (s3) {$\s_3$};
				\draw ($(s3)+(-2,0.8)$) node[rond] (s1) {$\s_1$};
				\draw ($(s3)+(-2,-0.8)$) node[rond] (s2) {$\s_2$};
				\draw ($(s3)+(1.8,0.8)$) node[rond] (s4) {$\s_4$};
				\draw ($(s3)+(1.8,-0.8)$) node[rond] (s5) {$\s_5$};
				\draw (s1) edge[-latex',decorate] node[above=4pt,xshift=1pt] {$\word_1$} (s3);
				\draw (s2) edge[-latex',decorate] node[below=4pt,xshift=1pt] {$\word_2$} (s3);
				\draw (s3) edge[-latex',decorate] node[above=4pt,xshift=-1pt] {$x_1$} (s4);
				\draw (s3) edge[-latex',decorate] node[below=4pt,xshift=-1pt] {$x_2$} (s5);
				\draw (s4) edge[-latex',decorate,in=30,out=-30,distance=0.8cm] node[right=4pt] {$y_1$} (s4);
				\draw (s5) edge[-latex',decorate,in=30,out=-30,distance=0.8cm] node[right=4pt] {$y_2$} (s5);
			\end{tikzpicture}
		\end{minipage}%
		\caption{Arenas in which $\Pone$ cannot play optimally with a positional strategy used in the proof of Proposition~\ref{prop:orderedNecessary}.
			The arena on the right is used for the $\omega$-regular case.}
		\label{fig:ordered}
	\end{figure}

	\subsection{Progress-consistency} \label{sec:progConsNecessary}
	Progress-consistency is necessary for half-positionality over countably infinite arenas in general, and even over finite arenas for $\omega$-regular objectives.

	\begin{prop} \label{prop:progCons}
		Let $\wc\subseteq \colors^\omega$ be an objective.
		If $\wc$ is half-positional over countably infinite one-player arenas, then it is progress-consistent.
		If $\wc$ is $\omega$-regular and half-positional even over finite one-player arenas, then it is progress-consistent.
	\end{prop}
	\begin{proof}
		By contrapositive, we assume that $\wc$ is not progress-consistent.
		Then there exist $\word_1 \in \colors^*$ and $\word_2 \in \colors^+$ such that $\word_1 \strictPrefOrd \word_1\word_2$, but $\word_1(\word_2)^\omega \notin \wc$.
		As $\word_1 \strictPrefOrd \word_1\word_2$, there exists an infinite continuation $\word'\in\colors^\omega$ such that $\word_1\word' \notin \wc$ and $\word_1\word_2\word' \in \wc$.
		Using these three words, we build a countably infinite one-player arena depicted in Figure~\ref{fig:progCons} (left).
		In this arena, from vertex $\s_1$, a positional strategy can only achieve words $\word_1(\word_2)^\omega$ or $\word_1\word'$, which are both losing.
		However, there is a (non-positional) winning strategy achieving word~$\word_1\word_2\word'$.

		If $\wc$ is additionally $\omega$-regular, using Lemma~\ref{lem:ultPer}, we can assume w.l.o.g.\ that $\word' = xy^\omega$ is ultimately periodic, and we can carry out similar arguments with the finite arena depicted in Figure~\ref{fig:progCons} (right).
	\end{proof}
	\begin{figure}[tbh]
		\centering
		\begin{minipage}{0.5\columnwidth}
			\centering
			\begin{tikzpicture}[every node/.style={font=\small,inner sep=1pt}]
				\draw (0,0) node[rond] (s1) {$\s_1$};
				\draw ($(s1)+(2,0)$) node[rond] (s2) {$\s_2$};
				\draw ($(s2)+(2,0)$) node[rond,draw=none] (s3) {$\cdots$};
				\draw (s1) edge[-latex',decorate] node[above=4pt] {$\word_1$} (s2);
				\draw (s2) edge[-latex',decorate,out=60,in=120,distance=0.8cm] node[above=4pt] {$\word_2$} (s2);
				\draw (s2) edge[-latex',decorate] node[above=4pt] {$\word'$} (s3);
			\end{tikzpicture}
		\end{minipage}%
		\begin{minipage}{0.5\columnwidth}
			\centering
			\begin{tikzpicture}[every node/.style={font=\small,inner sep=1pt}]
				\draw (0,0) node[rond] (s1) {$\s_1$};
				\draw ($(s1)+(2,0)$) node[rond] (s2) {$\s_2$};
				\draw ($(s2)+(2,0)$) node[rond] (s3) {$\s_3$};
				\draw (s1) edge[-latex',decorate] node[above=4pt] {$\word_1$} (s2);
				\draw (s2) edge[-latex',decorate,out=60,in=120,distance=0.8cm] node[above=4pt] {$\word_2$} (s2);
				\draw (s2) edge[-latex',decorate] node[above=4pt] {$x$} (s3);
				\draw (s3) edge[-latex',decorate,out=-30,in=30,distance=0.8cm] node[right=4pt] {$y$} (s3);
			\end{tikzpicture}
		\end{minipage}%
		\caption{Arenas in which $\Pone$ cannot play optimally with a positional strategy used in the proof of Proposition~\ref{prop:progCons}.
			The arena on the right is used in the $\omega$-regular case.}
		\label{fig:progCons}
	\end{figure}

	\subsection{Recognizability by the prefix-classifier} \label{sec:DBAnecessary}
	We now prove that for a DBA-recognizable objective, being recognized by a B\"uchi automaton built on top of its prefix-classifier $\minStateAtmtn$ is necessary for half-positionality.

	\begin{prop} \label{prop:necessaryDBA}
		Let $\wc \subseteq \colors^\omega$ be a DBA-recognizable objective that is half-positional over finite one-player arenas.
		Then, $\wc$ is recognized by a B\"uchi automaton built on top of~$\minStateAtmtn$.
	\end{prop}

	The rest of Section~\ref{sec:DBAnecessary} is devoted to the proof of this result, which is more involved than the proofs in Sections~\ref{sec:totalNecessary} and~\ref{sec:progConsNecessary}.
	We fix an objective $\wc\subseteq\colors^\omega$ recognized by a DBA $\dba = \dbaFull$.
	We make the assumption that $\wc$ is half-positional over finite one-player arenas.
	Our goal is to show that $\wc$ can be defined by a B\"uchi automaton built on top of $\minStateAtmtn$.
	We assume w.l.o.g.\ that $\dba$ is saturated.
	Many upcoming arguments heavily rely on this assumption through the use of Lemma~\ref{lem:safeComp} (any $\accepSet$-free word can be completed into an $\accepSet$-free cycle).

	Our proof first assumes in Section~\ref{sec:PI} that $\dba$ recognizes a prefix-independent objective.
	We will then build on this first case to conclude for the general case in Section~\ref{sec:generalCase}.
	We provide a proof sketch at the start of each subsection.

	\subsubsection{Prefix-independent case} \label{sec:PI}
	We assume that the objective $\wc$ recognized by $\dba$ is prefix-independent, so all the states of $\dba$ are equivalent for $\prefEq$.
	We want to show that $\wc$ can be recognized by a B\"uchi automaton built on top of $\minStateAtmtn$, and in this case, the automaton structure $\minStateAtmtn$ has just one state.
	Therefore, we want to find $\goodColors \subseteq \colors$ such that $\wc = \Buchi{\goodColors}$.
	We start with a high level description of the proof technique.
	\begin{proof}[Proof sketch]
		The goal is to find a suitable definition for $\goodColors$.
		To do so, we exhibit a state $\atmtnStateMax$ of $\dba$ that is ``the most rejecting state of the automaton'': it has the property that the set of $\accepSet$-free words from $\atmtnStateMax$ contains the $\accepSet$-free words from all the other states ($\atmtnStateMax$ is then called an \emph{$\accepSet$-free-maximum}) and that the set of $\accepSet$-free cycles on $\atmtnStateMax$ contains the $\accepSet$-free cycles on all the other states (it is also an \emph{$\accepSet$-free-cycle-maximum}).
		We define $\goodColors$ as the set of colors $\clr$ such that $(\atmtnStateMax, \clr)\in\accepSet$.

		We first show that if an $\accepSet$-free-maximum exists, we can assume w.l.o.g.\ that it is unique (Lemma~\ref{lem:uniqueMax}).
		In Lemmas~\ref{lem:noML},~\ref{lem:noStrat} and~\ref{ref:existenceStrongSafeMaximum}, we show the existence of an $\accepSet$-free-cycle-maximum.
		This part of the proof relies on the half-positionality over finite one-player arenas of $\wc$.
		Finally, defining $\goodColors$ using $\atmtnStateMax$ as explained above, we prove that $\wc = \Buchi{\goodColors}$ (Lemma~\ref{lem:PrefIndMeansBuchi}).
	\end{proof}

	We call a state $\atmtnStateMax\in\atmtnStates$ of $\dba$ an \emph{$\accepSet$-free-maximum} (resp.\ an \emph{$\accepSet$-free-cycle-maximum}) if for all $\atmtnState\in\atmtnStates$, we have $\safe{\atmtnState} \subseteq \safe{\atmtnStateMax}$ (resp.\ if for all $\atmtnState\in\atmtnStates$, we have $\safeCycles{\atmtnState} \subseteq \safeCycles{\atmtnStateMax}$).
	We remark that if $\dba$ is saturated, an $\accepSet$-free-cycle-maximum is also an $\accepSet$-free-maximum (this can be shown using Lemma~\ref{lem:safeComp}).

	\newcommand{\modTrans}{\ensuremath{\atmtnTrans_{\not\to \atmtnStateMaxTwo}}}
	\newcommand{\atmtnStateMaxOne}{\ensuremath{\atmtnStateMax^1}}
	\newcommand{\atmtnStateMaxTwo}{\ensuremath{\atmtnStateMax^2}}
	We first show that we can remove states from $\dba$, while still recognizing the same objective, until it has at most one $\accepSet$-free-maximum.%
	\begin{lem} \label{lem:uniqueMax}
		There exists a DBA $\dba'$ recognizing $\wc$ with at most one $\accepSet$-free-maximum.
	\end{lem}
	\begin{proof}
		Assume that $\atmtnStateMaxOne, \atmtnStateMaxTwo\in\atmtnStates$ are distinct $\accepSet$-free-maxima.
		In particular, we have that $\safe{\atmtnStateMaxOne} = \safe{\atmtnStateMaxTwo}$.
		We show that in such a situation, the objective recognized by $\dba$ can be recognized by an automaton with one less state, in which we discard one of the two $\accepSet$-free-maxima.
		To simplify the upcoming arguments, we assume that $\atmtnInit = \atmtnStateMaxOne$ (which is without loss of generality as all states of $\dba$ are equivalent for $\prefEq$).

		We define a new automaton in which we remove $\atmtnStateMaxTwo$ and redirect all its incoming transitions to $\atmtnStateMaxOne$.
		Formally, let $\dba' = (\atmtnStates', \colors, \atmtnInit', \atmtnUpd', \accepSet')$ with update function $\atmtnUpd'$ be such that
		\begin{itemize}
			\item $\atmtnStates' = \atmtnStates \setminus \{\atmtnStateMaxTwo\}$, $\atmtnInit' = \atmtnStateMaxOne$,
			\item for $\atmtnState\in\atmtnStates'$ and $\clr\in\colors$, if $\atmtnUpd(\atmtnState, \clr) = \atmtnStateMaxTwo$, then $\atmtnUpd'(\atmtnState, \clr) = \atmtnStateMaxOne$; otherwise, $\atmtnUpd'(\atmtnState, \clr) = \atmtnUpd(\atmtnState, \clr)$,
			\item for $\atmtnState\in\atmtnStates'$ and $\clr\in\colors$, $(\atmtnState, \clr)\in\accepSet'$ if and only if $(\atmtnState, \clr)\in\accepSet$.
		\end{itemize}
		We also assume that states that are not reachable from $\atmtnInit'$ in $\dba'$ are removed from $\atmtnStates'$.

		We show that this automaton with (at least) one less state recognizes the same objective as $\dba$.
		Let $\word = \clr_1\clr_2\ldots\in\colors^\omega$ be an infinite word.
		We show that $\word$ is accepted by $\dba$ if and only if it is accepted by $\dba'$.

		Let $\modTrans = \{(\atmtnState, \clr) \in \atmtnStates' \times \colors \mid \atmtnUpd(\atmtnState, \clr) = \atmtnStateMaxTwo\}$ be the transitions of $\dba'$ that were directed to $\atmtnStateMaxTwo$ in $\dba$ but are now redirected to $\atmtnStateMaxOne$ in $\dba'$.
		Let $\run$ be the infinite run of $\dba$ on $\word$, and $\run' = (\atmtnState'_0, \clr_1)(\atmtnState'_1, \clr_2)\ldots$ be the infinite run of $\dba'$ on $\word$.
		The two runs start by taking corresponding transitions, but differ once a transition in $\modTrans$ is taken.

		We first assume that $\run'$ uses transitions in $\modTrans$ only finitely many times.
		Then, there exists $k\in\IN$ such that $\atmtnState'_k = \atmtnStateMaxOne$ and for all $l \ge k$, $(\atmtnState'_{l}, \clr_{l+1}) \notin \modTrans$.
		Let $\word_{> k} = \clr_{k+1}\clr_{k+2}\ldots$ be the infinite word consisting of the colors taken after the last occurrence of a transition in $\modTrans$.
		We have that
		\begin{align*}
			\text{$\dba$ accepts $\word$}
			&\Longleftrightarrow \text{$\dba$ accepts $\word_{> k}$}
			&&\text{as $\dba$ recognizes a prefix-independent objective}\\
			&\Longleftrightarrow \text{$\dba'$ accepts $\word_{> k}$}
			&&\text{as $\word_{> k}$ visits exactly the same transitions as in $\dba$}\\
			&\Longleftrightarrow \text{$\dba'$ accepts $\word$}
			&&\text{as $\clr_1\ldots\clr_k$ is a cycle on the initial state $\atmtnStateMaxOne$ of $\dba'$.}
		\end{align*}

		We now assume that $\run'$ uses transitions in $\modTrans$ infinitely many times.
		We decompose $\run'$ into infinitely many finite runs $\run'_1, \run'_2, \ldots$ such that $\run' = \run'_1\run'_2\ldots$ and every run $\run'_i$ contains exactly one transition in $\modTrans$ as its last transition.
		This implies that all these finite runs start in state $\atmtnStateMaxOne$.
		We represent run $\run'$ in Figure~\ref{fig:run'}.
		We define words $\word_1, \word_2, \ldots$ as the respective projection of runs $\run'_1, \run'_2, \ldots$ to their colors (we have $\word = \word_1\word_2\ldots$).
		Notice that
		\begin{align}
			\forall i \ge 1,
			\word_i\in\safe{\atmtnStateMaxOne} \Longleftrightarrow \word_i \in \safep{\atmtnStateMaxOne}
			\label{eq:safe},
		\end{align}
		as the transitions used by $\word_i$ from $\atmtnStateMaxOne$ in $\dba'$ correspond to the transitions used by $\word_i$ from $\atmtnStateMaxOne$ in $\dba$ (this property is not true for all words, but this is true for these words that read only one transition in $\modTrans$ as their last transition).
		We also have by construction that
		\begin{align}
			\forall i \ge 1,
			\atmtnUpdWord(\atmtnStateMaxOne, \word_i) = \atmtnStateMaxTwo.
			\label{eq:q1toq2}
		\end{align}
		We distinguish whether $\word$ is accepted or rejected by $\dba'$.

		Assume $\word$ is accepted by $\dba'$.
		Then, we know that for infinitely many $i\in\IN$, $\word_i\notin \safep{\atmtnStateMaxOne}$.
		This implies that for these indices $i$, $\word_i\notin \safe{\atmtnStateMaxOne}$ by~\eqref{eq:safe}.
		As $\atmtnStateMaxOne$ is an $\accepSet$-free-maximum, for all $\atmtnState\in\atmtnStates$, $\word_i\notin\safe{\atmtnState}$ (this is simply the contrapositive of the definition of $\accepSet$-free-maximum).
		Hence, for infinitely many $i\in\IN$, when $\word_i$ is read in $\dba$ (no matter from where), a B\"uchi transition is used, so $\word$ is accepted by $\dba$.

		Assume $\word$ is rejected by $\dba'$.
		Then there exists $k\in\IN$ such that for all $l \ge k$, $\word_l \in \safep{\atmtnStateMaxOne}$.
		As $\dba$ is prefix-independent, up to removing the start of $\word$, we assume w.l.o.g.\ that $k = 1$.
		We show by induction that
		\[
		\forall i\ge 1,
		\safe{\atmtnStateMaxOne} = \safe{\atmtnUpdWord(\atmtnStateMaxOne, \word_1\ldots\word_{i})}.
		\]
		This is true for $i = 1$, as $\atmtnUpdWord(\atmtnStateMaxOne, \word_1) = \atmtnStateMaxTwo$ by~\eqref{eq:q1toq2} and the fact that $\atmtnStateMaxOne$ and $\atmtnStateMaxTwo$ are both $\accepSet$-free-maxima.
		Assume $\safe{\atmtnStateMaxOne} = \safe{\atmtnUpdWord(\atmtnStateMaxOne, \word_1\ldots\word_{i-1})}$ for some $i \ge 2$.
		Then, by Lemma~\ref{lem:safeCongruence}, as $\word_i\in\safe{\atmtnStateMaxOne}$, we have $\safe{\atmtnUpdWord(\atmtnStateMaxOne, \word_i)} = \safe{\atmtnUpdWord(\atmtnStateMaxOne, \word_1\ldots\word_{i-1}\word_i)}$.
		Property~\eqref{eq:q1toq2} gives $\safe{\atmtnUpdWord(\atmtnStateMaxOne, \word_i)} = \safe{\atmtnStateMaxTwo}$, which is itself equal to $\safe{\atmtnStateMaxOne}$.
		We now know that for all $i\ge 1$, $\word_i\in\safe{\atmtnStateMaxOne}$ by~\eqref{eq:safe}.
		Therefore, we conclude that for all $i\ge 1$, $\word_i\in\safe{\atmtnUpdWord(\atmtnStateMaxOne, \word_1\ldots\word_{i-1})}$.
		In particular, $\word$ uses no B\"uchi transition when read from $\atmtnStateMaxOne$ in $\dba$ and is also rejected by $\dba$.

		We have shown that $\dba'$ is a DBA with fewer states than $\dba$ recognizing $\wc$.
		If $\dba'$ still has two or more $\accepSet$-free-maxima, we repeat our construction until there is at most one left.
	\end{proof}
	\begin{figure}[tbh]
		\centering
		\begin{tikzpicture}[every node/.style={font=\small,inner sep=1pt}]
			\draw (0,0) node[] (run') {$\run'\colon$};
			\draw ($(run')+(0.5,0)$) node[] (s1) {};
			\draw ($(s1)+(2,0)$) node[] (s2) {};
			\draw ($(s2)+(2,0)$) node[] (s3) {};
			\draw ($(s3)+(2,0)$) node[] (s4) {};
			\draw ($(s4)+(0.5,0)$) node[] (s5) {\hspace{2pt}$\cdots$};
			\draw ($(s1)+(0,0.5)$) node[] (q1) {$\atmtnStateMaxOne$};
			\draw ($(s2)+(0,0.5)$) node[] (q2) {$\atmtnStateMaxOne$};
			\draw ($(s3)+(0,0.5)$) node[] (q3) {$\atmtnStateMaxOne$};
			\draw ($(s4)+(0,0.5)$) node[] (q4) {$\atmtnStateMaxOne$};
			\draw ($(s2)+(-0.2,-0.5)$) node[] (q5) {$\modTrans$};
			\draw ($(s3)+(-0.2,-0.5)$) node[] (q6) {$\modTrans$};
			\draw ($(s4)+(-0.2,-0.5)$) node[] (q7) {$\modTrans$};
			\draw (s1) edge[-latex',decorate] node[above=4pt] {$\word_1$} (s2);
			\draw (s2) edge[-latex',decorate] node[above=4pt] {$\word_2$} (s3);
			\draw (s3) edge[-latex',decorate] node[above=4pt] {$\word_3$} (s4);
			\draw (q1) edge[dashed] (s1);
			\draw (q2) edge[dashed] (s2);
			\draw (q3) edge[dashed] (s3);
			\draw (q4) edge[dashed] (s4);
			\draw (q5) edge[dashed] ($(s2)+(-0.2,0)$);
			\draw (q6) edge[dashed] ($(s3)+(-0.2,0)$);
			\draw (q7) edge[dashed] ($(s4)+(-0.2,0)$);
		\end{tikzpicture}
		\caption{Features of run $\run'$ when it takes infinitely many transitions in $\modTrans$.}
		\label{fig:run'}
	\end{figure}

	Thanks to Lemma~\ref{lem:uniqueMax}, we now assume w.l.o.g.\ that $\dba$ has at most one $\accepSet$-free-maximum.
	We intend to show that there exists an $\accepSet$-free-cycle-maximum.
	To do so, we exhibit an (infinite) arena in which $\player{1}$ has no winning strategy, which we prove by using half-positionality of $\wc$ over finite one-player arenas.
	We then prove that the non-existence of an $\accepSet$-free-cycle-maximum would imply that $\player{1}$ has a winning strategy in this arena (Lemma~\ref{ref:existenceStrongSafeMaximum}).

	Let $\arena_\dba$ be the infinite one-player arena of $\Pone$ depicted in Figure~\ref{fig:arenaPI}.
	This arena consists of one vertex $\s$ with a choice to make among all non-empty words that are $\accepSet$-free cycles from some state of $\dba$.
	Vertex $\s$ is the only vertex with multiple outgoing edges.
	The goal of the next three short lemmas is to show that in this arena, $\Pone$ has no winning strategy.%
	\begin{figure}[tbh]
		\centering
		\begin{tikzpicture}[every node/.style={font=\small,inner sep=1pt}]
			\draw (0,0) node[rond, scale=1.4] (s) {$\s$};
			\draw (s) edge[-latex',decorate,in=170,out=-140,distance=0.75cm] node[left=4pt] {$w_1$} (s);
			\draw (s) edge[-latex',decorate,in=100,out=150,distance=0.75cm] node[above=4pt] {$w_2$} (s);
			\draw (s) edge[-latex',decorate,in=-10,out=40,distance=0.75cm] node[right=4pt] {$w_n$} (s);
			\draw (0.3,0.7) node[] (SafeCycles) {$\cdots$};
			\draw (0.6,-0.5) node[] (SafeCycles) {$\cdots$};
			\draw (5,0) node[] (SafeCycles) {$\{w_1,w_2,\dots \}= \bigcup_{\atmtnState\in\atmtnStates} \safeCycles{\atmtnState} \setminus \{\emptyWord\}$};
		\end{tikzpicture}%
		\caption{Infinite one-player arena $\arena_\dba$ of $\Pone$, with choices from $\s$ among every non-empty word in $\bigcup_{\atmtnState\in\atmtnStates} \safeCycles{\atmtnState}$.}%
		\label{fig:arenaPI}
	\end{figure}%

	\begin{lem} \label{lem:strategy-implies-positional}
		If $\Pone$ has a winning strategy in $\arena_\dba$, then $\Pone$ has a positional winning strategy.
	\end{lem}
	\begin{proof}
		Suppose that there is a winning strategy of $\Pone$ in $\arena_\dba$.
		Let $\word = \word_1\word_2\ldots$ be an infinite winning word such that for $i \ge 1$, $\word_i \in \bigcup_{\atmtnState\in\atmtnStates} \safeCycles{\atmtnState} \setminus \{\emptyWord\}$.
		Let $\atmtnState_0 = \atmtnInit$, and for $i \ge 1$, let $\atmtnState_i = \atmtnUpdWord(\atmtnInit, \word_1\ldots\word_i)$ be the current automaton state after reading the first $i$ finite words composing $\word$.
		As there are only finitely many automaton states and $\word$ is winning, there are $k, l\ge 1$ with $k < l$ such that $\atmtnState_k = \atmtnState_l$ and $\word_{k+1}\ldots\word_l \notin \safe{\atmtnState_k}$.
		Word $\word_1\ldots\word_k(\word_{k+1}\ldots\word_l)^\omega$ is also a winning word and uses only finitely many different words in~$\bigcup_{\atmtnState\in\atmtnStates} \safeCycles{\atmtnState} \setminus \{\emptyWord\}$.

		Hence, there is a finite restriction (``subarena'') $\arena'_\dba$ of the arena $\arena_\dba$ with at most~$l$ choices in $\s$ in which $\Pone$ has a winning strategy.
		Arena $\arena'_\dba$ being finite and one-player, half-positionality of $\wc$ over finite one-player arenas implies that $\Pone$ has a positional winning strategy in $\arena'_\dba$.
		This positional winning strategy can also be played in $\arena_\dba$ (as every choice available in $\arena_\dba'$ is also available in~$\arena_\dba$).
	\end{proof}

	\begin{lem} \label{lem:noML}
		No positional strategy of $\Pone$ is winning in $\arena_\dba$.
	\end{lem}
	\begin{proof}
		Any positional strategy of $\Pone$ generates a word $\word^\omega$, where $\word\in\safeCycles{\atmtnState}\setminus\{\emptyWord\}$ for some $\atmtnState\in\atmtnStates$.
		In particular, word $\word^\omega$ is rejected when it is read from state $\atmtnState$.
		As all the states in $\atmtnStates$ are equivalent for $\prefEq$ (as we assume that $\wc$ is prefix-independent), we have $\atmtnState \prefEq \atmtnInit$, so $\word^\omega$ is also rejected when read from the initial state $\atmtnInit$ of the automaton.
	\end{proof}

	Using Lemmas~\ref{lem:strategy-implies-positional} and~\ref{lem:noML}, we deduce the desired result.
	\begin{lem}\label{lem:noStrat}
		No strategy of $\player{1}$ is winning in $\arena_\dba$.
	\end{lem}

	We use the statement of Lemma~\ref{lem:noStrat} to show the existence of an $\accepSet$-free-cycle-maximum.

	\begin{lem} \label{ref:existenceStrongSafeMaximum}
		There exists a unique $\accepSet$-free-cycle-maximum in $\dba$.
	\end{lem}
	\begin{proof}
		We first show that there exists an $\accepSet$-free-cycle-maximum.
		Let us assume by contradiction that there is no $\accepSet$-free-cycle-maximum.
		We show how to build a winning strategy of $\Pone$ in $\arena_\dba$, contradicting Lemma~\ref{lem:noStrat}.
		To do so, we build an infinite word accepted by $\dba$ by combining finite words that are $\accepSet$-free cycles from some state.

		We claim that for every $\atmtnState\in\atmtnStates$ which is not an $\accepSet$-free-maximum, there exists
		\[
		\text{ $\word_\atmtnState\in\bigcup_{\atmtnState'\in\atmtnStates} \safeCycles{\atmtnState'} \setminus \{\emptyWord\}$ such that $\word_\atmtnState \notin \safe{\atmtnState}$.}
		\]
		Let $\atmtnState'\in\atmtnStates$ be such that $\safe{\atmtnState'} \not\subseteq \safe{\atmtnState}$, which exists as $\atmtnState$ is not an $\accepSet$-free-maximum.
		Let $\word_1\in \safe{\atmtnState'} \setminus \safe{\atmtnState}$.
		By Lemma~\ref{lem:safeComp}, there exists $\word_2\in\colors^*$ such that $\word_1\word_2\in\safeCycles{\atmtnState'}$ (this holds as we have assumed w.l.o.g.\ that $\dba$ is saturated).
		As $\word_1\notin\safe{\atmtnState}$, we also have $\word_1\word_2\notin\safe{\atmtnState}$.
		Taking $\word_\atmtnState = \word_1\word_2$ proves the claim.
		For $\atmtnState\in\atmtnStates$ not an $\accepSet$-free-maximum, we fix $\word_\atmtnState\in\colors^+$ such that $\word_\atmtnState\in\bigcup_{\atmtnState'\in\atmtnStates} \safeCycles{\atmtnState'} \setminus \{\emptyWord\}$ and $\word_\atmtnState \notin \safe{\atmtnState}$.

		Thanks to Lemma~\ref{lem:uniqueMax}, we can assume that there is at most one $\accepSet$-free-maximum in $\dba$.
		If it exists, let $\atmtnStateMax \in \atmtnStates$ be a unique $\accepSet$-free-maximum.
		Using our initial assumption, $\atmtnStateMax$ is not an $\accepSet$-free-cycle-maximum: there is $\atmtnState\in\atmtnStates$ such that $\safeCycles{\atmtnState} \not\subseteq \safeCycles{\atmtnStateMax}$.
		Let $\word_{\mathsf{max}}\in \safeCycles{\atmtnState} \setminus \safeCycles{\atmtnStateMax}$.
		Notice that as $\word_{\mathsf{max}}\in\safe{\atmtnState}$ and $\atmtnStateMax$ is an $\accepSet$-free-maximum, $\word_{\mathsf{max}}\in\safe{\atmtnStateMax}$.
		Therefore, $\word_{\mathsf{max}}$ cannot be a cycle on $\atmtnStateMax$, i.e., $\atmtnUpdWord(\atmtnStateMax, \word_\mathsf{max}) \neq \atmtnStateMax$.

		We build iteratively an infinite winning word that can be played by $\Pone$ in $\arena_\dba$ (our construction works whether there is an $\accepSet$-free-maximum or not).
		As $\Pone$ plays, we keep track in parallel of the current automaton state.
		The game starts in $\s$, with current automaton state $\atmtnState_0 = \atmtnInit$.
		Let $n \ge 0$.
		We distinguish two cases.
		\begin{itemize}
			\item If $\atmtnState_n$ is not an $\accepSet$-free-maximum, then $\Pone$ plays word $\word_{\atmtnState_n}$.
			As $\word_{\atmtnState_n} \notin\safe{\atmtnState_n}$, a B\"uchi transition is used along the way.
			The current automaton state becomes $\atmtnState_{n+1} = \atmtnUpdWord(\atmtnState_n, \word_{\atmtnState_n})$.
			\item If $\atmtnState_n = \atmtnStateMax$ is the $\accepSet$-free-maximum, then $\Pone$ plays $\word_\mathsf{max}$.
			The current automaton state becomes $\atmtnState_{n+1} = \atmtnUpd(\atmtnStateMax, \word_{\mathsf{max}})$, which is not equal to $\atmtnStateMax$.
		\end{itemize}
		For infinitely many $i\ge 0$, the automaton state $\atmtnState_i$ is not an $\accepSet$-free-maximum (as there is at most one $\accepSet$-free-maximum in $\dba$, and it cannot appear twice in a row).
		Therefore, we have described a winning strategy for $\Pone$, since the corresponding run over $\dba$ visits infinitely often a B\"uchi transition.

		We have obtained the existence of an $\accepSet$-free-cycle-maximum by contradiction.
		Using Lemma~\ref{lem:safeComp}, we can show that an $\accepSet$-free-cycle-maximum is also an $\accepSet$-free-maximum.
		As we have assumed there is at most one $\accepSet$-free-maximum (through Lemma~\ref{lem:uniqueMax}), we also obtain the uniqueness of the $\accepSet$-free-cycle-maximum.
	\end{proof}

	We now know that there exists a unique $\accepSet$-free-cycle-maximum $\atmtnStateMax\in\atmtnStates$.
	We show how to use the outgoing transitions of $\atmtnStateMax$ in order to realize $\wc$ as an objective of the kind $\Buchi{\goodColors}$ for some $\goodColors \subseteq \colors$, which is the goal of the current subsection.
	\begin{lem} \label{lem:PrefIndMeansBuchi}
		There exists $\goodColors\subseteq\colors$ such that $\wc = \Buchi{\goodColors}$.
	\end{lem}
	\begin{proof}
		Let $\goodColors = \{\clr\in\colors \mid (\atmtnStateMax, \clr) \in \accepSet\}$ --- equivalently, if we consider colors as words with one letter, $\goodColors$ is the set of colors $\clr$ such that $\clr\notin\safe{\atmtnStateMax}$.

		We first show that $\Buchi{\goodColors} \subseteq \wc$.
		Let $\clr\in\goodColors$.
		Then, $\clr\notin\safe{\atmtnStateMax}$.
		As $\atmtnStateMax$ is an $\accepSet$-free-maximum, for all $\atmtnState\in\atmtnStates$, $\clr\notin\safe{\atmtnState}$.
		Therefore, any word containing infinitely many occurrences of colors in $\goodColors$ uses infinitely many B\"uchi transitions and is accepted by $\dba$.

		We now show that $\wc \subseteq \Buchi{\goodColors}$.
		By contrapositive, let $\word = \clr_1\clr_2\ldots\notin\Buchi{\goodColors}$ be an infinite word with only finitely many colors in $\goodColors$.
		We show that $\word \notin \wc$.
		As $\wc$ is prefix-independent, we may assume w.l.o.g.\ that $\word$ has no color in $\goodColors$, i.e., that for all $i\ge 1$, $\clr_i \in \colors \setminus \goodColors$.
		We claim that when read from $\atmtnStateMax$, word $\word$ uses no B\"uchi transition and is thus rejected.
		This implies that $\word\notin\wc$ as $\atmtnStateMax \prefEq \atmtnInit$.

		Assume by contradiction that there is some B\"uchi transition when reading $\word$ from $\atmtnStateMax$, i.e., there exists $k \ge 0$ such that for $\word_{\le k} = \clr_1\ldots\clr_k$, $\word_{\le k}\in\safe{\atmtnStateMax}$, but $\word_{\le k}\clr_{k+1}\notin\safe{\atmtnStateMax}$.
		We will deduce that $(\atmtnStateMax, \clr_{k+1})$ is a B\"uchi transition, contradicting that~$\clr_{k+1} \in \colors \setminus \goodColors$.

		We depict the situation in Figure~\ref{fig:wkn}.
		Let $\atmtnState_1 = \atmtnUpdWord(\atmtnStateMax, \word_{\le k})$ (whether $\atmtnState_1$ equals $\atmtnStateMax$ or not does not matter).
		By Lemma~\ref{lem:safeComp}, there exists $\word'\in\colors^*$ such that $\word_{\le k}\word'\in\safeCycles{\atmtnStateMax}$.
		By construction, we have $\word'\word_{\le k} \in \safeCycles{\atmtnState_1}$.
		As $\atmtnStateMax$ is an $\accepSet$-free-cycle-maximum, we have $\safeCycles{\atmtnState_1} \subseteq \safeCycles{\atmtnStateMax}$, so we also have $\word'\word_{\le k} \in \safeCycles{\atmtnStateMax}$.
		Let $\atmtnState_2 = \atmtnUpdWord(\atmtnStateMax, \word')$.
		Notice that $\atmtnStateMax = \atmtnUpdWord(\atmtnState_2, \word_{\le k})$ and $\word_{\le k} \in \safe{\atmtnState_2}$.
		As $\atmtnStateMax$ is an $\accepSet$-free-maximum and $\word_{\le k}\clr_{k+1}\notin\safe{\atmtnStateMax}$, we also have that $\word_{\le k}\clr_{k+1}\notin\safe{\atmtnState_2}$.
		Therefore, transition $(\atmtnStateMax, \clr_{k+1})$ must be a B\"uchi transition, which contradicts that $\clr_{k+1} \in \colors \setminus \goodColors$.
	\end{proof}
	\begin{figure}[tbh]
		\centering
		\begin{tikzpicture}[every node/.style={font=\small,inner sep=1pt}]
			\draw (0,0) node[diamant] (q) {$\atmtnStateMax$};
			\draw ($(q)+(2.5,0)$) node[diamant] (q1) {$\atmtnState_1$};
			\draw ($(q)-(2.5,0)$) node[diamant] (q2) {$\atmtnState_2$};
			\draw (q) edge[-latex',out=30,in=180-30,decorate,distance=0.8cm] node[above=4pt] {$\word_{\le k}$} (q1);
			\draw (q1) edge[-latex',out=180+30,in=-30,decorate,distance=0.8cm] node[below=4pt] {$\word'$} (q);
			\draw (q1) edge[-latex',accepting] node[above=4pt] {$\clr_{k+1}$} ($(q1)+(1.5,0)$);
			\draw (q) edge[-latex',out=180-30,in=30,decorate,distance=0.8cm] node[above=4pt] {$\word'$} (q2);
			\draw (q2) edge[-latex',out=-30,in=180+30,decorate,distance=0.8cm] node[below=4pt] {$\word_{\le k}$} (q);
			\draw (q) edge[-latex',accepting] node[left=2pt] {$\clr_{k+1}$} ($(q)+(0,1.5)$);
		\end{tikzpicture}
		\caption{Situation in the proof of Lemma~\ref{lem:PrefIndMeansBuchi}, with $\word_{\le k}\in\safe{\atmtnStateMax}$ but $\word_{\le k}\clr_{k+1}\notin\safe{\atmtnStateMax}$.}
		\label{fig:wkn}
	\end{figure}

	\subsubsection{General case} \label{sec:generalCase}
	We now remove the prefix-independence assumption on $\wc$.
	We still assume that $\wc$ is half-positional over finite one-player arenas, and show that $\wc$ can be recognized by a B\"uchi automaton built on top of~$\minStateAtmtn$.
	If $\dba$ has exactly one state per equivalence class of $\prefEq$, it means that it is built on top of $\minStateAtmtn$, and we are done.
	If not, let $\qEquiv\in\atmtnStates$ be a state such that~$\card{\eqClass{\qEquiv}} \ge 2$.

	We briefly sketch the proof of this section.
	\begin{proof}[Proof sketch]
		Our proof will show how to modify $\dba$ by ``merging'' all states in equivalence class $\eqClass{\qEquiv}$ into a single state, while still recognizing the same objective $\wc$.
		The main technical argument is to build a variant $\wc_{\eqClass{\qEquiv}}$ of objective $\wc$ on a new set of colors $\alphProof_{\eqClass{\qEquiv}}$, that turns out to also be half-positional over finite one-player arenas and DBA-recognizable, but which is \emph{prefix-independent}.
		We can therefore use Lemma~\ref{lem:PrefIndMeansBuchi} from Section~\ref{sec:PI} and find $\alphProofGood \subseteq \alphProof_{\eqClass{\qEquiv}}$ such that $\wc_{\eqClass{\qEquiv}} = \Buchi{\alphProofGood}$.
		Then, we exhibit a state $\atmtnStateMax \in \eqClass{\qEquiv}$ whose $\accepSet$-free words are tightly linked to the elements of $\alphProofGood$ (Lemma~\ref{lem:generalCaseAccepting} and Corollary~\ref{cor:generalCase}).
		Finally, akin to the way we removed an $\accepSet$-free-maximum in Lemma~\ref{lem:uniqueMax}, we show that it is still possible to recognize $\wc$ while keeping only state $\atmtnStateMax$ in $\eqClass{\qEquiv}$ (Lemma~\ref{lem:B'sameLanguage}).

		Once we know how to merge the equivalence class $\eqClass{\qEquiv}$ into a single state, we can simply repeat the operation for each equivalence class with multiple states, until we obtain a DBA built on top of $\prefClass$.
	\end{proof}

	We define a new set of colors $\alphProof_{\eqClass{\qEquiv}}$ using finite words in $\colors^+$ such that
	\[
	\alphProof_{\eqClass{\qEquiv}} = \{ \word = \clr_1\ldots\clr_i\in\colors^+ \mid \atmtnUpdWord(\qEquiv, \word) \prefEq \qEquiv \land \forall j\ \text{s.t.}\ 1 \le j < i,   \atmtnUpdWord(\qEquiv, \clr_1\ldots\clr_j) \not\prefEq \qEquiv\}.
	\]
	This (possibly infinite) set contains all the finite words that, read from $\qEquiv$, come back to a state in $\eqClass{\qEquiv}$ for the first time after being read.
	By Lemma~\ref{lem:alwaysEquiv}, for all $\atmtnState \in \eqClass{\qEquiv}$, for all $\word\in\alphProof_{\eqClass{\qEquiv}}$, we also have that $\atmtnUpdWord(\atmtnState, \word) \prefEq \atmtnUpdWord(\qEquiv, \word) \prefEq \qEquiv$.
	The set $\alphProof_{\eqClass{\qEquiv}}$ therefore corresponds to the set of words with the seemingly stronger property that, when read from any state in $\eqClass{\qEquiv}$, come back to a state in $\eqClass{\qEquiv}$ for the first time.
	We define an objective $\wc_{\eqClass{\qEquiv}}$ of infinite words on this new set of colors such that
	\[
	\wc_{\eqClass{\qEquiv}} = \{\word_1\word_2\ldots \in \alphProof_{\eqClass{\qEquiv}}^\omega \mid \word_1\word_2\ldots \in \atmtnLang{\dba^{\qEquiv}}\}.
	\]
	We show that $\wc_{\eqClass{\qEquiv}}$ has the three conditions allowing us to apply Lemma~\ref{lem:PrefIndMeansBuchi} to it.
	\begin{itemize}
		\item Objective $\wc_{\eqClass{\qEquiv}}$ is DBA-recognizable: we consider the DBA $\dba_{\eqClass{\qEquiv}} = (\eqClass{\qEquiv}, \alphProof_{\eqClass{\qEquiv}}, \qEquiv, \atmtnUpd', \accepSet')$, whose update function $\atmtnUpd'$ is the restriction of $\atmtnUpdWord$ to $\eqClass{\qEquiv} \times \alphProof_{\eqClass{\qEquiv}}$, and $\accepSet' = \{(\atmtnState, \word) \in \eqClass{\qEquiv} \times \alphProof_{\eqClass{\qEquiv}} \mid \word \notin \safe{\atmtnState}\}$.
		\item Objective $\wc_{\eqClass{\qEquiv}}$ is prefix-independent, as adding or removing a finite number of cycles on $\qEquiv$ does not affect the accepted status of a word in $\atmtnLang{\dba^{\qEquiv}}$.
		\item Half-positionality of $\wc_{\eqClass{\qEquiv}}$ over finite one-player arenas is implied by half-positionality of $\wc$ over finite one-player arenas.
		Indeed, every (finite one-player) arena $\arena_{\eqClass{\qEquiv}}$ using colors in $\alphProof_{\eqClass{\qEquiv}}$ can be transformed into a (finite one-player) arena $\arena$ with similar properties using colors in $\colors$.
		Two transformations are applied: $(i)$ we replace every $\alphProof_{\eqClass{\qEquiv}}$-colored edge in $\arena_{\eqClass{\qEquiv}}$ by a corresponding finite chain of $\colors$-colored edges, and $(ii)$ for every vertex $\s$ of $\arena_{\eqClass{\qEquiv}}$, we prefix it with a finite chain of $\colors$-colored edges starting from a vertex $\s'$ reading a word $\word_{\qEquiv}$ such that $\atmtnUpdWord(\atmtnInit, \word_{\qEquiv}) = \qEquiv$.
		We then have that $\Pone$ has a winning strategy from a vertex $\s$ in $\arena_{\eqClass{\qEquiv}}$ if and only if $\Pone$ has a winning strategy from $\s'$ in $\arena$, and a positional winning strategy from $\s'$ in $\arena$ can be transformed into a positional winning strategy from $\s$ in $\arena_{\eqClass{\qEquiv}}$.
	\end{itemize}
	By Lemma~\ref{lem:PrefIndMeansBuchi}, there exists a set $\alphProofGood \subseteq \alphProof_{\eqClass{\qEquiv}}$ such that $\wc_{\eqClass{\qEquiv}} = \Buchi{\alphProofGood}$.

	We now show links between $\accepSet$-free words from states of $\eqClass{\qEquiv}$ and the words in $\alphProofGood$.
	The arguments once again rely on the saturation of $\dba$.
	\begin{lem} \label{lem:generalCaseAccepting}\hfill
		\begin{itemize}
			\item Let $\word\in\alphProof_{\eqClass{\qEquiv}}$.
			If $\word\in\alphProofGood$, then for all $\atmtnState\in\eqClass{\qEquiv}$, $\word \notin \safe{\atmtnState}$.
			\item There exists $\atmtnState\in\eqClass{\qEquiv}$ such that, for all $\word\in\alphProof_{\eqClass{\qEquiv}} \setminus \alphProofGood$, $\word\in\safe{\atmtnState}$.
		\end{itemize}
	\end{lem}
	\begin{proof}
		For the first statement, we prove the contrapositive.
		We assume that there exists a state $\atmtnState\in\eqClass{\qEquiv}$ such that $\word\in\safe{\atmtnState}$.
		By Lemma~\ref{lem:safeComp}, there is $\word'\in\colors^*$ such that $\word\word'\in\safeCycles{\atmtnState}$.
		In particular, $(\word\word')^\omega \notin \atmtnLang{\dba^{\atmtnState}} = \atmtnLang{\dba^{\qEquiv}}$.
		We can assume w.l.o.g.\ that $\word'\in\colors^+$ (if $\word' = \emptyWord$, then we can simply replace it with $\word' = \word$).
		As $\word'$ is read from a state in $\eqClass{\qEquiv}$ to a state in $\eqClass{\qEquiv}$, it can be written uniquely as $\word' = \word_1'\ldots\word_n'$ for some words $\word_i'\in\alphProof_{\eqClass{\qEquiv}}$.
		Therefore, $(\word\word_1'\ldots\word_n')^\omega$ is also an infinite word on $\alphProof_{\eqClass{\qEquiv}}$, and we have $(\word\word_1'\ldots\word_n')^\omega \notin \wc_{\eqClass{\qEquiv}}$ since $(\word\word_1'\ldots\word_n')^\omega \notin \atmtnLang{\dba^{\qEquiv}}$.
		As $\wc_{\eqClass{\qEquiv}} = \Buchi{\alphProofGood}$, clearly $\word\notin\alphProofGood$, which ends the proof of the first statement.

		For the second statement, assume by contradiction that for all $\atmtnState\in\eqClass{\qEquiv}$, there exists $\word_{\atmtnState}\in\alphProof_{\eqClass{\qEquiv}} \setminus \alphProofGood$ such that $\word_{\atmtnState}\notin\safe{\atmtnState}$.
		As there are only finitely many states in $\eqClass{\qEquiv}$, it is then possible to build a word $\word = \word_{\atmtnState_1}\ldots\word_{\atmtnState_n}$ such that for all $1 \le i < n$, $\atmtnUpdWord(\atmtnState_i, \word_{\atmtnState_i}) = \atmtnState_{i+1}$, $\atmtnUpdWord(\atmtnState_n, \word_{\atmtnState_n}) = \atmtnState_1$, and for all $1 \le i \le n$, $\word_{\atmtnState_i} \in \alphProof_{\eqClass{\qEquiv}} \setminus \alphProofGood$ and $\word_{\atmtnState_i} \notin \safe{\atmtnState_i}$.
		Word $\word^\omega$ is accepted from $\atmtnState_1$ as it uses infinitely many B\"uchi transitions, so it is in $\atmtnLang{\dba^{\atmtnState_1}} = \atmtnLang{\dba^{\qEquiv}}$.
		However, if we consider $(\word_{\atmtnState_1}\ldots\word_{\atmtnState_n})^\omega$ as an infinite word on $\alphProof_{\eqClass{\qEquiv}}$, then it is not in $\wc_{\eqClass{\qEquiv}} = \Buchi{\alphProofGood}$ as every letter of the word is in $\alphProof_{\eqClass{\qEquiv}} \setminus \alphProofGood$.
		This yields a contradiction.
	\end{proof}

	We use the previous result in a straightforward way to exhibit a state $\atmtnStateMax$ whose non-$\accepSet$-free words in $\alphProof_{\eqClass{\qEquiv}}$ are exactly the words in $\alphProofGood$.
	The reader may notice that, echoing the proof of the prefix-independent case (Section~\ref{sec:PI}), the state $\atmtnStateMax$ given by Corollary~\ref{cor:generalCase} is actually an $\accepSet$-free-maximum among states in $\eqClass{\qEquiv}$.
	\begin{cor} \label{cor:generalCase}
		There exists $\atmtnStateMax\in\eqClass{\qEquiv}$ such that for all $\word\in\alphProof_{\eqClass{\qEquiv}}$, $\word\in\alphProofGood$ if and only if $\word\notin\safe{\atmtnStateMax}$.
	\end{cor}
	\begin{proof}
		By the second item of Lemma~\ref{lem:generalCaseAccepting}, we take $\atmtnStateMax\in\eqClass{\qEquiv}$ such that, for all $\word\in\alphProof_{\eqClass{\qEquiv}} \setminus \alphProofGood$, $\word\in\safe{\atmtnStateMax}$.
		Let $\word\in\alphProof_{\eqClass{\qEquiv}}$.
		The property we already have on $\atmtnStateMax$ gives us by contrapositive that for $\word\in\alphProof_{\eqClass{\qEquiv}}$, $\word\notin\safe{\atmtnStateMax}$ implies that $\word\in \alphProofGood$.
		Reciprocally, the first item of Lemma~\ref{lem:generalCaseAccepting} gives us that if $\word\in\alphProofGood$, then $\word\notin\safe{\atmtnStateMax}$.
	\end{proof}

	From now on, we assume that $\atmtnStateMax\in\eqClass{\qEquiv}$ is a state having the property of Corollary~\ref{cor:generalCase}.
	We show that $\wc$ can be recognized by a smaller DBA consisting of DBA $\dba$ in which all the states in $\eqClass{\qEquiv}$ have been merged into the single state $\atmtnStateMax$, by redirecting all incoming transitions of $\eqClass{\qEquiv}$ to $\atmtnStateMax$.
	We assume w.l.o.g.\ that if $\atmtnInit \in \eqClass{\qEquiv}$, then $\atmtnInit = \atmtnStateMax$ (this does not change the objective recognized by $\dba$, and will be convenient in the upcoming construction).
	We consider DBA $\dba' = (\atmtnStates', \colors, \atmtnInit', \atmtnUpd', \accepSet')$ with
	\begin{itemize}
		\item $\atmtnStates' = (\atmtnStates \setminus \eqClass{\qEquiv}) \cup \{\atmtnStateMax\}$, $\atmtnInit' = \atmtnInit$,
		\item for $\atmtnState\in\atmtnStates'$ and $\clr\in\colors$, if $\atmtnUpd(\atmtnState, \clr) \in \eqClass{\qEquiv}$, then $\atmtnUpd'(\atmtnState, \clr) = \atmtnStateMax$; otherwise, $\atmtnUpd'(\atmtnState, \clr) = \atmtnUpd(\atmtnState, \clr)$,
		\item for $\atmtnState\in\atmtnStates'$ and $\clr\in\colors$, $(\atmtnState, \clr)\in\accepSet'$ if and only if $(\atmtnState, \clr)\in\accepSet$.
	\end{itemize}
	We also assume that states that are not reachable from $\atmtnInit'$ in $\dba'$ are removed from $\atmtnStates'$.

	\begin{lem} \label{lem:B'sameLanguage}
		The objective recognized by $\dba'$ is also $\wc$.
	\end{lem}
	\begin{proof}
		Let $\word=\clr_1\clr_2\ldots\in\colors^\omega$, $\run = (\atmtnState_0, \clr_1, \atmtnState_1)(\atmtnState_1, \clr_2, \atmtnState_2)\ldots$ be the infinite run of $\dba$ on $\word$, and $\run' = (\atmtnState'_0, \clr_1, \atmtnState'_1)(\atmtnState'_1, \clr_2, \atmtnState'_2)\ldots$ be the infinite run of $\dba'$ on $\word$.
		We have that $\atmtnState_0 = \atmtnState'_0$, but states in both runs may not coincide after a state in $\eqClass{\qEquiv}$ has been reached.
		Yet, we show inductively that
		\begin{align}
			\forall i \ge 0,
			\atmtnState_i \prefEq_\dba \atmtnState'_i.
			\label{eq:equiv}
		\end{align}
		It is true for $i = 0$ (we even have equality in this case), and if $\atmtnState_n \prefEq_\dba \atmtnState'_n$, then by construction of the transitions of $\dba'$ and Lemma~\ref{lem:alwaysEquiv}, we still have $\atmtnState_{n+1} \prefEq_\dba \atmtnState'_{n+1}$.

		We want to show that $\word$ is accepted by $\dba$ if and only if it is accepted by $\dba'$.
		This is clear if $\word$ never goes through a state in $\eqClass{\qEquiv}$ (as the same transitions are then taken in $\dba$ and $\dba'$).

		We first assume that run $\run$ visits $\eqClass{\qEquiv}$ finitely many times, and that the last visit to $\eqClass{\qEquiv}$ happens in $\atmtnState_n$ for some $n \ge 0$.
		Notice that run $\run'$ also visits $\eqClass{\qEquiv}$ for the last time in $\atmtnState'_n$ by~\eqref{eq:equiv}.
		Therefore, word $\clr_{n+1}\clr_{n+2}\ldots$ is accepted from $\atmtnState'_n$ in $\dba'$ if and only if it is accepted from $\atmtnState'_n$ in $\dba$: all the subsequent transitions coincide.
		As $\atmtnState_n \prefEq_\dba \atmtnState'_n$, we have moreover that $\clr_{n+1}\clr_{n+2}\ldots$ is accepted from $\atmtnState'_n$ in $\dba'$ if and only if it is accepted from $\atmtnState_n$ in $\dba$.
		This implies that $\word$ is accepted by $\dba$ if and only if it is accepted by $\dba'$.

		We now assume that run $\run$ visits $\eqClass{\qEquiv}$ infinitely many times.
		We decompose $\word$ inductively into a prefix $\word_{\qEquiv}$ reaching $\eqClass{\qEquiv}$ followed by cycles $\word_1, \word_2, \ldots$ on $\eqClass{\qEquiv}$.
		Formally, let $\word_{\qEquiv}$ be any finite prefix of $\word$ such that $\atmtnUpdWord(\atmtnInit, \word_{\qEquiv})\in\eqClass{\qEquiv}$.
		We then proceed by induction: we assume that $\word_{\qEquiv}\word_1\ldots\word_n$ is a prefix of $\word$ such that $\atmtnUpdWord(\atmtnInit, \word_{\qEquiv}\word_1\ldots\word_n)\in\eqClass{\qEquiv}$.
		We define $\word_{n+1}\in\colors^+$ as the shortest non-empty word such that
		$\word_{\qEquiv}\word_1\ldots\word_n\word_{n+1}$ is a prefix of $\word$ and $\atmtnUpdWord(\atmtnInit, \word_{\qEquiv}\word_1\ldots\word_n\word_{n+1})\in\eqClass{\qEquiv}$.
		In particular, $\word_{n+1}\in\alphProof_{\eqClass{\qEquiv}}$.
		We have $\word = \word_{\qEquiv}\word_1\word_2\ldots$ by construction.
		By~\eqref{eq:equiv}, we also have that for all $n \ge 0$, $(\atmtnUpd')^*(\atmtnInit', \word_\atmtnState\word_1\ldots\word_n) \prefEq_\dba \qEquiv$, so $(\atmtnUpd')^*(\atmtnInit', \word_\atmtnState\word_1\ldots\word_n) = \atmtnStateMax$ as this is the only state left in that class in $\dba'$.

		If $\word$ is accepted by $\dba$, then for infinitely many $i \ge 1$, word $\word_i$ is in $\alphProofGood$.
		By Corollary~\ref{cor:generalCase}, all these infinitely many words are not in $\safe{\atmtnStateMax}$ and therefore use a B\"uchi transition when read from $\atmtnStateMax$, so $\word$ is also accepted by $\dba'$.

		If $\word$ is rejected by $\dba$, then there exists $n\ge 1$ such that for all $n' \ge n$, $\word_{n'}\in \alphProof_{\eqClass{\qEquiv}} \setminus \alphProofGood$.
		By Corollary~\ref{cor:generalCase}, for all $n' \ge n$, $\word_{n'}$ is in $\safe{\atmtnStateMax}$ and therefore does not use a B\"uchi transition when read from $\atmtnStateMax$, so $\word$ is also rejected by $\dba'$.
	\end{proof}

	We have all the arguments to show our goal for the section (Proposition~\ref{prop:necessaryDBA}), that is, to show that $\wc$ can be recognized by a B\"uchi automaton built on top of $\prefClass$.
	\begin{proof}[Proof of Proposition~\ref{prop:necessaryDBA}]
		We have shown in Lemma~\ref{lem:B'sameLanguage} how to merge an equivalence class of $\dba$ into a single state, while still recognizing the same objective.
		Repeating this construction for each equivalence class with two or more states, we end up with a DBA with exactly one state per equivalence class of $\prefEq$ still recognizing $\wc$.
		By definition of $\prefClass$, this DBA is necessarily built on top of (some automaton isomorphic to) $\prefClass$.
	\end{proof}

	\section{Sufficiency of the conditions} \label{sec:sufficient}
	We show that a DBA $\dba$ with the three conditions from Section~\ref{sec:three-conditions} (recognizing a progress-consistent objective having a total prefix preorder and being recognizable by a B\"uchi automaton built on top of $\minStateAtmtn$) recognizes a half-positional objective.
	As these three conditions have been shown to be necessary for the half-positionality of objectives recognized by a DBA, this will imply a characterization of half-positionality.

	Our main technical tool is to construct, thanks to these three conditions, a family of \emph{completely well-monotonic universal graphs}.
	The existence of such objects implies, thanks to recent results~\cite{Ohl23}, that $\Pone$ has positional optimal strategies, even in two-player arenas of arbitrary cardinality.

	\subsection{Completely well-monotonic universal graphs}
	We fix extra terminology about graphs only used in Section~\ref{sec:sufficient}, and recall the relevant results from~\cite{Ohl23}.

	\paragraph{Extra preliminaries on graphs.}
	Let $\gr = \grFull$ be a graph and $\wc\subseteq\colors^\omega$ be an objective.
	A vertex $\grSt$ of $\gr$ \emph{satisfies $\wc$} if for all infinite paths $\grSt_0 \grTr{\clr_1} \grSt_1 \grTr{\clr_2} \ldots$ from $\grSt$, we have $\clr_1\clr_2\ldots\in\wc$.

	Given two graphs $\gr = \grFull$ and $\gr' = (\grStates', \grEdges')$, a \emph{(graph) morphism from $\gr$ to $\gr'$} is a function $\morph\colon \grStates \to \grStates'$ such that $(\grSt_1, \clr, \grSt_2)\in\grEdges$ implies $(\morph(\grSt_1), \clr, \morph(\grSt_2))\in\grEdges'$.

	A morphism $\morph$ from $\gr$ to $\gr'$ is \emph{$\wc$-preserving} if for all $\grSt\in\grStates$, $\grSt$ satisfies $\wc$ implies that $\morph(\grSt)$ satisfies $\wc$.
	Notice that if $\morph(\grSt)$ satisfies $\wc$, then $\grSt$ satisfies $\wc$, as any path $\grSt \grTr{\clr_1} \grSt_1 \grTr{\clr_2} \ldots$ of $\gr$ implies the existence of a path $\morph(\grSt) \grTr{\clr_1} \morph(\grSt_1) \grTr{\clr_2} \ldots$ of $\gr'$ --- there are ``more paths'' in $\gr'$.
	For $\kappa$ a cardinal, a graph $\univGr$ is \emph{$(\cardinal, \wc)$-universal} if for all graphs $\gr$ of cardinality $< \cardinal$, there is a $\wc$-preserving morphism from $\gr$ to $\univGr$.

	We consider a graph $\gr = \grFull$ along with a total order $\le$ on its vertex set $\grStates$.
	We say that $\gr$ is \emph{monotonic} if for all $\grSt, \grSt', \grSt''\in\grStates$, for all $\clr\in\colors$,
	\begin{itemize}
		\item $(\grSt \grTr{\clr} \grSt'\ \text{and}\ \grSt' \ge \grSt'') \implies \grSt \grTr{\clr} \grSt''$, and
		\item $(\grSt \ge \grSt'\ \text{and}\ \grSt' \grTr{\clr} \grSt'') \implies \grSt \grTr{\clr} \grSt''$.
	\end{itemize}
	This means that $(i)$ whenever there is an edge $\grSt \grTr{\clr} \grSt'$, there is also an edge with color $\clr$ from $\grSt$ to all vertices smaller than $\grSt'$ for $\le$, and $(ii)$ whenever $\grSt \ge \grSt'$, then $\grSt$ has at least the same outgoing edges as $\grSt'$.
	Graph $\gr$ is \emph{well-monotonic} if it is monotonic and the total order $\le$ is a well-order (i.e., any set of vertices has a minimum).
	Graph $\gr$ is \emph{completely well-monotonic} if it is well-monotonic and there exists a vertex $\top\in\grStates$ maximum for $\le$ such that for all $\grSt\in\grStates$, $\clr\in\colors$, $\top \grTr{\clr} \grSt$.

	\begin{exa}
		We provide an example illustrating these notions with $\colors = \{a, b\}$ and $\wc = \Buchi{\{a\}}$.
		This is a special case of our upcoming construction, and it is already discussed in more depth in~\cite[Section~4]{Ohl23}.
		For $\mainOrd$ an ordinal, we define a graph $\univGr_\mainOrd$ with vertices $\univGrStates_\mainOrd = \mainOrd \cup \{\top\}$, and such that for all ordinals $\ord, \ord' < \mainOrd$,
		\begin{itemize}
			\item $\ord \grTr{a} \ord'$, and
			\item $\ord \grTr{b} \ord'$ if and only if $\ord' < \ord$.
		\end{itemize}
		Moreover, for all $\grSt \in \univGrStates_\mainOrd$ (including the case $\grSt = \top$), we define edges $\top \grTr{a} \grSt$ and $\top \grTr{b} \grSt$.
		We order vertices using the natural order on the ordinals, and with $\ord < \top$ for all $\ord < \mainOrd$.

		Vertex $\top$ does not satisfy $\wc$, as there is an infinite path $\top \grTr{b} \top \grTr{b} \ldots$, and $b^\omega\notin\wc$.
		All other vertices satisfy $\wc$ by construction: they cannot reach $\top$, and as reading $b$ decreases the current ordinal, a path cannot have an infinite suffix using only color $b$ (there is no infinite decreasing sequence of ordinals).

		This graph is $(\cardinal, \wc)$-universal for $\cardinal = \card{\mainOrd}$.
		Intuitively, for any graph $\gr$ with less than $\cardinal$ vertices, a $\wc$-preserving morphism from $\gr$ to $\univGr_\mainOrd$ can be defined by mapping vertices not satisfying $\wc$ to $\top$, and vertices satisfying $\wc$ to an ordinal $\ord$ that depends on how long it may take to guarantee reading an $a$ from them.

		Graph $\univGr_\mainOrd$ is monotonic, which can be quickly checked with the definition.
		As $\le$ is a well-order and there is a vertex $\top$ with the right properties, $\univGr_\mainOrd$ is even completely well-monotonic.
		The properties of $\univGr_\mainOrd$ imply half-positionality of $\Buchi{\{a\}}$, thanks to the \linebreak[1] following theorem.
		\qedEx
	\end{exa}

	We state an important result linking half-positionality and completely well-monotonic universal graphs from~\cite{Ohl23}.
	\begin{thmC}[{(One direction of~\cite[Theorem~3.1]{Ohl23})}] \label{thm:univGr}
		Let $\wc\subseteq \colors^\omega$ be an objective.
		If for all cardinals $\cardinal$, there exists a completely well-monotonic $(\cardinal, \wc)$-universal graph, then $\wc$ is half-positional (over all arenas).
	\end{thmC}
	A more precise result from Ohlmann's thesis~\cite[Theorem~1.1]{OhlmannThesis} can actually be instantiated on more precise classes of arenas.
	However, we use it to prove here half-positionality of a family of objectives over \emph{all} arenas, so the above result turns out to be sufficient.

	\newcommand{\neutralLetter}{\ensuremath{e}}
	\begin{rem}\label{rem:epsilon-positional}
		Our approach also implies the half-positionality of $\wc$ over arenas with $\epsilon$-edges (cf.\ Remark~\ref{rem:eps-transitions-definition}).
		Indeed, we can add a fresh color $\neutralLetter$ to $\colors$ and define an objective $\wc_\neutralLetter\subseteq (\colors\cup \{\neutralLetter\})^\omega$ such that for $w\in (\colors\cup \{\neutralLetter\})^\omega$,
		\[ w\in \wc_\neutralLetter \; \Longleftrightarrow \; \left\{ \begin{array}{c}
			\text{the word obtained by removing $\neutralLetter$ is infinite and belongs to $\wc$, or}\\
			w = w'\neutralLetter^\omega \text{ for some } w'\in (\colors\cup \{\neutralLetter\})^*.
		\end{array}\right.\]
		Then, the existence of a family of completely well-monotonic $(\cardinal, \wc)$-universal graphs for $\wc$ implies that $\wc_\neutralLetter$ is half-positional~\cite{Ohl23} (in the terminology of~\cite{Ohl23}, $\neutralLetter$ is a \emph{strongly neutral letter}).
		Therefore, we can label the $\epsilon$-edges of any arena with $\neutralLetter$ and obtain an equivalent game with objective~$\wc_\neutralLetter$.
		\qedEx
	\end{rem}

	\subsection{Universal graphs for B\"uchi automata}
	We show that for a DBA-recognizable objective, the three conditions that were shown to be necessary for half-positionality in Section~\ref{sec:necessary} are actually sufficient.

	\begin{prop} \label{prop:sufficient}
		Let $\wc\subseteq\colors^\omega$ be an objective that has a total prefix preorder, is progress-consistent, and is recognizable by a B\"uchi automaton built on top of $\prefClass$.
		Then, $\wc$ is half-positional.
	\end{prop}

	The rest of the section is devoted to the proof of this result, using Theorem~\ref{thm:univGr}.
	Let $\wc\subseteq\colors^\omega$ be an objective with a total prefix preorder, that is progress-consistent, and that is recognized by a DBA $\dba = \dbaFull$ built on top of $\minStateAtmtn$ for the rest of this section.
	We assume as in previous sections that $\dba$ is saturated (in particular, by Lemma~\ref{lem:safeComp}, any $\accepSet$-free path can be extended to an $\accepSet$-free cycle).
	As $\dba$ is built on top of $\minStateAtmtn$, for $\atmtnState, \atmtnState'\in \atmtnStates$, $\atmtnState \prefEq \atmtnState'$ if and only if~$\atmtnState = \atmtnState'$.

	For $\mainOrd$ an ordinal, we build a graph $\univGrB$ in the following way.
	\begin{itemize}
		\item We set the vertices as $\univGrStatesB = \{(\atmtnState, \ord) \mid \atmtnState\in\atmtnStates, \ord < \mainOrd\} \cup \{\top\}$.
		\item For every transition $\atmtnUpd(\atmtnState, \clr) = \atmtnState'$ of $\dba$,
		\begin{itemize}
			\item if $(\atmtnState, \clr) \in \accepSet$, then for all ordinals $\ord$, $\ord'$, we define an edge $(\atmtnState, \ord) \grTr\clr (\atmtnState', \ord')$;
			\item if $(\atmtnState, \clr) \notin \accepSet$, then for all ordinals $\ord$, $\ord'$ s.t.\ $\ord' < \ord$, we define an edge $(\atmtnState, \ord) \grTr\clr (\atmtnState', \ord')$.
			\item for $\atmtnState'' \strictPrefOrd \atmtnState'$, then for all ordinals $\ord$, $\ord''$, we define an edge $(\atmtnState, \ord) \grTr\clr (\atmtnState'', \ord'')$.
		\end{itemize}
		\item For all $\clr\in\colors$, $\grSt\in\univGrStatesB$, we define an edge $\top \grTr\clr \grSt$.
	\end{itemize}
	We order the vertices lexicographically: $(\atmtnState, \ord) \le (\atmtnState', \ord')$ if $\atmtnState \strictPrefOrd \atmtnState'$ or ($\atmtnState = \atmtnState'$ and $\ord \le \ord'$), and we define $\top$ as the maximum for $\le$ ($(\atmtnState, \ord) < \top$ for all $\atmtnState\in\atmtnStates$, $\ord<\mainOrd$).

	Graph $\univGrB$ is built such that on the one hand, it is sufficiently large and has sufficiently many edges so that there is a morphism from any graph $\gr$ (of cardinality smaller than some function of $\card{\mainOrd}$) to $\univGrB$.
	On the other hand, for the morphism to be $\wc$-preserving, at least some vertices of $\univGrB$ need to satisfy $\wc$, which imposes a restriction on the infinite paths from vertices.
	Graph $\univGrB$ is actually built so that for any automaton state $\atmtnState\in \atmtnStates$ and ordinal $\ord < \mainOrd$, the vertex $(\atmtnState,\ord)$ satisfies $\atmtnLang{\dba^{\atmtnState}}$ (see Lemma~\ref{lem:satisfies}).
	The intuitive idea is that for a non-B\"uchi transition $(\atmtnState, \clr) \notin \accepSet$ of the automaton such that $\atmtnUpd(\atmtnState, \clr) = \atmtnState'$, a $\clr$-colored edge from a vertex $(\atmtnState, \ord)$ in the graph either $(i)$ reaches a vertex with first component $\atmtnState'$, in which case the ordinal must decrease on the second component, or $(ii)$ reaches a vertex with first component $\atmtnState'' \strictPrefOrd \atmtnState'$, with no restriction on the second component, but therefore with fewer winning continuations.
	Using progress-consistency and the fact that there is no infinitely decreasing sequence of ordinals, we can show that this implies that no infinite path in $\univGrB$ corresponds to an infinite run in the automaton visiting only non-B\"uchi transitions.

	We state two properties that directly follow from the definition of $\univGrB$:
	\begin{align}
		&\text{if $(\atmtnState, \ord) \grTr{\clr} (\atmtnState', \ord')$, then $\atmtnState' \prefOrd \atmtnUpd(\atmtnState, \clr)$;} \label{eq:graphProp1}
		\\
		&\text{if $(\atmtnState, \ord) \grTr{\clr} (\atmtnState', \ord')$ and $\ord'' \le \ord'$, then $(\atmtnState, \ord) \grTr{\clr} (\atmtnState', \ord'')$}. \label{eq:graphProp2}
	\end{align}

	\begin{exa} \label{ex:univGrExample}
		We consider again the DBA $\dba$ from Example~\ref{ex:AAorBuchiA}, recognizing the words containing infinitely many occurrences of $a$, or $a$ twice in a row at some point.
		We represent the graph $\univGrB$, with $\mainOrd = \omega$ in Figure~\ref{fig:univGrExample}.
		\qedEx
	\end{exa}
	\begin{figure}[htb]
		\centering
		\begin{tikzpicture}[every node/.style={font=\small}]
			\draw (0,0) node[matrix,carre] (qinit) {
				\draw (0,0) node[rond,font=\scriptsize,minimum height=3mm] (qinit1) {$0$};
				\draw ($(qinit1)+(0.8,0)$) node[rond,font=\scriptsize,minimum height=3mm] () {$1$};
				\draw ($(qinit1)+(1.6,0)$) node[font=\scriptsize] () {$\cdots$};
				\draw ($(qinit1)+(-0.8,0.05)$) node[] () {$\atmtnInit$};
				\\
			};
			\draw ($(qinit)+(4,0)$) node[matrix,carre] (qa) {
				\draw (0,0) node[rond,font=\scriptsize,minimum height=3mm] (qa1) {$0$};
				\draw ($(qa1)+(0.8,0)$) node[rond,font=\scriptsize,minimum height=3mm] () {$1$};
				\draw ($(qa1)+(1.6,0)$) node[font=\scriptsize] () {$\cdots$};
				\draw ($(qa1)+(-0.8,0.05)$) node[] () {$\atmtnState_a$};
				\\
			};
			\draw ($(qa)+(4,0)$) node[matrix,carre] (qaa) {
				\draw (0,0) node[rond,font=\scriptsize,minimum height=3mm] (qaa1) {$0$};
				\draw ($(qaa1)+(0.8,0)$) node[rond,font=\scriptsize,minimum height=3mm] () {$1$};
				\draw ($(qaa1)+(1.6,0)$) node[font=\scriptsize] () {$\cdots$};
				\draw ($(qaa1)+(-0.8,0.05)$) node[] () {$\atmtnState_{aa}$};
				\\
			};
			\draw ($(qaa)+(3.5,0)$) node[rond] (top) {$\top$};

			\draw (qinit) edge[-latex',out=30,in=150,very thick] node[above=2pt] {$a$} (qa);
			\draw (qinit) edge[-latex',out=120,in=60,distance=0.8cm] node[above=2pt] {$a$} (qinit);
			\draw (qa) edge[-latex',out=120,in=60,distance=0.8cm] node[above=2pt] {$a$} (qa);
			\draw (qa) edge[-latex',out=30,in=150,very thick] node[above=2pt] {$a$} (qaa);
			\draw (qaa) edge[-latex',out=120,in=60,distance=0.8cm,very thick] node[above=2pt] {$a,b$} (qaa);
			\draw (top) edge[-latex',out=120,in=60,distance=0.8cm] node[above=2pt] {$a,b$} (top);

			\draw (qinit) edge[-latex',out=-60,in=-120,distance=0.8cm,dashed,very thick] node[below=2pt] {$b$} (qinit);
			\draw (qa) edge[-latex',out=175,in=5] node[above=2pt] {$a$} (qinit);
			\draw (qa) edge[-latex',out=185,in=-5,very thick] node[below=2pt] {$b$} (qinit);

			\draw (qaa) edge[-latex'] node[above=2pt] {$a, b$} (qa);
			\draw (qaa) edge[-latex',out=210,in=-30,distance=1.34cm] node[below=2pt] {$a, b$} (qinit);

			\draw (top) edge[-latex'] node[above=2pt] {$a,b$} ($(top)+(-1.2,0.2)$);
			\draw (top) edge[-latex'] node[below=2pt] {$a,b$} ($(top)+(-1.2,-0.2)$);
		\end{tikzpicture}
		\caption{The graph $\univGr_{\dba, \omega}$, where $\dba$ is the automaton from Example~\ref{ex:AAorBuchiA} ($\atmtnLang{\dba}=\Buchi{\{a\}}\cup \colors^*aa\colors^\omega$).
			The dashed edge with color $b$ indicates that $(\atmtnInit, \ord) \grTr{b} (\atmtnInit, \ord')$ if and only if $\ord' < \ord$ (it corresponds to the only non-B\"uchi transition of~$\dba$).
			Elsewhere, an edge labeled with color $\clr$ between two rectangles $\atmtnState, \atmtnState'$ means that for all natural numbers (i.e., ordinals less than $\omega$) $\ord, \ord'$, $(\atmtnState, \ord) \grTr{\clr} (\atmtnState', \ord')$.
			Thick edges correspond to the original transitions of $\dba$.
			There are edges from $\top$ to all vertices of the graph with colors $a$ and $b$.
			Vertices are totally ordered from left to right.}
		\label{fig:univGrExample}
	\end{figure}

	In order to use Theorem~\ref{thm:univGr}, we show that the graph $\univGrB$ is completely well-monotonic (Lemma~\ref{lem:monotonic}) and, for all cardinals $\cardinal$, is $(\cardinal, \wc)$-universal for sufficiently large $\mainOrd$ (Proposition~\ref{prop:universality}).

	\begin{lem} \label{lem:monotonic}
		Graph $\univGrB$ is completely well-monotonic.
	\end{lem}
	\begin{proof}
		The order $\le$ on the vertices is a well-order, and there exists a vertex $\top\in\univGrStatesB$ maximum for $\le$ such that for all $\grSt\in\univGrStatesB$, $\clr\in\colors$, $\top \grTr{\clr} \grSt$.
		To show that $\univGrB$ is completely well-monotonic, it now suffices to show that $\univGrB$ is monotonic.

		The first item of the monotonicity definition follows from the construction of the graph.
		We assume that $(\atmtnState, \ord) \grTr{\clr} (\atmtnState', \ord')$ and $(\atmtnState', \ord') \ge (\atmtnState'', \ord'')$, and we show that $(\atmtnState, \ord) \grTr{\clr} (\atmtnState'', \ord'')$.
		By~\eqref{eq:graphProp1}, we have $\atmtnState' \prefOrd \atmtnUpd(\atmtnState, \clr)$.
		The inequality $(\atmtnState', \ord') \ge (\atmtnState'', \ord'')$ means by definition that $\atmtnState'' \strictPrefOrd \atmtnState'$ or ($\atmtnState'' = \atmtnState'$ and $\ord'' \le \ord'$).
		If $\atmtnState'' \strictPrefOrd \atmtnState'$, we obtain that $\atmtnState'' \strictPrefOrd \atmtnUpd(\atmtnState, \clr)$, so by construction of the graph (third case for transitions) we also have $(\atmtnState, \ord) \grTr{\clr} (\atmtnState'', \ord'')$.
		If $\atmtnState'' = \atmtnState'$ and $\ord'' \le \ord'$, then we also have $(\atmtnState, \ord) \grTr{\clr} (\atmtnState', \ord'')$ by~\eqref{eq:graphProp2}, which is what we want as $\atmtnState' = \atmtnState''$.

		The second item of the monotonicity definition is slightly more involved and follows from progress-consistency, the fact that the prefix preorder is total, and the saturation of $\dba$.
		We assume that
		$(\atmtnState, \ord) \ge (\atmtnState', \ord')$ and $(\atmtnState', \ord') \grTr{\clr} (\atmtnState'', \ord'')$,
		and we show that $(\atmtnState, \ord) \grTr{\clr} (\atmtnState'', \ord'')$.
		The assumption $(\atmtnState, \ord) \ge (\atmtnState', \ord')$ implies that $\atmtnState \invPrefOrd \atmtnState'$, so $\atmtnUpd(\atmtnState, \clr) \invPrefOrd \atmtnUpd(\atmtnState', \clr)$ by Lemma~\ref{lem:monotonTrans}.
		Moreover, $(\atmtnState', \ord') \grTr{\clr} (\atmtnState'', \ord'')$ implies that $\atmtnUpd(\atmtnState', \clr) \invPrefOrd \atmtnState''$ by~\eqref{eq:graphProp1}.
		Hence, $\atmtnUpd(\atmtnState, \clr) \invPrefOrd \atmtnState''$.
		If $\atmtnUpd(\atmtnState, \clr) \strictInvPrefOrd \atmtnState''$, then $(\atmtnState, \ord) \grTr{\clr} (\atmtnState'', \ord'')$ by definition of the graph.
		The same holds if $\atmtnUpd(\atmtnState, \clr) = \atmtnState''$ and $(\atmtnState, \clr)\in\accepSet$.

		It is left to discuss the case $\atmtnUpd(\atmtnState, \clr) = \atmtnState''$ and $(\atmtnState, \clr)\notin\accepSet$.
		By the above inequalities, this implies that we also have $\atmtnUpd(\atmtnState', \clr) = \atmtnState''$.
		\begin{itemize}
			\item If $\atmtnState = \atmtnState'$, then $\ord \ge \ord'$.
			Moreover, as $(\atmtnState', \clr) = (\atmtnState, \clr) \notin \accepSet$, the existence of edge $(\atmtnState', \ord') \grTr{\clr} (\atmtnState'', \ord'')$ implies that $\ord' > \ord''$.
			So $\ord > \ord''$ and we also have $(\atmtnState, \ord) \grTr{\clr} (\atmtnState'', \ord'')$.
			\item We show that $\atmtnState' \strictPrefOrd \atmtnState$ is not possible --- we assume it holds and draw a contradiction.
			As $(\atmtnState, \clr)\notin\accepSet$, we have $\clr\in\safe{\atmtnState}$.
			By Lemma~\ref{lem:safeComp}, there is $\word\in\colors^*$ such that $\clr\word\in\safeCycles{\atmtnState}$.
			As $\atmtnUpd(\atmtnState, \clr) = \atmtnUpd(\atmtnState', \clr)$ and $\atmtnUpdWord(\atmtnState, \clr\word) = \atmtnState$, we have $\atmtnUpdWord(\atmtnState', \clr\word) = \atmtnState$.
			As $\atmtnState' \strictPrefOrd \atmtnState$, by progress-consistency, the word $(\clr\word)^\omega$ must be accepted by $\dba$ when read from $\atmtnState'$.
			It must therefore also be accepted when read from $\atmtnState$ (as $\atmtnState' \strictPrefOrd \atmtnState$), which contradicts that $\clr\word\in\safeCycles{\atmtnState}$.
			\qedhere
		\end{itemize}
	\end{proof}

	We now intend to show, for all cardinals $\cardinal$, $(\cardinal, \wc)$-universality of some $\univGrB$ with a sufficiently large $\mainOrd$.
	Lemmas~\ref{lem:underapprox} and~\ref{lem:winningCycle} give insight into properties of the paths of $\univGrB$, to then establish which vertices of $\univGrB$ satisfy $\wc$ (Lemma~\ref{lem:satisfies}).
	Understanding which vertices satisfy $\wc$ is useful to later define a $\wc$-preserving morphism into $\univGrB$.
	We first show that paths in this graph ``underapproximate'' corresponding runs in the automaton: a finite path $\finGrPth = (\atmtnState_0, \ord_0)\grTr{\clr_1}\ldots\grTr{\clr_n}(\atmtnState_n, \ord_n)$ in $\univGrB$ visits vertices labeled with automaton states at most as large (for $\prefOrd$) as the corresponding states visited by the finite run from $\atmtnState_0$ on $\clr_1\ldots\clr_n$.

	\begin{lem} \label{lem:underapprox}
		Let $\finGrPth = (\atmtnState_0, \ord_0)\grTr{\clr_1}\ldots\grTr{\clr_n}(\atmtnState_n, \ord_n)$ be a finite path of $\univGrB$ and $\word = \clr_1\ldots\clr_n$.
		Then, $\atmtnState_n \prefOrd \atmtnUpdWord(\atmtnState_0, \word)$.
	\end{lem}
	\begin{proof}
		We proceed by induction on the length $n$ of $\finGrPth$.
		If $n = 0$, then $\word = \emptyWord$, so $\atmtnState_n = \atmtnState_0 = \atmtnUpdWord(\atmtnState_0, \word)$, and the result holds.
		If $n \ge 1$, we assume by induction hypothesis that $\atmtnState_{n-1} \prefOrd \atmtnUpdWord(\atmtnState_0, \clr_1\ldots\clr_{n-1})$.
		By Lemma~\ref{lem:monotonTrans}, we have $\atmtnUpd(\atmtnState_{n-1}, \clr_n) \prefOrd \atmtnUpd(\atmtnUpdWord(\atmtnState_0, \clr_1\ldots\clr_{n-1}), \clr_n) = \atmtnUpdWord(\atmtnState_0, \word)$.
		By~\eqref{eq:graphProp1}, if $(\atmtnState_{n-1}, \ord_{n-1}) \grTr{\clr_n}(\atmtnState_n, \ord_n)$, then $\atmtnState_n \prefOrd \atmtnUpd(\atmtnState_{n-1}, \clr_n)$.
		By transitivity, $\atmtnState_n \prefOrd \atmtnUpdWord(\atmtnState_0, \word)$.
	\end{proof}

	We now show that in $\univGrB$, a finite path that goes back to its initial value w.r.t.\ the first component without decreasing the ordinal necessarily induces an accepted word when repeated.
	\begin{lem} \label{lem:winningCycle}
		Let $\finGrPth = (\atmtnState_0, \ord_0)\grTr{\clr_1}\ldots\grTr{\clr_n}(\atmtnState_n, \ord_n)$ be a finite path of $\univGrB$ with $n \ge 1$, $\atmtnState_0 = \atmtnState_n$, and $\ord_0 \le \ord_n$.
		Let $\word = \clr_1\ldots\clr_n$.
		Then, $\word^\omega \in \atmtnLang{\dba^{\atmtnState_0}}$.
	\end{lem}
	\begin{proof}
		Let $\run = \dbaRun{\atmtnState_0, \word} = (\atmtnState_0', \clr_1, \atmtnState_1')\ldots(\atmtnState_{n-1}', \clr_n, \atmtnState_n')$ be the finite run of $\dba$ obtained by reading $\word$ from $\atmtnState_0$.
		States $\atmtnState_0, \ldots, \atmtnState_n$ correspond to the first component of the vertices visited by $\finGrPth$ in $\univGrB$, whereas $\atmtnState'_0,\ldots,\atmtnState'_n$ are the states visited by the finite word $\word$ in $\dba$.
		We have that $\atmtnState_0 = \atmtnState_0'$, but the subsequent states may or may not correspond.
		By Lemma~\ref{lem:underapprox}, we still know that for all $i$, $0 \le i \le n$, $\atmtnState_i \prefOrd \atmtnState_i'$.
		In particular, $\atmtnState_0 = \atmtnState_n \prefOrd \atmtnUpdWord(\atmtnState_0, \word)= \atmtnState_n'$.
		We distinguish three cases, depending on whether $\atmtnState_0 \strictPrefOrd \atmtnUpdWord(\atmtnState_0, \word)$ or $\atmtnState_0 \prefEq \atmtnUpdWord(\atmtnState_0, \word)$ (which implies $\atmtnState_0 = \atmtnUpdWord(\atmtnState_0, \word)$ as $\dba$ is built on top of its prefix-classifier), and depending on whether $\atmtnState_i = \atmtnState'_i$ for all $0 \le i \le n$ or not.
		\begin{itemize}
			\item If $\atmtnState_0 \strictPrefOrd \atmtnUpdWord(\atmtnState_0, \word)$, by progress-consistency, $\word^\omega \in \atmtnLang{\dba^{\atmtnState_0}}$.
			\item If $\atmtnState_0 = \atmtnUpdWord(\atmtnState_0, \word)$ and for all $i$, $0 \le i \le n$, $\atmtnState_i = \atmtnState_i'$ (i.e., $\finGrPth$ only uses edges that directly correspond to transitions of the automaton $\dba$), then for the ordinal to be greater than or equal to its starting value, some B\"uchi transition has to be taken since non-B\"uchi transitions strictly decrease the ordinal on the second component.
			Hence, $\word$ is not an $\accepSet$-free cycle from $\atmtnState_0$, so $\word^\omega \in \atmtnLang{\dba^{\atmtnState_0}}$.
			\item If $\atmtnState_0 = \atmtnUpdWord(\atmtnState_0, \word)$ and for some index $i$, $1 \le i < n$, we have $\atmtnState_i \strictPrefOrd \atmtnState_i'$ (i.e., $\finGrPth$ takes at least one edge that does not correspond to a transition of the automaton).
			We represent the situation in Figure~\ref{fig:winningCycle}.
			We have that $\atmtnUpdWord(\atmtnState_i', \clr_{i+1}\ldots\clr_n) = \atmtnUpdWord(\atmtnState_0, \word) = \atmtnState_0$.
			By Lemma~\ref{lem:monotonTrans}, as $\atmtnState_i \strictPrefOrd \atmtnState_i'$, this implies that $\atmtnUpdWord(\atmtnState_i, \clr_{i+1}\ldots\clr_n) \prefOrd \atmtnState_0$.
			Also, in the graph, there is a path from $\atmtnState_i$ to $\atmtnState_0$ with colors $\clr_{i+1}\ldots\clr_n$.
			Therefore, by Lemma~\ref{lem:underapprox}, $\atmtnState_0 \prefOrd \atmtnUpdWord(\atmtnState_i, \clr_{i+1}\ldots\clr_n)$.
			Thus, $\atmtnState_0 \prefEq \atmtnUpdWord(\atmtnState_i, \clr_{i+1}\ldots\clr_n)$, and as $\dba$ is built on top of $\minStateAtmtn$, $\atmtnState_0 = \atmtnUpdWord(\atmtnState_i, \clr_{i+1}\ldots\clr_n)$.
			Therefore, $\atmtnState_i \strictPrefOrd \atmtnUpdWord(\atmtnState_i, \clr_{i+1}\ldots\clr_n\clr_1\ldots\clr_i) = \atmtnState_i' $.
			By progress-consistency, $(\clr_{i+1}\ldots\clr_n\clr_1\ldots\clr_i)^\omega \in \atmtnLang{\dba^{\atmtnState_i}}$.
			As $\atmtnUpdWord(\atmtnState_i, \clr_{i+1}\ldots\clr_n) = \atmtnState_0$, this implies that $(\clr_1\ldots\clr_n)^\omega = \word^\omega \in \atmtnLang{\dba^{\atmtnState_0}}$.
			\qedhere
		\end{itemize}
	\end{proof}
	\begin{figure}[htb]
		\centering
		\begin{tikzpicture}[every node/.style={font=\small,inner sep=1pt}]
			\draw (0,0) node[diamant] (q0) {$\atmtnState_0$};
			\draw ($(q0)+(2.5,0)$) node[diamant] (qi) {$\atmtnState_i$};
			\draw ($(qi)+(2.5,0)$) node[diamant] (qi') {$\atmtnState_i'$};
			\draw ($(qi)!0.5!(qi')$) node[] () {$\strictPrefOrd$};
			\draw (q0) edge[-latex',out=60,in=180-60,decorate,distance=0.8cm] node[above=4pt] {$\clr_1\ldots\clr_i$} (qi');
			\draw (qi') edge[-latex',out=-150,in=-30,decorate,distance=1.2cm] node[below=4pt] {$\clr_{i+1}\ldots\clr_n$} (q0);
			\draw (qi) edge[-latex',decorate] node[above=4pt] {$\clr_{i+1}\ldots\clr_n$} (q0);
		\end{tikzpicture}
		\caption{Situation in the last case of the proof of Lemma~\ref{lem:winningCycle}.}
		\label{fig:winningCycle}
	\end{figure}

	We show that every vertex $(\atmtnState_0, \ord_0)$ such that $\atmtnState_0 \prefOrd \atmtnInit$ satisfies $\wc$.
	\begin{lem} \label{lem:satisfies}
		Let $\atmtnState_0\in\atmtnStates$.
		For all ordinals $\ord_0<\mainOrd$, the vertex $(\atmtnState_0, \ord_0)$ of $\univGrB$ satisfies $\atmtnLang{\dba^{\atmtnState_0}}$.
		In particular, if $\atmtnState_0 \prefOrd \atmtnInit$, $(\atmtnState_0, \ord_0)$ satisfies $\wc$.
	\end{lem}
	\begin{proof}
		Let $\grPth = (\atmtnState_0, \ord_0)\grTr{\clr_1}(\atmtnState_1, \ord_1)\grTr{\clr_2} \ldots$ be an infinite path of $\univGrB$ from $(\atmtnState_0, \ord_0)$ and $\word = \clr_1\clr_2\ldots$ be the sequence of colors along its edges.
		We show that $\word\in\atmtnLang{\dba^{\atmtnState_0}}$.

		Let $q_i'=\atmtnUpdWord(q_0, c_1\ldots c_i)$ be the state of $\dba$ reached after reading the first $i$ colors of $\word$.
		We claim that there are two states $\atmtnState, \atmtnState'\in \atmtnStates$ occurring infinitely often in the sequences $(\atmtnState_i)_{i \ge 0}$, $(\atmtnState_i')_{i \ge 0}$, respectively, and an increasing sequence of indices $(i_k)_{k\ge 1}$ satisfying that for all $k\ge 1$,
		\begin{itemize}
			\item $\atmtnState_{i_k} = \atmtnState$, $\atmtnState_{i_k}' = \atmtnState'$, and
			\item $\ord_{i_k} \le \ord_{i_{k+1}}$.
		\end{itemize}

		Indeed, we can first choose a state $q$ appearing infinitely often in $(\atmtnState_i)_{i \ge 0}$ and pick a sequence $(i_j)_{j\ge 1}$ such that $q_{i_j} = q$ and $(\ord_{i_j})_{j\ge 1}$ is not decreasing (this is possible since there is no infinite decreasing sequence of ordinals).
		Then, we can just pick $q'$ appearing infinitely often in $(q_{i_j}')_{j\ge 1}$ and extract the subsequence corresponding to its occurrences.

		Let $q,q'\in \atmtnStates$, $(i_k)_{k\ge 1}$ satisfying the above properties.
		Let $\word_0 = c_1\ldots c_{i_1}$ and for $k \ge 1$, let $\word_k = \clr_{i_k+1}\ldots\clr_{i_{k+1}} \in \colors^+$ be the colors over the edges in $\grPth$ from $(\atmtnState_{i_k}, \ord_{i_k})$ to $(\atmtnState_{i_{k+1}}, \ord_{i_{k+1}})$.

		By Lemma~\ref{lem:underapprox}, it holds that $\atmtnState \prefOrd \atmtnState' = \atmtnUpdWord(\atmtnState_0, \word_0)$.
		Moreover, since $\ord_{i_k} \le \ord_{i_{k+1}}$, by Lemma~\ref{lem:winningCycle} we have that $\word_k^\omega \in \atmtnLang{\dba^{\atmtnState}}\subseteq \atmtnLang{\dba^{\atmtnState'}}$ for all $k\ge 1$.
		We conclude that for all $k\ge 1$, the word $\word_k$ labels a cycle over $q'$ in $\dba$ visiting some B\"uchi transition, and therefore $\word = \word_0\word_1\word_2\ldots \in \atmtnLang{\dba^{\atmtnState_0}}$.
	\end{proof}

	We now have all the tools to show, for all cardinals $\cardinal$, $(\cardinal, \wc)$-universality of $\univGrB$ for sufficiently large $\mainOrd$.

	\begin{prop} \label{prop:universality}
		Let $\cardinal$ be a cardinal, and $\mainOrd'$ be an ordinal such that $\cardinal \le \card{\mainOrd'}$.
		Let $\mainOrd = \card{\atmtnStates}\cdot\mainOrd'$.
		Graph $\univGrB$ is $(\cardinal, \wc)$-universal.
	\end{prop}
	\begin{proof}
		Let $\gr = (\grStates, \grEdges)$ be a graph such that $\card{\grStates} < \cardinal$ (in particular, $\card{\grStates} < \card{\mainOrd'}$).
		For $\grSt\in\grStates$, let $\atmtnState_\grSt\in \atmtnStates \cup \{\top\}$ be the smallest automaton state (for $\prefOrd$) such that $\grSt$ satisfies $\atmtnLang{\dba^{\atmtnState_\grSt}}$ (i.e., such that all infinite paths from $\grSt$ are in objective $\atmtnLang{\dba^{\atmtnState_\grSt}}$), or $\top$ if it satisfies none of them.
		We remark that $\atmtnState_\grSt \prefOrd \atmtnState$ if and only if $\grSt$ satisfies $\atmtnLang{\dba^{\atmtnState}}$.
		To show that there is a $\wc$-preserving morphism from $\gr$ to $\univGrB$, we follow the six steps outlined in~\cite[Lemma~4.3]{Ohl23}.
		\begin{enumerate}[(i)]
			\item \label{item:i}
			In this first step, we classify and order vertices of $\gr$ in an inductive way, which will later be used to map them to vertices of $\univGrB$.
			For $\atmtnState\in\atmtnStates$ and $\ord$ an ordinal, we define by transfinite induction
			\[
			\grStates_\ord^\atmtnState =
			\left\{\grSt\in\grStates \mid \atmtnState_\grSt \prefOrd \atmtnState,\ \text{and}\ \forall \clr\in \colors \, \left(
			\grSt \grTr{\clr} \grSt' \implies [(\atmtnState, \clr) \in \accepSet\ \text{or}\ \exists\ordBis < \ord, \grSt'\in \grStates^{\atmtnUpd(\atmtnState, \clr)}_{\ordBis}]\right )\right\}.
			\]
			Intuitively, for $\grSt$ to be in $\grStates_\ord^\atmtnState$, it has to satisfy $\atmtnLang{\dba^{\atmtnState}}$ and to guarantee that a B\"uchi transition is visited ``quickly'' when colors of paths from $\grSt$ are read from $\atmtnState$ in $\dba$ (how quickly depends on the value of $\ord$).
			We remark that for each state $\atmtnState\in \atmtnStates$, the sequence $(\grStates_\ord^\atmtnState)_{\ord}$ is non-decreasing: for $\ord\leq \ord'$, $\grStates_\ord^\atmtnState \subseteq \grStates_{\ord'}^\atmtnState$.
			We illustrate this induction on a concrete case in Example~\ref{ex:graphMorphism}.
			We also discuss another example in Remark~\ref{rmk:usefulProduct} to illustrate why we need to consider the product $\mainOrd = \card{\atmtnStates}\cdot\mainOrd'$ rather than just taking $\mainOrd = \mainOrd'$ for our construction.
			The subsequent steps mostly follow from this definition.
			\item \label{univ:ii}
			Let $\grStates^\atmtnState = \bigcup_{\ord} \grStates_\ord^{\atmtnState}$.
			We show that if $\grSt$ satisfies $\atmtnLang{\dba^{\atmtnState}}$, then it is in $\grStates^\atmtnState$.
			Assume that $\grSt \notin \grStates^\atmtnState$.
			If $\atmtnState \strictPrefOrd \atmtnState_\grSt$, then we immediately have that $\grSt$ does not satisfy $\atmtnLang{\dba^{\atmtnState}}$.
			If $\atmtnState_\grSt \prefOrd \atmtnState$, then $\grSt$ has an outgoing edge $\grSt \grTr{\clr} \grSt'$ such that $(\atmtnState, \clr) \notin \accepSet$ and $\grSt' \notin \bigcup_{\ord} \grStates_{\ord}^{\atmtnUpd(\atmtnState, \clr)}$.
			By induction, we build an infinite path from $\grSt$ whose projection in $\dba$ only uses non-B\"uchi transitions, so $\grSt$ does not satisfy $\atmtnLang{\dba^{\atmtnState}}$.
			\item \label{item:iii}
			In this step and the next one, we show that there is no use in considering ordinals beyond $\mainOrd$ in our construction.
			We first show that if for all $\atmtnState\in\atmtnStates$, $\grStates_\ord^\atmtnState = \grStates_{\ord + 1}^\atmtnState$, then for all $\atmtnState\in\atmtnStates$ and all $\ord' \ge \ord$, $\grStates_\ord^\atmtnState = \grStates_{\ord'}^\atmtnState$.
			For $\ord \le \ord'$, we always have $\grStates_\ord^\atmtnState \subseteq \grStates_{\ord'}^\atmtnState$.
			For the other inclusion, we assume by transfinite induction that $\grStates_\ord^{\atmtnState'} = \grStates_{\ordBis}^{\atmtnState'}$ for all $\atmtnState'\in\atmtnStates$ and for all $\ordBis$ such that $\ord \le \ordBis < \ord'$.
			Let $\grSt\in\grStates_{\ord'}^\atmtnState$.
			Every edge $\grSt \grTr{\clr} \grSt'$ either satisfies $(\atmtnState, \clr) \in \accepSet$, or there exists $\ordBis < \ord'$ such that $\grSt'\in \grStates^{\atmtnUpd(\atmtnState, \clr)}_{\ordBis}$.
			Since $\grStates^{\atmtnUpd(\atmtnState, \clr)}_{\ordBis} \subseteq \grStates^{\atmtnUpd(\atmtnState, \clr)}_{\ord}$ by induction hypothesis, we have $\grSt'\in \grStates^{\atmtnUpd(\atmtnState, \clr)}_{\ord}$.
			Hence, $\grSt\in\grStates_{\ord + 1}^\atmtnState = \grStates_{\ord}^\atmtnState$.
			\item \label{item:iv}
			We prove that there exists $\ord < \mainOrd$ such that for all $\atmtnState\in\atmtnStates$, $\grStates_\ord^\atmtnState = \grStates_{\ord+1}^\atmtnState$.
			If not, using the axiom of choice, we can build a map $\psi\colon \mainOrd \to \atmtnStates \times \grStates$ such that for $\ord < \mainOrd$, $\psi(\ord) = (\atmtnState, \grSt)$ for some $\atmtnState\in\atmtnStates$ and $\grSt \in \grStates_{\ord + 1}^\atmtnState \setminus \grStates_{\ord}^\atmtnState$.
			This map is injective, as any pair $(\atmtnState, \grSt)$ can be chosen at most once (as $(\grStates_{\ord}^\atmtnState)_\ord$ is non-decreasing).
			This implies that $\card{\mainOrd} = \card{\atmtnStates}\cdot \card{\mainOrd'} \le \card{\atmtnStates}\cdot\card{\grStates}$, a contradiction since $\card{\grStates} < \card{\mainOrd'}$.

			Using additionally Item~\ref{item:iii}, we deduce that there exists $\ord < \mainOrd$ such that for all $\ord' \ge \ord$, $\grStates_\ord^\atmtnState = \grStates_{\ord'}^\atmtnState$.
			\item \label{item:v}
			Let $\morph\colon \grStates \to \univGrStatesB$ be such that
			\[
			\phi(\grSt) = \begin{cases}
				(\atmtnState_\grSt, \min \{\ord \mid \grSt \in \grStates^{\atmtnState_\grSt}_\ord\}) &\text{if $\atmtnState_\grSt \strictPrefOrd \top$,}\\
				\top &\text{if $\atmtnState_\grSt = \top$.}
			\end{cases}
			\]
			By Item~\ref{univ:ii}, for all $\grSt\in\grStates$, there exists $\ord$ such that $\grSt\in\grStates_\ord^{\atmtnState_\grSt}$, so $\{\ord \mid \grSt \in \grStates^{\atmtnState_\grSt}_\ord\}$ is non-empty.
			By Item~\ref{item:iv}, we have that if $\morph(\grSt) = (\atmtnState, \ord)$, then $\ord < \mainOrd$, so the image of $\morph$ is indeed in $\univGrStatesB$.
			We show that $\morph$ is $\wc$-preserving: if $\grSt$ satisfies $\wc$, then $\atmtnState_\grSt \prefOrd \atmtnInit$, so by Lemma~\ref{lem:satisfies}, $\morph(\grSt)$ also satisfies $\wc$.
			\item \label{item:vi}
			We show that $\morph$ is a graph morphism.
			Let $\grSt \grTr{\clr} \grSt'$ be an edge of $\gr$ --- we need to show that $\morph(\grSt) \grTr{\clr} \morph(\grSt')$ is an edge of $\univGrB$.
			If $\morph(\grSt) = \top$, this is clear as there are all possible outgoing edges from $\top$.
			If not, we denote $\morph(\grSt) = (\atmtnState, \ord)$ and $\morph(\grSt') = (\atmtnState', \ord')$.
			We have that $\grSt$ satisfies $\atmtnLang{\dba^{\atmtnState}}$.
			Thus, $\grSt'$ must satisfy $\atmtnLang{\dba^{\atmtnUpd(\atmtnState, \clr)}}$.
			This implies that $\atmtnState' \prefOrd \atmtnUpd(\atmtnState, \clr)$.
			We distinguish two cases.
			\begin{itemize}
				\item If $(\atmtnState, \clr) \in \accepSet$, then by construction of $\univGrB$, there are $\clr$-colored edges from $(\atmtnState, \ord)$ to $(\atmtnUpd(\atmtnState, \clr), \ord'')$ for all ordinals $\ord''$.
				By monotonicity, as $\atmtnState' \prefOrd \atmtnUpd(\atmtnState, \clr)$, there is also a $\clr$-colored edge from $(\atmtnState, \ord)$ to $(\atmtnState', \ord')$.
				\item We assume that $(\atmtnState, \clr) \notin \accepSet$.
				If $\atmtnState' = \atmtnUpd(\atmtnState, \clr)$, this means that there exists $\ordBis < \ord$ such that $\grSt'\in \grStates^{\atmtnState'}_{\ordBis}$.
				Therefore, $\ord' \le \ordBis < \ord$.
				So the edge $(\atmtnState, \ord) \grTr{\clr} (\atmtnState', \ord')$ exists by construction of $\univGrB$.
				If $\atmtnState' \strictPrefOrd \atmtnUpd(\atmtnState, \clr)$, then by construction of $\univGrB$, there are $\clr$-colored edges from $(\atmtnState, \ord)$ to $(\atmtnState', \ord'')$ for all ordinals $\ord''$.
			\end{itemize}
		\end{enumerate}
		We have shown in Item~\ref{item:v} and Item~\ref{item:vi} that $\morph$ is a $\wc$-preserving morphism from $\gr$ to~$\univGrB$.
	\end{proof}

	\begin{exa} \label{ex:graphMorphism}
		We consider the objective $\wc=\Buchi{\{a\}}\cup \colors^*aa\colors^\omega$ recognized by the DBA $\dba$ from Example~\ref{ex:AAorBuchiA}, for which graph $\univGr_{\dba, \omega}$ was shown in Example~\ref{ex:univGrExample}.
		We discuss how our construction maps the vertices of graph $\gr$ in Figure~\ref{fig:graphMorphism}.
		Notice that $\s_1$, $\s_2$ and $\s_3$ satisfy $\wc$, but not $\s_4$ and $\s_5$.

		We build explicitly the sets $\grStates_\ord^\atmtnState$ from Item~\ref{item:i} of the proof of Proposition~\ref{prop:universality}.
		All five vertices are in $\grStates_\ord^{\atmtnState_{aa}}$ for every ordinal $\ord$, as all vertices satisfy $\atmtnLang{\dba^{\atmtnState_{aa}}} = \colors^\omega$, and any transition from $\atmtnState_{aa}$ is a B\"uchi transition.
		For the same reasons, $\s_1$, $\s_2$, $\s_3$, and $\s_4$ are in $\grStates_\ord^{\atmtnState_{a}}$ for all~$\ord$.
		We determine the sets $\grStates_\ord^{\atmtnInit}$ using the inductive definition.
		First, $\s_3$ is in $\grStates_\ord^{\atmtnInit}$ for all $\ord$, since $(\atmtnInit,a)\in \accepSet$.
		Therefore, $\s_2\in \grStates_\ord^{\atmtnInit}$ for $\ord\geq 1$ (the $b$-colored edge from $\s_2$ leads to $\s_3$, and $\atmtnUpd(\atmtnInit, b) = \atmtnInit$).
		Finally, $\s_1\in \grStates_\ord^{\atmtnInit}$ for $\ord\geq 2$ for the same reason.
		The morphism $\morph$ from Item~\ref{item:v} assigns $\morph(\s_1) = (\atmtnInit, 2)$, $\morph(\s_2) = (\atmtnInit, 1)$, $\morph(\s_3) = (\atmtnInit, 0)$, $\morph(\s_4) = (\atmtnState_a, 0)$, $\morph(\s_5) = (\atmtnState_{aa}, 0)$.
		\qedEx
	\end{exa}

	\begin{figure}[tbh]
		\centering
		\begin{tikzpicture}[every node/.style={font=\small,inner sep=1pt}]
			\draw (0,0) node[rond] (v1) {$\s_1$};
			\draw ($(v1)+(1.5,0)$) node[rond] (v2) {$\s_2$};
			\draw ($(v1)+(3,0)$) node[rond] (v21) {$\s_3$};
			\draw ($(v1)+(4.5,0)$) node[rond] (v3) {$\s_4$};
			\draw ($(v1)+(6,0)$) node[rond] (v4) {$\s_5$};
			\draw (v1) edge[-latex',out=30,in=180-30] node[above=4pt] {$b$} (v2);
			\draw (v2) edge[-latex',out=180+30,in=-30] node[below=4pt] {$a$} (v1);
			\draw (v2) edge[-latex'] node[above=4pt] {$b$} (v21);
			\draw (v21) edge[-latex',out=30,in=180-30] node[above=4pt] {$a$} (v3);
			\draw (v3) edge[-latex',out=180+30,in=-30] node[below=4pt] {$b$} (v21);
			\draw (v3) edge[-latex'] node[above=4pt] {$a$} (v4);
			\draw (v4) edge[-latex',out=-30,in=30,distance=0.8cm] node[right=4pt] {$b$} (v4);
		\end{tikzpicture}
		\caption{Graph $\gr$ used in Example~\ref{ex:graphMorphism}.}
		\label{fig:graphMorphism}
	\end{figure}

	\begin{rem} \label{rmk:usefulProduct}
		We illustrate why, in the statement of Proposition~\ref{prop:universality}, we use $\card{\atmtnStates}\cdot\mainOrd'$ vertices for each automaton state and not just $\mainOrd'$.

		Let $\colors = \{a, b, c\}$, $\wc$ be the objective recognized by the DBA in Figure~\ref{fig:usefulProduct} (left), and $\gr$ be the graph with two vertices in Figure~\ref{fig:usefulProduct} (right).
		This objective has a total prefix preorder, is progress-consistent, is recognized by a DBA built on top of its prefix-classifier, and the DBA in Figure~\ref{fig:usefulProduct} is saturated.
		Both vertices of $\gr$ satisfy $\wc = \atmtnLang{\dba^{\atmtnState_1}}$ but do not satisfy $\atmtnLang{\dba^{\atmtnState_0}}$, so $\atmtnState_{\s_1} = \atmtnState_{\s_2} = \atmtnState_1$.
		We have $\s_1 \in \grStates_0^{\atmtnState_3}$ (the base case of the induction), and this inductively implies that $\s_2 \in \grStates_1^{\atmtnState_3}$, $\s_1 \in \grStates_2^{\atmtnState_2}$, $\s_2 \in \grStates_3^{\atmtnState_2}$, $\s_1 \in \grStates_4^{\atmtnState_1}$, $\s_2 \in \grStates_5^{\atmtnState_1}$.
		We represent these steps in the figure, below the graph.
		Thus, $\morph(\s_1) = (\atmtnState_1, 4)$ and $\morph(\s_2) = (\atmtnState_1, 5)$.
		We see here that two copies of each automaton state in our universal graph construction would not have sufficed; we actually used six different indices to fully understand the situation, which is due to the interlacing between the graph and the structure of the automaton.
		\qedEx
	\end{rem}

	\begin{figure}[htb]
		\centering
		\begin{minipage}{0.61\linewidth}
			\centering
			\begin{tikzpicture}[every node/.style={font=\small,inner sep=1pt}]
				\draw (0,0) node[diamant] (q1) {$\atmtnState_1$};
				\draw ($(q1.south)-(0,0.4)$) edge[-latex'] (q1);
				\draw ($(q1)-(1.5,0)$) node[diamant] (q0) {$\atmtnState_0$};
				\draw ($(q1)+(1.5,0)$) node[diamant] (q2) {$\atmtnState_2$};
				\draw ($(q1)+(3,0)$) node[diamant] (q3) {$\atmtnState_3$};

				\draw (q1) edge[-latex'] node[above=4pt] {$b$} (q0);
				\draw (q1) edge[-latex',out=30,in=180-30] node[above=4pt] {$a$} (q2);
				\draw (q2) edge[-latex',out=180+30,in=-30] node[below=4pt] {$b$} (q1);
				\draw (q2) edge[-latex',out=30,in=180-30] node[above=4pt] {$a$} (q3);
				\draw (q3) edge[-latex',out=180+30,in=-30] node[below=4pt] {$b$} (q2);
				\draw (q3) edge[-latex',out=-30,in=30,distance=0.8cm,accepting] node[right=4pt] {$a$} (q3);
				\draw (q3) edge[-latex',out=60,in=120,distance=0.8cm] node[above=4pt] {$c$} (q3);
				\draw (q1) edge[-latex',out=60,in=120,distance=0.8cm] node[above=4pt] {$c$} (q1);
				\draw (q2) edge[-latex',out=60,in=120,distance=0.8cm] node[above=4pt] {$c$} (q2);
				\draw (q0) edge[-latex',out=60,in=120,distance=0.8cm] node[above=4pt] {$a, b, c$} (q0);
			\end{tikzpicture}
		\end{minipage}%
		\begin{minipage}{0.39\linewidth}
			\centering
			\begin{tikzpicture}[every node/.style={font=\small,inner sep=1pt}]
				\draw (0,0) node[rond] (v1) {$\s_1$};
				\draw ($(v1)+(1.5,0)$) node[rond] (v2) {$\s_2$};
				\draw (v1) edge[-latex',out=30,in=180-30] node[above=4pt] {$a$} (v2);
				\draw (v2) edge[-latex',out=180+30,in=-30] node[above=4pt] {$c$} (v1);

				\draw ($(v1)-(0,.58)$) node[rond,draw=none] (in1) {\rotatebox{-90}{$\in$}};
				\draw ($(v2)-(0,.75)$) node[rond,draw=none] (in1) {\rotatebox{-90}{$\in$}};
				\draw ($(v1)-(0,1.)$) node[rond,draw=none] (v11) {$\grStates_0^{\atmtnState_3}$};
				\draw ($(v2)-(0,1.4)$) node[rond,draw=none] (v21) {$\grStates_1^{\atmtnState_3}$};
				\draw ($(v11)-(0,0.8)$) node[rond,draw=none] (v12) {$\grStates_2^{\atmtnState_2}$};
				\draw ($(v21)-(0,0.8)$) node[rond,draw=none] (v22) {$\grStates_3^{\atmtnState_2}$};
				\draw ($(v12)-(0,0.8)$) node[rond,draw=none] (v13) {$\grStates_4^{\atmtnState_1}$};
				\draw ($(v22)-(0,0.8)$) node[rond,draw=none] (v23) {$\grStates_5^{\atmtnState_1}$};
				\draw (v21) edge[-latex'] node[above=1pt,xshift=3pt] {$c$} (v11);
				\draw (v22) edge[-latex'] node[above=1pt,xshift=3pt] {$c$} (v12);
				\draw (v23) edge[-latex'] node[above=1pt,xshift=3pt] {$c$} (v13);
				\draw (v12) edge[-latex'] node[above=1pt,xshift=-3pt] {$a$} (v21);
				\draw (v13) edge[-latex'] node[above=1pt,xshift=-3pt] {$a$} (v22);
			\end{tikzpicture}
		\end{minipage}
		\caption{DBA (left) and graph $\gr$ (right) used in Remark~\ref{fig:usefulProduct} to illustrate the need for $\card{\atmtnStates}\cdot\mainOrd'$ vertices for each automaton state in the universal graph used in the construction of Proposition~\ref{prop:universality}.}
		\label{fig:usefulProduct}
	\end{figure}

	We conclude this section by proving Proposition~\ref{prop:sufficient}, showing that $\wc$ is half-positional under the three conditions from Theorem~\ref{thm:mainChar}.
	\begin{proof}[Proof of Proposition~\ref{prop:sufficient}]
		Using Lemma~\ref{lem:monotonic} and Proposition~\ref{prop:universality}, we have for all cardinals $\cardinal$ that there exists a completely well-monotonic $(\cardinal, \wc)$-universal graph.
		By Theorem~\ref{thm:univGr}, this implies that $\wc$ is half-positional.
	\end{proof}

	\section{Conclusion and future work}
	We have provided a characterization of half-positionality for DBA-recognizable objectives.
	This novel result is necessary to advance toward a full understanding of half-positionality of $\omega$-regular objectives.

	Besides $\omega$-regular objectives, one could study the half-positionality of classes of objectives defined by topological properties; for instance, those corresponding to different levels of the Borel hierarchy. In particular, a natural step in this direction would be to study the half-positionality of $\Pi_2^0$ objectives, which correspond to those that can be recognized by deterministic, but infinite, B\"uchi automata.

	Another interesting extension is to characterize the memory requirements of DBA-recognizable objectives.
	An intermediate and already seemingly difficult step would be to characterize the memory requirements of objectives recognized by deterministic \emph{weak} automata~\cite{Wagner1979omega,Sta83,Lod01}, generalizing the characterization for safety specifications~\cite{CFH14}.

	\section*{Acknowledgments}
	The authors would like to thank Igor Walukiewicz for suggesting a simplification of the proof of Lemma~\ref{lem:PrefIndMeansBuchi} and Pierre Ohlmann for interesting discussions on the subject.

	\bibliographystyle{alphaurl}
	\bibliography{articlesBuchi}

\newcommand{\etalchar}[1]{$^{#1}$}
\begin{thebibliography}{BFMM11}

\bibitem[AF21]{AF21}
Dana Angluin and Dana Fisman.
\newblock Regular $\omega$-languages with an informative right congruence.
\newblock {\em Information and Computation}, 278:104598, 2021.
\newblock \href {https://doi.org/10.1016/j.ic.2020.104598} {\path{doi:10.1016/j.ic.2020.104598}}.

\bibitem[AK22]{AK22Minimizing}
Bader {Abu Radi} and Orna Kupferman.
\newblock Minimization and canonization of {GFG} transition-based automata.
\newblock {\em Logical Methods in Computer Science}, 18(3), 2022.
\newblock \href {https://doi.org/10.46298/lmcs-18(3:16)2022} {\path{doi:10.46298/lmcs-18(3:16)2022}}.

\bibitem[AR17]{AR17}
Benjamin Aminof and Sasha Rubin.
\newblock First-cycle games.
\newblock {\em Information and Computation}, 254:195--216, 2017.
\newblock \href {https://doi.org/10.1016/j.ic.2016.10.008} {\path{doi:10.1016/j.ic.2016.10.008}}.

\bibitem[BCJ18]{BCJ18}
Roderick Bloem, Krishnendu Chatterjee, and Barbara Jobstmann.
\newblock Graph games and reactive synthesis.
\newblock In Edmund~M. Clarke, Thomas~A. Henzinger, Helmut Veith, and Roderick Bloem, editors, {\em Handbook of Model Checking}, pages 921--962. Springer, 2018.
\newblock \href {https://doi.org/10.1007/978-3-319-10575-8_27} {\path{doi:10.1007/978-3-319-10575-8_27}}.

\bibitem[BCRV22]{BCRV22Conf}
Patricia Bouyer, Antonio Casares, Mickael Randour, and Pierre Vandenhove.
\newblock Half-positional objectives recognized by deterministic {B}{\"{u}}chi automata.
\newblock In Bartek Klin, S{\l}awomir Lasota, and Anca Muscholl, editors, {\em Proceedings of the 33rd International Conference on Concurrency Theory, {CONCUR} 2022, Warsaw, Poland, September 12--16, 2022}, volume 243 of {\em LIPIcs}, pages 20:1--20:18. Schloss Dagstuhl -- Leibniz-Zentrum f{\"{u}}r Informatik, 2022.
\newblock \href {https://doi.org/10.4230/LIPIcs.CONCUR.2022.20} {\path{doi:10.4230/LIPIcs.CONCUR.2022.20}}.

\bibitem[BFL{\etalchar{+}}08]{BFLMS08}
Patricia Bouyer, Ulrich Fahrenberg, Kim~G. Larsen, Nicolas Markey, and Jir{\'i} Srba.
\newblock Infinite runs in weighted timed automata with energy constraints.
\newblock In Franck Cassez and Claude Jard, editors, {\em Proceedings of the 6th International Conference on Formal Modeling and Analysis of Timed Systems, {FORMATS} 2008, Saint Malo, France, September 15--17, 2008}, volume 5215 of {\em Lecture Notes in Computer Science}, pages 33--47. Springer, 2008.
\newblock \href {https://doi.org/10.1007/978-3-540-85778-5_4} {\path{doi:10.1007/978-3-540-85778-5_4}}.

\bibitem[BFMM11]{BFMM11}
Alessandro Bianco, Marco Faella, Fabio Mogavero, and Aniello Murano.
\newblock Exploring the boundary of half-positionality.
\newblock {\em Annals of Mathematics and Artificial Intelligence}, 62(1-2):55--77, 2011.
\newblock \href {https://doi.org/10.1007/s10472-011-9250-1} {\path{doi:10.1007/s10472-011-9250-1}}.

\bibitem[BHR16]{BHR16}
V{\'{e}}ronique Bruy{\`{e}}re, Quentin Hautem, and Mickael Randour.
\newblock Window parity games: an alternative approach toward parity games with time bounds.
\newblock In Domenico Cantone and Giorgio Delzanno, editors, {\em Proceedings of the 7th International Symposium on Games, Automata, Logics, and Formal Verification, {GandALF} 2016, Catania, Italy, September 14--16, 2016}, volume 226 of {\em {EPTCS}}, pages 135--148, 2016.
\newblock \href {https://doi.org/10.4204/EPTCS.226.10} {\path{doi:10.4204/EPTCS.226.10}}.

\bibitem[BHRR19]{BHRR19}
V{\'{e}}ronique Bruy{\`{e}}re, Quentin Hautem, Mickael Randour, and Jean{-}Fran{\c{c}}ois Raskin.
\newblock Energy mean-payoff games.
\newblock In Wan~J. Fokkink and Rob van Glabbeek, editors, {\em Proceedings of the 30th International Conference on Concurrency Theory, {CONCUR} 2019, Amsterdam, the Netherlands, August 27--30, 2019}, volume 140 of {\em LIPIcs}, pages 21:1--21:17. Schloss Dagstuhl -- Leibniz-Zentrum f{\"{u}}r Informatik, 2019.
\newblock \href {https://doi.org/10.4230/LIPIcs.CONCUR.2019.21} {\path{doi:10.4230/LIPIcs.CONCUR.2019.21}}.

\bibitem[BK08]{BK08}
Christel Baier and Joost{-}Pieter Katoen.
\newblock {\em Principles of model checking}.
\newblock {MIT} Press, 2008.

\bibitem[BL69]{BL69}
J.~Richard B{\"{u}}chi and Lawrence~H. Landweber.
\newblock Definability in the monadic second-order theory of successor.
\newblock {\em Journal of Symbolic Logic}, 34(2):166--170, 1969.
\newblock \href {https://doi.org/10.2307/2271090} {\path{doi:10.2307/2271090}}.

\bibitem[BLO{\etalchar{+}}22]{BLORV22}
Patricia Bouyer, St{\'{e}}phane {Le Roux}, Youssouf Oualhadj, Mickael Randour, and Pierre Vandenhove.
\newblock Games where you can play optimally with arena-independent finite memory.
\newblock {\em Logical Methods in Computer Science}, 18(1), 2022.
\newblock \href {https://doi.org/10.46298/lmcs-18(1:11)2022} {\path{doi:10.46298/lmcs-18(1:11)2022}}.

\bibitem[BMR{\etalchar{+}}18]{BMRLL18}
Patricia Bouyer, Nicolas Markey, Mickael Randour, Kim~G. Larsen, and Simon Laursen.
\newblock Average-energy games.
\newblock {\em Acta Informatica}, 55(2):91--127, 2018.
\newblock \href {https://doi.org/10.1007/s00236-016-0274-1} {\path{doi:10.1007/s00236-016-0274-1}}.

\bibitem[Bok18]{Bok18}
Udi Boker.
\newblock Why these automata types?
\newblock In Gilles Barthe, Geoff Sutcliffe, and Margus Veanes, editors, {\em Proceedings of the 22nd International Conference on Logic for Programming, Artificial Intelligence and Reasoning, {LPAR 2022}, Awassa, Ethiopia, 16--21 November 2018}, volume~57 of {\em {EP}i{C} Series in Computing}, pages 143--163. EasyChair, 2018.
\newblock \href {https://doi.org/10.29007/c3bj} {\path{doi:10.29007/c3bj}}.

\bibitem[BORV21]{BORV21}
Patricia Bouyer, Youssouf Oualhadj, Mickael Randour, and Pierre Vandenhove.
\newblock Arena-independent finite-memory determinacy in stochastic games.
\newblock In Serge Haddad and Daniele Varacca, editors, {\em Proceedings of the 32nd International Conference on Concurrency Theory, {CONCUR} 2021, Virtual Conference, August 24--27, 2021}, volume 203 of {\em LIPIcs}, pages 26:1--26:18. Schloss Dagstuhl -- Leibniz-Zentrum f{\"{u}}r Informatik, 2021.
\newblock \href {https://doi.org/10.4230/LIPIcs.CONCUR.2021.26} {\path{doi:10.4230/LIPIcs.CONCUR.2021.26}}.

\bibitem[BRV23]{BRV23}
Patricia Bouyer, Mickael Randour, and Pierre Vandenhove.
\newblock Characterizing omega-regularity through finite-memory determinacy of games on infinite graphs.
\newblock {\em Theoreti{CS}}, 2:1--48, 2023.
\newblock \href {https://doi.org/10.46298/theoretics.23.1} {\path{doi:10.46298/theoretics.23.1}}.

\bibitem[B{\"{u}}c62]{Buchi1962decision}
J.~Richard B{\"{u}}chi.
\newblock On a decision method in restricted second order arithmetic.
\newblock {\em Proceedings of the International Congress on Logic, Methodology and Philosophy of Science}, pages 1--11, 1962.
\newblock \href {https://doi.org/10.1007/978-1-4613-8928-6_23} {\path{doi:10.1007/978-1-4613-8928-6_23}}.

\bibitem[Cas22]{Cas22}
Antonio Casares.
\newblock On the minimisation of transition-based {R}abin automata and the chromatic memory requirements of {M}uller conditions.
\newblock In Florin Manea and Alex Simpson, editors, {\em Proceedings of the 30th {EACSL} Annual Conference on Computer Science Logic, {CSL} 2022, G{\"{o}}ttingen, Germany, February 14--19, 2022}, volume 216 of {\em LIPIcs}, pages 12:1--12:17. Schloss Dagstuhl -- Leibniz-Zentrum f{\"{u}}r Informatik, 2022.
\newblock \href {https://doi.org/10.4230/LIPIcs.CSL.2022.12} {\path{doi:10.4230/LIPIcs.CSL.2022.12}}.

\bibitem[CCL22]{CCL22SizeGFG}
Antonio Casares, Thomas Colcombet, and Karoliina Lehtinen.
\newblock On the size of good-for-games {R}abin automata and its link with the memory in {M}uller games.
\newblock In Miko{\l}aj Boja{\'n}czyk, Emanuela Merelli, and David~P. Woodruff, editors, {\em Proceedings of the 49th International Colloquium on Automata, Languages, and Programming, {ICALP} 2022, Paris, France, July 4--8, 2022}, volume 229 of {\em LIPIcs}, pages 117:1--117:20. Schloss Dagstuhl -- Leibniz-Zentrum f{\"{u}}r Informatik, 2022.
\newblock \href {https://doi.org/10.4230/LIPIcs.ICALP.2022.117} {\path{doi:10.4230/LIPIcs.ICALP.2022.117}}.

\bibitem[CD12]{CD12EP}
Krishnendu Chatterjee and Laurent Doyen.
\newblock Energy parity games.
\newblock {\em Theoretical Computer Science}, 458:49--60, 2012.
\newblock \href {https://doi.org/10.1016/j.tcs.2012.07.038} {\path{doi:10.1016/j.tcs.2012.07.038}}.

\bibitem[CDRR15]{CDRR15}
Krishnendu Chatterjee, Laurent Doyen, Mickael Randour, and Jean{-}Fran{\c{c}}ois Raskin.
\newblock Looking at mean-payoff and total-payoff through windows.
\newblock {\em Information and Computation}, 242:25--52, 2015.
\newblock \href {https://doi.org/10.1016/j.ic.2015.03.010} {\path{doi:10.1016/j.ic.2015.03.010}}.

\bibitem[CF13]{CF13}
Krishnendu Chatterjee and Nathana{\"{e}}l Fijalkow.
\newblock Infinite-state games with finitary conditions.
\newblock In Simona Ronchi~Della Rocca, editor, {\em Proceedings of the 27th International Workshop on Computer Science Logic, {CSL} 2013, 22nd Annual Conference of the {EACSL}, Torino, Italy, September 2--5, 2013}, volume~23 of {\em LIPIcs}, pages 181--196. Schloss Dagstuhl -- Leibniz-Zentrum f{\"{u}}r Informatik, 2013.
\newblock \href {https://doi.org/10.4230/LIPIcs.CSL.2013.181} {\path{doi:10.4230/LIPIcs.CSL.2013.181}}.

\bibitem[CFGO22]{CFGO22}
Thomas Colcombet, Nathana{\"{e}}l Fijalkow, Pawel Gawrychowski, and Pierre Ohlmann.
\newblock The theory of universal graphs for infinite duration games.
\newblock {\em Logical Methods in Computer Science}, 18(3), 2022.
\newblock \href {https://doi.org/10.46298/lmcs-18(3:29)2022} {\path{doi:10.46298/lmcs-18(3:29)2022}}.

\bibitem[CFH14]{CFH14}
Thomas Colcombet, Nathana{\"{e}}l Fijalkow, and Florian Horn.
\newblock Playing safe.
\newblock In Venkatesh Raman and S.~P. Suresh, editors, {\em Proceedings of the 34th {IARCS} Annual Conference on Foundations of Software Technology and Theoretical Computer Science, {FSTTCS} 2014, New Delhi, India, December 15--17, 2014}, volume~29 of {\em LIPIcs}, pages 379--390. Schloss Dagstuhl -- Leibniz-Zentrum f{\"{u}}r Informatik, 2014.
\newblock \href {https://doi.org/10.4230/LIPIcs.FSTTCS.2014.379} {\path{doi:10.4230/LIPIcs.FSTTCS.2014.379}}.

\bibitem[CHJ05]{CHJ05MPPar}
Krishnendu Chatterjee, Thomas~A. Henzinger, and Marcin Jurdzi{\'n}ski.
\newblock Mean-payoff parity games.
\newblock In {\em Proceedings of the 20th Annual {IEEE} Symposium on Logic in Computer Science, {LICS} 2005, Chicago, {IL}, {USA}, June 26--29, 2005}, pages 178--187. {IEEE} Computer Society, 2005.
\newblock \href {https://doi.org/10.1109/LICS.2005.26} {\path{doi:10.1109/LICS.2005.26}}.

\bibitem[CN06]{CN06}
Thomas Colcombet and Damian Niwi\'nski.
\newblock On the positional determinacy of edge-labeled games.
\newblock {\em Theoretical Computer Science}, 352(1-3):190--196, 2006.
\newblock \href {https://doi.org/10.1016/j.tcs.2005.10.046} {\path{doi:10.1016/j.tcs.2005.10.046}}.

\bibitem[DJW97]{DJW97}
Stefan Dziembowski, Marcin Jurdzi{\'n}ski, and Igor Walukiewicz.
\newblock How much memory is needed to win infinite games?
\newblock In {\em Proceedings of the 12th Annual {IEEE} Symposium on Logic in Computer Science, {LICS} 1997, Warsaw, Poland, June 29 -- July 2, 1997}, pages 99--110. {IEEE} Computer Society, 1997.
\newblock \href {https://doi.org/10.1109/LICS.1997.614939} {\path{doi:10.1109/LICS.1997.614939}}.

\bibitem[EJ91]{EJ91}
E.~Allen Emerson and Charanjit~S. Jutla.
\newblock Tree automata, mu-calculus and determinacy (extended abstract).
\newblock In {\em Proceedings of the 32nd Annual Symposium on Foundations of Computer Science, FOCS 1991, San Juan, Puerto Rico, October, 1991}, pages 368--377. {IEEE} Computer Society, 1991.
\newblock \href {https://doi.org/10.1109/SFCS.1991.185392} {\path{doi:10.1109/SFCS.1991.185392}}.

\bibitem[EM79]{EM79}
Andrzej Ehrenfeucht and Jan Mycielski.
\newblock Positional strategies for mean payoff games.
\newblock {\em International Journal of Game Theory}, 8(2):109--113, 1979.
\newblock \href {https://doi.org/10.1007/BF01768705} {\path{doi:10.1007/BF01768705}}.

\bibitem[GH82]{GH82}
Yuri Gurevich and Leo Harrington.
\newblock Trees, automata, and games.
\newblock In Harry~R. Lewis, Barbara~B. Simons, Walter~A. Burkhard, and Lawrence~H. Landweber, editors, {\em Proceedings of the 14th Annual {ACM} Symposium on Theory of Computing, {STOC} 1982, San Francisco, {CA}, {USA}, May 5--7, 1982}, pages 60--65. {ACM}, 1982.
\newblock \href {https://doi.org/10.1145/800070.802177} {\path{doi:10.1145/800070.802177}}.

\bibitem[GK14a]{GK14}
Hugo Gimbert and Edon Kelmendi.
\newblock Submixing and shift-invariant stochastic games.
\newblock {\em CoRR}, abs/1401.6575, 2014.
\newblock \href {https://doi.org/10.48550/arXiv.1401.6575} {\path{doi:10.48550/arXiv.1401.6575}}.

\bibitem[GK14b]{GK14Old}
Hugo Gimbert and Edon Kelmendi.
\newblock Two-player perfect-information shift-invariant submixing stochastic games are half-positional.
\newblock {\em CoRR}, abs/1401.6575v2, 2014.
\newblock \href {https://arxiv.org/abs/1401.6575v2} {\path{arXiv:1401.6575v2}}.

\bibitem[GZ04]{GZ04}
Hugo Gimbert and Wies{\l}aw Zielonka.
\newblock When can you play positionally?
\newblock In Ji{\v{r}}{\'{i}} Fiala, V{\'{a}}clav Koubek, and Jan Kratochv{\'{i}}l, editors, {\em Proceedings of the 29th International Symposium on Mathematical Foundations of Computer Science, {MFCS} 2004, Prague, Czech Republic, August 22--27, 2004}, volume 3153 of {\em Lecture Notes in Computer Science}, pages 686--697. Springer, 2004.
\newblock \href {https://doi.org/10.1007/978-3-540-28629-5_53} {\path{doi:10.1007/978-3-540-28629-5_53}}.

\bibitem[GZ05]{GZ05}
Hugo Gimbert and Wies{\l}aw Zielonka.
\newblock Games where you can play optimally without any memory.
\newblock In Mart{\'i}n Abadi and Luca {de Alfaro}, editors, {\em Proceedings of the 16th International Conference on Concurrency Theory, {CONCUR} 2005, San Francisco, {CA}, {USA}, August 23--26, 2005}, volume 3653 of {\em Lecture Notes in Computer Science}, pages 428--442. Springer, 2005.
\newblock \href {https://doi.org/10.1007/11539452_33} {\path{doi:10.1007/11539452_33}}.

\bibitem[Hor09]{Horn2009RandomFruits}
Florian Horn.
\newblock Random fruits on the {Z}ielonka tree.
\newblock In Susanne Albers and Jean{-}Yves Marion, editors, {\em Proceedings of the 26th International Symposium on Theoretical Aspects of Computer Science, {STACS} 2009, Freiburg, Germany, February 26--28, 2009}, volume~3 of {\em LIPIcs}, pages 541--552. Schloss Dagstuhl -- Leibniz-Zentrum f{\"{u}}r Informatik, Germany, 2009.
\newblock \href {https://doi.org/10.4230/LIPIcs.STACS.2009.1848} {\path{doi:10.4230/LIPIcs.STACS.2009.1848}}.

\bibitem[KK91]{KK91RabinMeasures}
Nils Klarlund and Dexter Kozen.
\newblock {R}abin measures and their applications to fairness and automata theory.
\newblock In {\em Proceedings of the 6th Annual {IEEE} Symposium on Logic in Computer Science, {LICS} 1991, Amsterdam, The Netherlands, July 15--18, 1991}, pages 256--265. {IEEE} Computer Society, 1991.
\newblock \href {https://doi.org/10.1109/LICS.1991.151650} {\path{doi:10.1109/LICS.1991.151650}}.

\bibitem[Kla94]{Kla94}
Nils Klarlund.
\newblock Progress measures, immediate determinacy, and a subset construction for tree automata.
\newblock {\em Annals of Pure and Applied Logic}, 69(2-3):243--268, 1994.
\newblock \href {https://doi.org/10.1016/0168-0072(94)90086-8} {\path{doi:10.1016/0168-0072(94)90086-8}}.

\bibitem[Kop06]{Kop06}
Eryk Kopczy{\'n}ski.
\newblock Half-positional determinacy of infinite games.
\newblock In Michele Bugliesi, Bart Preneel, Vladimiro Sassone, and Ingo Wegener, editors, {\em Proceedings (Part {II}) of the 33rd International Colloquium on Automata, Languages and Programming, {ICALP} 2006, Venice, Italy, July 10--14, 2006}, volume 4052 of {\em Lecture Notes in Computer Science}, pages 336--347. Springer, 2006.
\newblock \href {https://doi.org/10.1007/11787006_29} {\path{doi:10.1007/11787006_29}}.

\bibitem[Kop07]{Kop07}
Eryk Kopczy{\'n}ski.
\newblock Omega-regular half-positional winning conditions.
\newblock In Jacques Duparc and Thomas~A. Henzinger, editors, {\em Proceedings of the 21st International Workshop on Computer Science Logic, {CSL} 2007, 16th Annual Conference of the {EACSL}, Lausanne, Switzerland, September 11--15, 2007}, volume 4646 of {\em Lecture Notes in Computer Science}, pages 41--53. Springer, 2007.
\newblock \href {https://doi.org/10.1007/978-3-540-74915-8_7} {\path{doi:10.1007/978-3-540-74915-8_7}}.

\bibitem[Kop08]{KopThesis}
Eryk Kopczy{\'n}ski.
\newblock {\em Half-positional Determinacy of Infinite Games}.
\newblock PhD thesis, Warsaw University, 2008.

\bibitem[Koz22a]{Koz22Refutation}
Alexander Kozachinskiy.
\newblock Energy games over totally ordered groups.
\newblock {\em CoRR}, abs/2205.04508, 2022.
\newblock \href {https://doi.org/10.48550/arXiv.2205.04508} {\path{doi:10.48550/arXiv.2205.04508}}.

\bibitem[Koz22b]{Koz22}
Alexander Kozachinskiy.
\newblock One-to-two-player lifting for mildly growing memory.
\newblock In Petra Berenbrink and Benjamin Monmege, editors, {\em Proceedings of the 39th International Symposium on Theoretical Aspects of Computer Science, {STACS} 2022, Marseille, France, March 15--18, 2022}, volume 219 of {\em LIPIcs}, pages 43:1--43:21. Schloss Dagstuhl -- Leibniz-Zentrum f{\"{u}}r Informatik, 2022.
\newblock \href {https://doi.org/10.4230/LIPIcs.STACS.2022.43} {\path{doi:10.4230/LIPIcs.STACS.2022.43}}.

\bibitem[KS15]{KS15}
Denis Kuperberg and Michal Skrzypczak.
\newblock On determinisation of good-for-games automata.
\newblock In Magn{\'{u}}s~M. Halld{\'{o}}rsson, Kazuo Iwama, Naoki Kobayashi, and Bettina Speckmann, editors, {\em Proceedings (Part {II}) of the 42nd International Colloquium on Automata, Languages, and Programming, {ICALP} 2015, Kyoto, Japan, July 6--10, 2015}, volume 9135 of {\em Lecture Notes in Computer Science}, pages 299--310. Springer, 2015.
\newblock \href {https://doi.org/10.1007/978-3-662-47666-6_24} {\path{doi:10.1007/978-3-662-47666-6_24}}.

\bibitem[Kur87]{Kur87}
Robert~P. Kurshan.
\newblock Complementing deterministic {B}{\"{u}}chi automata in polynomial time.
\newblock {\em Journal of Computer and System Sciences}, 35(1):59--71, 1987.
\newblock \href {https://doi.org/10.1016/0022-0000(87)90036-5} {\path{doi:10.1016/0022-0000(87)90036-5}}.

\bibitem[Lan69]{Landweber69}
Lawrence~H. Landweber.
\newblock Decision problems for omega-automata.
\newblock {\em Mathematical Systems Theory}, 3(4):376--384, 1969.
\newblock \href {https://doi.org/10.1007/BF01691063} {\path{doi:10.1007/BF01691063}}.

\bibitem[{Le }90]{Saec90}
Bertrand {Le Sa{\"{e}}c}.
\newblock Saturating right congruences.
\newblock {\em {RAIRO} -- Theoretical Informatics and Applications}, 24:545--559, 1990.
\newblock \href {https://doi.org/10.1051/ita/1990240605451} {\path{doi:10.1051/ita/1990240605451}}.

\bibitem[L{\"{o}}d01]{Lod01}
Christof L{\"{o}}ding.
\newblock Efficient minimization of deterministic weak omega-automata.
\newblock {\em Information Processing Letters}, 79(3):105--109, 2001.
\newblock \href {https://doi.org/10.1016/S0020-0190(00)00183-6} {\path{doi:10.1016/S0020-0190(00)00183-6}}.

\bibitem[LPR18]{LPR18}
St{\'{e}}phane {Le Roux}, Arno Pauly, and Mickael Randour.
\newblock Extending finite-memory determinacy by {B}oolean combination of winning conditions.
\newblock In Sumit Ganguly and Paritosh~K. Pandya, editors, {\em Proceedings of the 38th {IARCS} Annual Conference on Foundations of Software Technology and Theoretical Computer Science, {FSTTCS} 2018, Ahmedabad, India, December 11--13, 2018}, volume 122 of {\em LIPIcs}, pages 38:1--38:20. Schloss Dagstuhl -- Leibniz-Zentrum f{\"{u}}r Informatik, 2018.
\newblock \href {https://doi.org/10.4230/LIPIcs.FSTTCS.2018.38} {\path{doi:10.4230/LIPIcs.FSTTCS.2018.38}}.

\bibitem[McN66]{McN66}
Robert McNaughton.
\newblock Testing and generating infinite sequences by a finite automaton.
\newblock {\em Information and Control}, 9(5):521--530, 1966.
\newblock \href {https://doi.org/10.1016/S0019-9958(66)80013-X} {\path{doi:10.1016/S0019-9958(66)80013-X}}.

\bibitem[Mos84]{Mos84}
Andrzej~W. Mostowski.
\newblock Regular expressions for infinite trees and a standard form of automata.
\newblock In Andrzej Skowron, editor, {\em Proceedings of the 5th Symposium on Computation Theory, SCT 1984, Zabor\'ow, Poland, December 3--8, 1984}, volume 208 of {\em Lecture Notes in Computer Science}, pages 157--168. Springer, 1984.
\newblock \href {https://doi.org/10.1007/3-540-16066-3_15} {\path{doi:10.1007/3-540-16066-3_15}}.

\bibitem[MS97]{MS97}
Oded Maler and Ludwig Staiger.
\newblock On syntactic congruences for omega-languages.
\newblock {\em Theoretical Computer Science}, 183(1):93--112, 1997.
\newblock \href {https://doi.org/10.1016/S0304-3975(96)00312-X} {\path{doi:10.1016/S0304-3975(96)00312-X}}.

\bibitem[Ner58]{Ner58}
Anil Nerode.
\newblock Linear automaton transformations.
\newblock {\em Proceedings of the American Mathematical Society}, 9(4):541--544, 1958.
\newblock \href {https://doi.org/10.2307/2033204} {\path{doi:10.2307/2033204}}.

\bibitem[Ohl21]{OhlmannThesis}
Pierre Ohlmann.
\newblock {\em Monotonic graphs for parity and mean-payoff games}.
\newblock PhD thesis, IRIF -- Research Institute on the Foundations of Computer Science, 2021.

\bibitem[Ohl23]{Ohl23}
Pierre Ohlmann.
\newblock Characterizing positionality in games of infinite duration over infinite graphs.
\newblock {\em Theoreti{CS}}, 2:1--51, 2023.
\newblock \href {https://doi.org/10.46298/theoretics.23.3} {\path{doi:10.46298/theoretics.23.3}}.

\bibitem[Pnu77]{Pnu77}
Amir Pnueli.
\newblock The temporal logic of programs.
\newblock In {\em Proceedings of the 18th Annual Symposium on Foundations of Computer Science, FOCS 1977, Providence, {RI}, {USA}, October 31 -- November 1, 1977}, pages 46--57. {IEEE} Computer Society, 1977.
\newblock \href {https://doi.org/10.1109/SFCS.1977.32} {\path{doi:10.1109/SFCS.1977.32}}.

\bibitem[PP04]{PP04}
Dominique Perrin and Jean{-}Eric Pin.
\newblock {\em Infinite words -- automata, semigroups, logic and games}, volume 141 of {\em Pure and applied mathematics series}.
\newblock Elsevier Morgan Kaufmann, 2004.

\bibitem[PR89]{PR89Synthesis}
Amir Pnueli and Roni Rosner.
\newblock On the synthesis of a reactive module.
\newblock In {\em Proceedings of the 16th Annual {ACM} Symposium on Principles of Programming Languages, POPL 1989, Austin, {TX}, {USA}, January 11--13, 1989}, pages 179--190. {ACM} Press, 1989.
\newblock \href {https://doi.org/10.1145/75277.75293} {\path{doi:10.1145/75277.75293}}.

\bibitem[Put94]{Put94}
Martin~L. Puterman.
\newblock {\em {M}arkov Decision Processes: Discrete Stochastic Dynamic Programming}.
\newblock Wiley Series in Probability and Statistics. Wiley, 1994.
\newblock \href {https://doi.org/10.1002/9780470316887} {\path{doi:10.1002/9780470316887}}.

\bibitem[RS59]{RS59}
Michael~O. Rabin and Dana~S. Scott.
\newblock Finite automata and their decision problems.
\newblock {\em {IBM} Journal of Research and Development}, 3(2):114--125, 1959.
\newblock \href {https://doi.org/10.1147/rd.32.0114} {\path{doi:10.1147/rd.32.0114}}.

\bibitem[Sha53]{Sha53}
Lloyd~S. Shapley.
\newblock Stochastic games.
\newblock {\em Proceedings of the National Academy of Sciences}, 39(10):1095--1100, 1953.
\newblock \href {https://doi.org/10.1073/pnas.39.10.1095} {\path{doi:10.1073/pnas.39.10.1095}}.

\bibitem[Sta83]{Sta83}
Ludwig Staiger.
\newblock Finite-state $\omega$-languages.
\newblock {\em Journal of Computer and System Sciences}, 27(3):434--448, 1983.
\newblock \href {https://doi.org/10.1016/0022-0000(83)90051-X} {\path{doi:10.1016/0022-0000(83)90051-X}}.

\bibitem[Tar72]{Tar72}
Robert~E. Tarjan.
\newblock Depth-first search and linear graph algorithms.
\newblock {\em SIAM Journal on Computing}, 1(2):146--160, 1972.
\newblock \href {https://doi.org/10.1137/0201010} {\path{doi:10.1137/0201010}}.

\bibitem[Tho08]{Tho08}
Wolfgang Thomas.
\newblock {C}hurch's problem and a tour through automata theory.
\newblock In Arnon Avron, Nachum Dershowitz, and Alexander Rabinovich, editors, {\em Pillars of Computer Science, Essays Dedicated to Boris (Boaz) Trakhtenbrot on the Occasion of His 85th Birthday}, volume 4800 of {\em Lecture Notes in Computer Science}, pages 635--655. Springer, 2008.
\newblock \href {https://doi.org/10.1007/978-3-540-78127-1_35} {\path{doi:10.1007/978-3-540-78127-1_35}}.

\bibitem[Wag79]{Wagner1979omega}
Klaus~W. Wagner.
\newblock On $\omega$-regular sets.
\newblock {\em Information and Control}, 43(2):123--177, 1979.
\newblock \href {https://doi.org/10.1016/S0019-9958(79)90653-3} {\path{doi:10.1016/S0019-9958(79)90653-3}}.

\bibitem[Zie98]{Zie98}
Wies{\l}aw Zielonka.
\newblock Infinite games on finitely coloured graphs with applications to automata on infinite trees.
\newblock {\em Theoretical Computer Science}, 200(1-2):135--183, 1998.
\newblock \href {https://doi.org/10.1016/S0304-3975(98)00009-7} {\path{doi:10.1016/S0304-3975(98)00009-7}}.

\end{thebibliography}

	\appendix
	\newpage

	\newcommand{\worseSet}{\ensuremath{\sqsubseteq_\wc}}
	\newcommand{\worseStrict}{\ensuremath{\sqsubset_\wc}}

	\section{Relation with other properties from the literature} \label{app:relation}
	We show precise links between the first two conditions we have defined in Section~\ref{sec:three-conditions} and the \emph{(strong) monotony} and \emph{(strong) selectivity} notions from~\cite{GZ05,BFMM11}.

	Let $\wc\subseteq\colors^\omega$ be an objective.
	For $M \subseteq \colors^*$ a set of \emph{finite} words, we write $[M]$ for the set of infinite words whose every finite prefix is a prefix of a finite word in $M$.
	For $\word\in\colors^*$, we write $\word\wc$ for set of infinite words $\word\word'$ with $\word'\in\wc$.
	For two sets of infinite words $M, N \subseteq \colors^\omega$, we write $M \worseSet N$ if $M \cap \wc \neq \emptyset$ implies $N \cap \wc \neq \emptyset$ (in other words, either all words in $M$ are losing, or there exists a winning word in $N$).
	We write $M \worseStrict N$ if $M \cap \wc = \emptyset$ and $N \cap \wc \neq \emptyset$ (in other words, all words in $M$ are losing and there exists a winning word in $N$).

	\paragraph{Properties of prefixes.}
	Objective $\wc$ is \emph{strongly monotone} if for all languages of infinite words $M, N \subseteq \colors^\omega$, if there exists $\word\in\colors^*$ such that $\word M \worseStrict \word N$, then for all $\word'\in\colors^*$, $\word' M \worseSet \word' N$.
	Objective $\wc$ is \emph{monotone} if for all \emph{regular} languages (of finite words) $M, N \subseteq \colors^*$, if there exists $\word\in\colors^*$ such that $\word [M] \worseStrict \word [N]$, then for all $\word'\in\colors^*$, $\word' [M] \worseSet \word' [N]$.

	We show that having a total prefix preorder is in general equivalent to strong monotony as defined in~\cite{BFMM11}, and is even equivalent to monotony as defined in~\cite{GZ05} for $\omega$-regular objectives.
	In particular, strong monotony and monotony coincide for $\omega$-regular objectives.
	\begin{lem}
		Let $\wc\subseteq \colors^\omega$ be an objective.
		Objective $\wc$ has a total prefix preorder if and only if $\wc$ is strongly monotone.
		If $\wc$ is $\omega$-regular, $\wc$ has a total prefix preorder if and only if $\wc$ is monotone.
	\end{lem}
	\begin{proof}
		We first prove the first equivalence, with no restriction on $\wc$.

		We assume that $\wc$ has a total prefix preorder and we show that it is strongly monotone.
		Let $\word\in\colors^*$, $M, N\subseteq \colors^\omega$ such that $\word M \worseStrict \word N$.
		Let $\word'\in\colors^*$.
		As the prefix-preorder is total, we have that $\word \prefOrd \word'$ or $\word' \prefOrd \word$.
		If $\word \prefOrd \word'$, as $\word$ has a winning continuation in $N$, this continuation is also winning for $\word'$ --- so $\word' M \worseSet \word' N$.
		If $\word' \prefOrd \word$, then all the infinite words in $\word' M$ are also losing, thus $\word' M \worseSet \word' N$.

		We assume that $\wc$ is strongly monotone and we show that the prefix preorder is total.
		Let $x, y\in\colors^*$.
		We show that $x$ and $y$ are comparable for $\prefOrd$: we assume w.l.o.g.\ that $x \not\prefOrd y$, and we show that $y \prefOrd x$, i.e., that $\inverse{y}\wc \subseteq \inverse{x}\wc$.
		Let $y' \in \inverse{y}\wc$.
		As $x \not\prefOrd y$, there is $x'\in\colors^\omega$ such that $xx' \in \wc$ but $yx' \notin \wc$.
		By taking $M = \{x'\}$ and $N = \{y'\}$, we have $yM \worseStrict yN$.
		Hence, $xM \worseSet xN$ by strong monotony.
		As $xM$ contains a winning word, so does $xN$, so $y'\in\inverse{x}\wc$.

		We now assume that $\wc$ is $\omega$-regular.
		The left-to-right implication still holds, as monotony is a weaker notion than strong monotony.
		We reprove the right-to-left implication with this extra assumption.
		We assume that $\wc$ is monotone and we show that the prefix preorder is total.
		Let $x, y\in\colors^*$.
		We show that $x$ and $y$ are comparable for $\prefOrd$: we assume w.l.o.g.\ that $x \not\prefOrd y$, and we show that $y \prefOrd x$, i.e., that $\inverse{y}\wc \subseteq \inverse{x}\wc$.
		As $\inverse{x}\wc$ and $\inverse{y}\wc$ are also $\omega$-regular, by Lemma~\ref{lem:ultPer}, it suffices to show that all ultimately periodic words in $\inverse{y}\wc$ are in $\inverse{x}\wc$.
		Let $u\in\colors^*, v\in\colors^+$ such that $uv^\omega \in \inverse{y}\wc$.
		As $x \not\prefOrd y$, once again by Lemma~\ref{lem:ultPer}, there is $u'\in\colors^*, v'\in\colors^+$ such that $xu'(v')^\omega \in \wc$ but $yu'(v')^\omega \notin \wc$.
		By taking $M = u'(v')^*$ and $N = uv^*$ (which are regular languages), we have $[M] = \{u'(v')^\omega\}$ and $[N] = \{uv^\omega\}$.
		Therefore, $y[M] \worseStrict y[N]$.
		Hence, $xM \worseSet xN$ by monotony.
		As $xM$ contains a winning word, so does $xN$, so $uv^\omega\in\inverse{x}\wc$.
	\end{proof}
	As having a total prefix preorder holds symmetrically for an objective and its complement, we deduce that so does strong monotony, and so does monotony for $\omega$-regular objectives.
	The latter is especially interesting for the study of $\omega$-regular objectives, as monotony is not symmetric in general (there are objectives that are monotone but whose complement is not).

	\paragraph{Properties of cycles.}
	Objective $\wc$ is \emph{strongly selective} if for all languages of finite words $M, N, K \subseteq \colors^*$, for all $\word\in\colors^*$, $\word[(M \cup N)^*K] \worseSet \word[M^*] \cup \word[N^*] \cup \word[K]$.
	Selectivity is the same definition with $M, N, K$ being restricted to being \emph{regular} languages.

	As for having a total prefix preorder and (strong) monotony, we can link progress-consistency and (strong) selectivity.
	We have that strong selectivity implies progress-consistency, and that selectivity implies progress-consistency when $\wc$ is $\omega$-regular.

	\begin{lem}
		Let $\wc\subseteq \colors^\omega$ be an objective.
		If $\wc$ is strongly selective, then it is progress-consistent.
		If $\wc$ is $\omega$-regular and selective, then it is progress-consistent.
	\end{lem}
	\begin{proof}
		We prove the first implication by showing the contrapositive.
		We assume that $\wc$ is not progress-consistent, i.e., there exist $\word_1\in\colors^*$ and $\word_2\in\colors^+$ such that $\word_1 \strictPrefOrd \word_1\word_2$ and $\word_1(\word_2)^\omega \notin \wc$.
		Let $\word_3\in\colors^\omega$ be such that $\word_1\word_3\notin\wc$ and $\word_1\word_2\word_3\in\wc$, which exists as $\word_1 \strictPrefOrd \word_1\word_2$.
		Let $\word = \word_1$, $M = \{\word_2\}$, $N = \emptyset$, and $K = \{\word_3'\in\colors^* \mid \word_3'\ \text{is a prefix of}\ \word_3\}$.
		We have that $\word[(M \cup N)^*K]$ contains the winning word $\word_1\word_2\word_3$.
		However, we have $\word[M^*] = \{\word_1(\word_2)^\omega\}$, $\word[N^*] = \emptyset$, and $\word[K] = \{\word_1\word_3\}$: the set $\word[M^*] \cup \word[N^*] \cup \word[K]$ contains only losing words.
		We do not have $\word[(M \cup N)^*K] \worseSet \word[M^*] \cup \word[N^*] \cup \word[K]$.

		For the second implication, we assume that $\wc$ is $\omega$-regular.
		The proof proceeds in the exact same way, except that we have to make sure that $M$, $N$, and $K$ are regular.
		This is the case for $M$ and $N$.
		To make $K$ regular, we use Lemma~\ref{lem:ultPer} and obtain that there is an ultimately periodic continuation $uv^\omega$ witnessing that $\word_1 \strictPrefOrd \word_1\word_2$.
		We can set $K = uv^*$ (which is regular) and the proof ends in the same way.
	\end{proof}

\end{document}